\documentclass[12 pt]{article}
\usepackage{amsthm,natbib,url,titlefoot,float,placeins,xr,nolbreaks,alltt,newtxtext,newtxmath,amsfonts,bm,lscape,graphicx,mathtools,multirow}

\urldef\chuanmail\path{gohc@uoguelph.ca}
\urldef\chuanweb\path{http://www.chuangoh.org}
\allowdisplaybreaks
\begin{document}

\title{Rate-Optimal Estimation of the Intercept in a Semiparametric Sample-Selection Model\thanks{The author is grateful to 
Shakeeb Khan for several helpful conversations and to seminar participants at Louisiana State University and The 
University of Wisconsin-Madison for their comments.  The usual disclaimer applies.  The author also thanks Daniel 
Siercks for computing support.}} 	
\author{Chuan Goh\thanks{Department of Economics and Finance, University of Guelph, 50 Stone Road East, Guelph, 
ON, Canada N1G 2W1.  \chuanmail, \chuanweb.}}
\date{September 2018}
\maketitle

\newtheorem{definition}{Definition}\newtheorem{hypothesis}{Hypothesis}\newtheorem{lemma}{Lemma}\newtheorem{proposition}{Proposition}\newtheorem{theorem}{Theorem}\newtheorem{corollary}{Corollary}
\newtheorem{assumption}{Assumption}\newtheorem{remark}{Remark}\newtheorem{example}{Example}\newtheorem{condition}{Condition}

\begin{abstract}
This paper presents a new estimator of the intercept of a linear regression model in cases where the outcome 
varaible is observed subject to a selection rule.  The intercept is often in this context of inherent interest; for 
example, in a program evaluation context, the difference between the intercepts in outcome equations for 
participants and non-participants can be interpreted as the difference in average outcomes of participants and 
their counterfactual average outcomes if they had chosen not to participate.  The new 
estimator can under mild conditions exhibit a rate of convergence in probability equal to $n^{-p/(2p+1)}$, 
where $p\ge 2$ is an integer that indexes the strength of certain smoothness assumptions.  This rate of 
convergence is shown in this context to be the optimal rate of convergence for estimation of the 
intercept parameter in terms of a minimax criterion.  The new estimator, unlike other proposals in the 
literature, is under mild conditions consistent and asymptotically normal with a rate of convergence that is the same regardless of the degree to which selection depends on 
unobservables in the outcome equation.  Simulation evidence and an empirical example are included.
\end{abstract}

\noindent JEL {\it Classification:}  C14, C31, C34

\clearpage
\section{Introduction}

\noindent This paper considers a sample-selection model \citep[e.g.,][]{Heckman79} given by
\begin{eqnarray}
Y^* &=&\theta_0+\bm{X}^{\top}\bm{\beta}_0+U,\label{outcome}\\
D &=& 1\left\{\bm{Z}^{\top}\bm{\gamma}_0\ge V\right\},\label{selection}\\
Y &=& DY^*,\label{observation}
\end{eqnarray}
where $[\begin{array}{cccc} D & \bm{X}^{\top} & \bm{Z}^{\top} & Y\end{array}]$ is an observed random vector, 
where 
$[\begin{array}{cc} U & V \end{array}]$ is an unobserved random vector such that $E\left[U^2\right]<\infty$ and 
$E[U|\bm{X}]=0$ and $E[U|\bm{Z}]=E[U]$ almost surely.  Equations~(\ref{outcome}) and 
(\ref{selection}) are typically referred to as the {\it outcome} and {\it selection} equations, respectively.  
Variants of the model given in (\ref{outcome})--(\ref{observation}) have been considered at least since the 
contributions of \citet{Gronau74,Heckman74} and \citet{Lewis74}.  These authors were primarily concerned with the 
issue of selectivity bias in empirical analyses of individual labour-force participation decisions, particularly 
for women.  This bias, which is present in least-squares estimates of the parameter 
$[\begin{array}{cc} \theta_0 & \bm{\beta}^{\top}_0\end{array}]$ appearing above in (\ref{outcome}), arises from 
the assumption that the observed wages of workers are affected by the self selection of 
those workers into the workforce.  In particular, one can only observe a wage that exceeds the reservation wage 
of the individual in question.  In terms of the model given in (\ref{outcome})--(\ref{observation}), the wage 
offer is represented by the variable $Y^*$, while the observed wage is denoted by $Y$.  The employment status of 
the individual in question is denoted by the binary variable $D$, which takes a value of one if the unobserved 
difference $\bm{Z}^{\top}\bm{\gamma}_0-V$ between the wage offer and the reservation wage is positive; $D$ is 
otherwise equal to zero.  Observed variables influencing individual participation decisions are collected in the 
vector $\bm{Z}$, while observed determinants of individual wage offers are collected in $\bm{X}$.  Variants of the 
sample-selection model appearing in (\ref{outcome})--(\ref{observation}) have been found useful for a wide variety of 
applied problems apart from the analysis of indidivual labour supply decisions \citep[e.g.,][]{Vella98}.  Relatively 
recent economic applications include those of \citet{HelpmanMelitzRubinstein08}, \citet{MulliganRubinstein08} and 
\citet{JimenezOngenaPeydroSaurina14}.

This paper focuses on statistical inference regarding the intercept $\theta_0$ appearing above in (\ref{outcome}).  
The intercept in the outcome equation is often of inherent interest in various applications of the model given by 
(\ref{outcome})--(\ref{observation}).  For example, suppose that (\ref{selection}) accurately describes the 
selection of individuals into some treatment group.  In this case, the difference between the intercepts in the 
outcome equations for treated and non-treated individuals may be interpreted as the causal effect of treatment 
when selection to treatment is mean independent of the unobservable $U$ in the outcome equation \citep[e.g.,][p. 500]{AndrewsSchafgans98}.  The intercept in the outcome 
equation is similarly crucial in computing the average wage difference in problems where the sample-selection 
model given above is applied to the analysis of wage differences between workers in two different socioeconomic 
groups, or between unionized and non-unionized workers 
\citep[e.g.,][]{Oaxaca73,Lewis86,Heckman90}.  Finally, the intercept in the outcome equation permits the 
evaluation of the net benefits of a social program in terms of the differences between the observed 
outcomes of participants and their counterfactual expected outcomes had they chosen not to participate 
\citep[e.g.,][]{HeckmanRobb85}.  

Early applied work generally proceeded from the assumption that the unobservables 
$[\begin{array}{cc} U & V\end{array}]$ appearing above in (\ref{outcome})--(\ref{selection}) are bivariate normal 
mean-zero with an unknown covariance matrix and independent of 
$[\begin{array}{cc} \bm{X}^{\top} & \bm{Z}^{\top}\end{array}]$.  This assumption in turn allowed for the 
estimation of the parameters appearing in (\ref{outcome})--(\ref{selection}) via the method of maximum likelihood 
or the related two-step procedure of \citet{Heckman74,Heckman76}.  These estimates, however, are generally 
inconsistent under departures from the assumed bivariate normality of $[\begin{array}{cc} U & V\end{array}]$ 
\citep[e.g.,][]{ArabmazarSchmidt82,Goldberger83,Schafgans04}.  The desirability of not imposing a parametric 
specification on the joint distribution of the unobservables in the outcome and selection equations has led in turn 
to the development of distribution-free methods of estimating the parameters appearing above in 
(\ref{outcome})--(\ref{selection}).  Distribution-free methods of estimating the intercept $\theta_0$ in 
(\ref{outcome}) include the proposals of \citet{GallantNychka87,Heckman90,AndrewsSchafgans98} and \citet{Lewbel07}.  

Estimators of the intercept of the outcome equation implemented by distribution-free procedures have 
large-sample behaviours that vary depending on the extent of endogeneity in the selection mechanism, 
i.e., on the nature and extent of any dependence between the random variables $U$ and $V$ appearing above in (\ref{outcome}) 
and (\ref{selection}), respectively.  Given that these features of the joint distribution of 
$[\begin{array}{cc} U & V \end{array}]$ are typically unknown 
in empirical practice, the dependence of the asymptotic behaviour of intercept estimators on these features potentially 
complicates statistical inference regarding $\theta_0$.  This issue is easily and starkly illustrated in the case of 
ordinary least squares (OLS).  In particular, suppose that selection in the model is based strictly on observables, which 
is equivalent to assuming that the unobservable $U$ appearing in (\ref{outcome}) and the selection indicator $D$ are conditionally 
mean independent given $\bm{X}$ and $\bm{Z}$, i.e., $P\left[E[U|D=1,\bm{X},\bm{Z}]=E[U|D=0,\bm{X},\bm{Z}]\right]=1$.  In this case $\theta_0$ can be 
consistently estimated at a parametric rate with no additional assumptions 
imposed on the joint distribution of $[\begin{array}{ccc} U & V & \bm{Z}^{\top}\bm{\gamma}_0\end{array}]$ by applying OLS 
to the outcome equation using only those observations for which $D=1$.  On the other hand, the OLS estimate of $\theta_0$ 
obtained in this way is inconsistent if the difference $1-P\left[E[U|D=1,\bm{X},\bm{Z}]=E[U|D=0,\bm{X},\bm{Z}]\right]$ 
is positive, even if arbitrarily small.  It follows that OLS generates inferences regarding $\theta_0$ that vary drastically with respect 
to the degree to which $E[U|D=1,\bm{X},\bm{Z}]$ may differ from $E[U|D=0,\bm{X},\bm{Z}]$.

This paper develops a distribution-free estimator of the intercept $\theta_0$ in the outcome equation that is 
consistent and asymptotically normal with a rate of convergence that is 
the same regardless of the joint distribution of $[\begin{array}{ccc} U & V & \bm{Z}^{\top}\bm{\gamma}_0\end{array}]$.  I show that there exists an 
implementation of the proposed estimator of $\theta_0$ that converges uniformly at the rate $n^{-p/(2p+1)}$, where $n$ denotes the sample size and where $p\ge 2$ is an integer 
that indexes the strength of certain smoothness assumptions described below.  The uniformity of this convergence involves uniformity over a class of joint distributions of 
$[\begin{array}{ccc} U & V & \bm{Z}^{\top}\bm{\gamma}_0\end{array}]$ satisfying necessary conditions for the identification of $\theta_0$.   In other words, the estimator developed below is {\it adaptive} to the nature of 
selection in the model.  

This paper also shows that the uniform $n^{-p/(2p+1)}$-rate attainable by the proposed estimator is in fact 
the optimal rate of convergence of an estimator of $\theta_0$ in terms of a minimax criterion.  It follows that 
the proposed estimator may be implemented in such a way as to converge in probability to 
$\theta_0$ at the fastest possible minimax rate.
   
The estimator developed below differs from earlier proposals of \citet{Heckman90} and \citet{AndrewsSchafgans98} 
that involve the {\it rate-adaptive} estimation of the intercept in the outcome equation.  These proposals also 
involve estimators that are consistent and asymptotically normal regardless of the extent to which selection is 
endogenous, but converge to the limiting normal distribution at unknown rates; in particular, see \citet[Theorems~1--2]{SchafgansZindeWalsh02} for the 
estimator of \citet{Heckman90} and \citet[Theorems~2, 3, 5 and 5*]{AndrewsSchafgans98}.  

The estimator developed below also differs from estimators of $\theta_0$ that take the form of averages weighted by the 
reciprocal of an estimate of the density of the selection index $\bm{Z}^{\top}\bm{\gamma}_0$ appearing above in (\ref{selection}).  Such 
estimators \citep[e.g.,][]{Lewbel07}, in common with the estimators of \citet{Heckman90} and 
\citet{AndrewsSchafgans98}, are known to converge generically at unknown rates.  In addition, estimators 
taking the form of inverse density-weighted averages may have sampling distributions that are not even 
asymptotically normal \citep{KhanTamer10,KhanNekipelov13,ChaudhuriHill16}.  In general, the estimators of 
\citet{Heckman90} and \citet{AndrewsSchafgans98} and of estimators in the class of 
inverse density-weighted averages converge at rates that depend critically on conditions involving the relative tail 
thicknesses of the distributions of the selection index $\bm{Z}^{\top}\bm{\gamma}_0$ and of the latent selection 
variable $V$ appearing in (\ref{selection}) \citep{KhanTamer10}.  These conditions may be difficult to verify in 
applications.  The estimator developed below converges by contrast at a known rate under conditions implied by the 
identification of $\theta_0$ to a normal distribution uniformly over the underlying parameter space regardless of the 
relative tail behaviours of $\bm{Z}^{\top}\bm{\gamma}_0$ and $V$.  This facilitates statistical inference regarding $\theta_0$.

The remainder of this paper proceeds as follows.  The following section discusses identification of the 
intercept parameter in the outcome equation and presents the new estimator along with its first-order asymptotic 
properties.  Section~\ref{opt} derives the minimax rate optimality of the new estimator.  Section~\ref{mc} 
presents the results of simulation experiments that investigate the behaviour in finite samples of the proposed 
estimator in relation to other methods.  Section~\ref{ee} considers an application of the new 
estimator to the analysis of gender wage gaps in Malaysia.  Section~\ref{concl} concludes.  Proofs of 
all theoretical results are collected in the appendix.

\section{The New Estimator}\label{est}

\noindent This section presents the new estimator of the intercept $\theta_0$ appearing in (\ref{outcome}) 
and describes its asymptotic behaviour to first order.  Let $\hat{\bm{\beta}}_n$ and 
$\hat{\bm{\gamma}}_n$ denote $\sqrt{n}$-consistent estimators of the parameters $\bm{\beta}_0$ and 
$\bm{\gamma}_0$ appearing above in (\ref{outcome}) and (\ref{selection}), respectively, where $\bm{\gamma}_0$ is 
assumed to be identified up to a location and scale normalization.  The existence of such 
estimators has long been established; see e.g., the proposals of \citet{Han87,Robinson88,PowellStockStoker89,Andrews91,IchimuraLee91,Ichimura93,KleinSpady93,Powell01} or 
\citet{Newey09}.  In addition, suppose that
$\left\{[\begin{array}{cccc} D_i & \bm{X}_i^{\top} & Y_i & \bm{Z}_i^{\top} \end{array}]:\, i=1,\ldots,n\right\}$ are iid 
copies of the random vector $[\begin{array}{cccc} D & \bm{X}^{\top} & Y & \bm{Z}^{\top} \end{array}]$.  Let 
$\bm{z}\in\mathbb{R}^{l}$ be an arbitrary vector, and define 
\begin{equation}\label{etahatn}
\hat{\eta}_n(\bm{z})\equiv \frac{1}{n}\sum_{i=1}^n 1\left\{(\bm{Z}_i-\bm{z})^{\top}\hat{\bm{\gamma}}_n\le 0\right\}
\end{equation}

Let 
\begin{equation}\label{what}
\hat{W}_i\equiv D_i\left(Y_i-\bm{X}^{\top}_i\hat{\bm{\beta}}_n\right)
\end{equation}
for each $i\in\{1,\ldots,n\}$.  This paper proposes to estimate $\theta_0$ via a locally linear smoother of the form
\begin{equation}\label{mhat}
\hat{\theta}_n\equiv \bm{e}_1^{\top}\left(\sum_{i=1}^n \bm{S}_i K_i\bm{S}_i^{\top}\right)^{-1}\sum_{i=1}^n\bm{S}_i K_i\hat{W}_i,
\end{equation}
where $\bm{e}_1=[\begin{array}{cc} 1 & 0\end{array}]^{\top}$ and where 
\begin{equation}\label{si}
\bm{S}_i=[\begin{array}{cc} 1 & \hat{\eta}_n(\bm{Z}_i)-1\end{array}]^{\top}
\end{equation} 
and 
\begin{equation}\label{ki}
K_i= K\left(h_n^{-1}\left(\hat{\eta}_n(\bm{Z}_i)-1\right)\right)
\end{equation} 
for $i=1,\ldots,n$.  The quantity $h_n$ appearing in each 
$K_i$ denotes a bandwidth $h_n>0$ such that $h_n\to 0$ with $nh_n^{3}\to\infty$ as $n\to\infty$, while for some 
$p\ge 2$, $K(\cdot)$ denotes a smoothing kernel of order $p$, i.e., one where 
$\int_{-\infty}^{\infty} K(u) du=1$, 
$\int_{-\infty}^{\infty} u^r K(u)=0$ for all $r\in \{1,\ldots,p-1\}$ and 
$\int_{-\infty}^{\infty} u^{p} K(u)<\infty$.

Assume that the disturbance $U$ in the outcome equation (\ref{outcome}) satisfies $E[|U|]<\infty$ with 
$P\left[E[U|\bm{Z}]=0\right]=1$.  In addition, let the selection index $\bm{Z}^{\top}\bm{\gamma}_0$ be 
distributed with distribution $F_0$, assumed to be absolutely continuous.    The estimator of $\theta_0$ given in (\ref{mhat}) exploits 
the fact that identification of $\theta_0$ occurs ``at infinity'' \citep{Chamberlain86}, or in any case depends crucially on the 
selection index $\bm{Z}^{\top}\bm{\gamma}_0$ being able to take values sufficiently large so that the corresponding 
conditional probabilities of selection take values close to one.  In particular, $\theta_0$ is characterized by the equalities
\begin{eqnarray}
 & & E\left[\left.D\left(Y-\bm{X}^{\top}\bm{\beta}_0\right)\right|F_0\left(\bm{Z}^{\top}\bm{\gamma}_0\right)=1\right]\nonumber\\
 &=& E\left[\left.1\left\{F_0(V)\le 1\right\}\left(\theta_0+U\right)\right|F_0\left(\bm{Z}^{\top}\bm{\gamma}_0\right)=1\right]\nonumber\\
 &=&\theta_0\label{theta0_id}
\end{eqnarray} 
The proposed estimator of $\theta_0$ exploits the representation of the estimand in (\ref{theta0_id}), which suggests the estimation of $\theta_0$ by direct 
estimation of the quantity 
$E\left[\left.D\left(Y-\bm{X}^{\top}\bm{\beta}_0\right)\right|F_0\left(\bm{Z}^{\top}\bm{\gamma}_0\right)=1\right]$.

One can view the estimator given in (\ref{mhat}) as an extension of the Yang--Stute symmetrized nearest-neighbours 
(SNN) estimator of a conditional mean \citep{Yang81,Stute84} to the problem of estimating the intercept 
in the outcome equation (\ref{outcome}).  SNN estimators are characterized by asymptotic behaviours that 
are asymptotically ``design adaptive'' in the sense that their asymptotic normality can generally be established without technical conditions on the 
probability of the design variable taking values in regions of low density \citep{Stute84}.  In the present context, it is shown under certain conditions 
that the estimator $\hat{\theta}_n$ in (\ref{mhat}) is asymptotically normal with a 
rate of convergence that depends neither on the extent to which 
the unobservables $U$ and $V$ are dependent, nor on the relationship between the behaviours of the selection index
$\bm{Z}^{\top}\bm{\gamma}_0$ and the unobservable $V$ in the right tails of their respective marginal distributions.  
This is essentially accomplished by transforming the estimated indices $\bm{Z}^{\top}_i\hat{\bm{\gamma}}_n$ into random 
variables $\hat{\eta}_n(\bm{Z}_i)$ that are approximately uniformly distributed on $[0,1]$.  

It is worth noting that estimators of $\theta_0$ that take the form of inverse density-weighted averages 
\citep[e.g.,][]{Lewbel07} have rates of convergence that generally vary with the extent to which $U$ and $V$ are 
dependent, as well as with the relative right-tail behaviours of $\bm{Z}^{\top}\bm{\gamma}_0$ and $V$ 
\citep{KhanTamer10}.  These estimators rely on an alternative representation of (\ref{theta0_id}) and are consistent 
under additional regularity conditions.  In particular, suppose 
$m_0\left(\bm{z}^{\top}\bm{\gamma}_0\right)\equiv E\left[\left. D\left(Y-\bm{X}^{\top}\bm{\beta}_0\right)\right|\bm{Z}^{\top}\bm{\gamma}_0=\bm{z}^{\top}\bm{\gamma}_0\right]$ 
is everywhere differentiable in $\bm{z}^{\top}\bm{\gamma}_0$ with derivative given by 
$m_0^{(1)}\left(\bm{z}^{\top}\bm{\gamma}_0\right)\equiv \partial m_0\left(\bm{z}^{\top}\bm{\gamma}_0\right)/\partial\bm{z}^{\top}\bm{\gamma}_0$.  
Suppose in addition that the distribution of $\bm{Z}^{\top}\bm{\gamma}_0$ is absolutely continuous with density 
$f_0(\cdot)$ such that 
$E\left[\left|m_0^{(1)}\left(\bm{Z}^{\top}\bm{\gamma}_0\right)/f_0\left(\bm{Z}^{\top}\bm{\gamma}_0\right)\right|\right]<\infty$.  One 
can then write
\begin{eqnarray}
\theta_0 &=& \lim_{F_0(\bm{z}^{\top}\bm{\gamma}_0)\to 1} E\left[\left. D\left(Y-\bm{X}^{\top}\bm{\beta}_0\right)\right|\bm{Z}=\bm{z}\right]\nonumber\\
	&=&\int_{-\infty}^{\infty} m^{(1)}_0\left(\bm{z}^{\top}\bm{\gamma}_0\right) d\bm{z}^{\top}\bm{\gamma}_0\nonumber\\
	&=&\int_{-\infty}^{\infty} \frac{m_0^{(1)}\left(\bm{z}^{\top}\bm{\gamma}_0\right)}{f_0\left(\bm{z}^{\top}\bm{\gamma}_0\right)}\cdot f_0\left(\bm{z}^{\top}\bm{\gamma}_0\right)d\bm{z}^{\top}\bm{\gamma}_0\nonumber\\
	&=& E\left[\frac{m_0^{(1)}\left(\bm{Z}^{\top}\bm{\gamma}_0\right)}{f_0\left(\bm{Z}^{\top}\bm{\gamma}_0\right)}\right],\label{id5}
\end{eqnarray}
which suggests estimating $\theta_0$ via its approximate sample analogue in which $m_0^{(1)}(\cdot)$ and 
$f_0(\cdot)$ are replaced by suitable preliminary estimates, and in which the systematic trimming of observations 
corresponding to small values of $f_0\left(\bm{Z}^{\top}_i\bm{\gamma}_0\right)$ may be required.  This approach 
to estimating $\theta_0$ follows that proposed by \citet{Lewbel97} for estimating a binary choice model arising 
from a latent linear model in which a mean restriction is imposed on the latent error term, and is the 
approach to estimating $\theta_0$ considered in more recent work by \citet{Lewbel07,KhanTamer10} and 
\citet{KhanNekipelov13}.  Consistent estimators of $\theta_0$ that exploit (\ref{id5}) in this way naturally 
depend critically on the assumed finiteness of 
$E\left[\left|m_0^{(1)}\left(\bm{Z}^{\top}\bm{\gamma}_0\right)/f_0\left(\bm{Z}^{\top}\bm{\gamma}_0\right)\right|\right]$, 
an assumption that in turn leads the rates at which they converge to their limiting distributions to depend on the 
relative tail behaviours of the variables in the selection equation or the extent to which selection is endogenous 
\citep{KhanTamer10}.  The non-uniformity in the rate of convergence as one varies the relative tail behaviours of the 
determinants of selection or the dependence between the disturbance terms $U$ and $V$ is a feature that is also shared by 
estimators of $\theta_0$ that involve locally constant or polynomial regressions of $\hat{W}_i$ on the 
untransformed estimated selection indices 
$\bm{Z}^{\top}_i\hat{\bm{\gamma}}_n$ \citep[e.g.,][]{Heckman90,AndrewsSchafgans98}.   Results of this 
nature significantly complicate the task of statistical inference regarding $\theta_0$.  By way of contrast, the 
transformations to $\hat{\eta}_n(\bm{Z}_i)$ of the selection indices used in the locally linear 
SNN estimator given in (\ref{mhat}) permit the locally linear SNN estimator to enjoy asymptotic normality with a rate 
of convergence that varies neither with the endogeneity of selection nor with the relative tail behaviours of the 
variables appearing in the selection equation.

It should also be noted that the result given above in (\ref{theta0_id}), in which the estimand is identified as
$\theta_0=E\left[\left.D\left(Y-\bm{X}^{\top}\bm{\beta}_0\right)\right|F_0(\bm{Z}^{\top}\bm{\gamma}_0)=1\right]$, 
motivates the formulation of the estimator in (\ref{mhat}) as the intercept in a locally linear regression.  
One could as easily in this context use (\ref{theta0_id}) to motivate an estimator of $\theta_0$ as the corresponding 
variant of a Nadaraya--Watson (i.e., locally constant regression) estimator; see in particular the approach taken 
in \citet{StuteZhu05}.  The focus on a locally linear regression estimator of $\theta_0$ is purely to improve the 
rate at which the bias of the proposed estimator vanishes in large samples, given that locally linear regression estimators 
have biases that converge at the same rate regardless of whether the conditioning variable is evaluated at an interior or at a limit point of 
its support \citep[e.g.,][]{FanGijbels92}.  Nadaraya--Watson estimators, on the other hand, have biases that 
converge relatively slowly when the conditioning variable is evaluated at a limit point of its support.

Assumptions underlying the first-order asymptotic behaviour of the estimator given by $\hat{\theta}_n$ in 
(\ref{mhat}) are given as follows.

\begin{assumption}\label{a1}
\begin{enumerate}
\item\label{a1a} 
\begin{enumerate}
\item $\bm{X}$ is $k$-variate, with support not contained in any proper linear subspace of 
$\mathbb{R}^k$;
\item $E\left[\left\|\bm{X}\right\|\right]<\infty$.
\end{enumerate} 
\item\label{a1b}
\begin{enumerate} 
\item $\bm{Z}$ is $l$-variate;
\item the support of $\bm{Z}$ is not contained in any proper linear subspace of $\mathbb{R}^{l}$,;
\item the first component of $\bm{\gamma}_0$ is equal to one;
\item $\bm{Z}$ does not contain a non-stochastic component;
\item\label{a1b4} the distribution $F_0$ of the selection index $\bm{Z}^{\top}\bm{\gamma}_0$ is absolutely 
continuous, with a density function $f_0$ that is differentiable on the support of 
$\bm{Z}^{\top}\bm{\gamma}_0$.
\end{enumerate}
\item\label{a1c} The set $\left\{(D_i,\bm{X}^{\top}_i, \bm{Z}^{\top}_i, Y_i):\,i=1,\ldots,n\right\}$ consists of 
independent observations each with the same distribution as the random vector 
$(D, \bm{X}^{\top}, \bm{Z}^{\top}, Y)$, which is generated according to the model given above in 
(\ref{outcome})--(\ref{observation}), and where 
\begin{eqnarray*}
E\left[U^2\right] &<& \infty;\\
P\left[E[U|\bm{X}]=0\right] &=& 1;\\
P\left[E[U|\bm{Z}] = E[U]\right] &=& 1.
\end{eqnarray*} 
\item\label{a1d} 
\begin{enumerate}
\item The joint conditional distribution given 
$[\begin{array}{cc} \bm{X}^{\top} & \bm{Z}^{\top}\end{array}]=[\begin{array}{cc} \bm{x}^{\top} & \bm{z}^{\top}\end{array}]$ of the 
disturbances $[\begin{array}{cc} U & V\end{array}]$ appearing in (\ref{outcome}) and (\ref{selection}) is absolutely continuous for all 
$[\begin{array}{cc} \bm{x}^{\top} & \bm{z}^{\top}\end{array}]$ in the support of $[\begin{array}{cc} \bm{X}^{\top} & \bm{Z}^{\top}\end{array}]$; 
moreover the corresponding joint density $g_{U,V|\bm{x},\bm{z}}(\cdot,\cdot)$ is continuously differentiable in both arguments almost everywhere on $\mathbb{R}^2$.
\item\label{a1d2} The conditional distribution given $[\begin{array}{cc} \bm{X}^{\top} & \bm{Z}^{\top}\end{array}]=[\begin{array}{cc} \bm{x}^{\top} & \bm{z}^{\top}\end{array}]$ 
of $V$ is absolutely continuous for all points $[\begin{array}{cc} \bm{x}^{\top} & \bm{z}^{\top}\end{array}]$ in the support of 
$[\begin{array}{cc} \bm{X}^{\top} & \bm{Z}^{\top}\end{array}]$; the corresponding density $g_{V|\bm{x},\bm{z}}(\cdot)$ is
differentiable almost everywhere on $\mathbb{R}$.
\end{enumerate}   
\end{enumerate}
\end{assumption} 

\begin{assumption}\label{a2}
\begin{enumerate}
\item\label{a2a} There exist estimators $\hat{\bm{\beta}}_n$ and $\hat{\bm{\gamma}}_n$ such that 
$\left\|\hat{\bm{\beta}}_n-\bm{\beta}_0\right\|=O_p\left(n^{-1/2}\right)$ 
and $\left\|\hat{\bm{\gamma}}_n-\bm{\gamma}_0\right\|=O_p\left(n^{-1/2}\right)$.  
\item\label{a2b} The smoothing kernel $K(\cdot)$ is bounded and twice continuously differentiable with 
$K(u)>0$ on $[0,1]$, $K(u)=0$ for all $u\not\in [0,1]$, with $\int K(u) du=1$ and where $\int K^2(u) du<\infty$.
\item\label{a2c} The bandwidth sequence $\left\{h_n\right\}$ satisfies $h_n>0$ with $h_n\to 0$ as $n\to\infty$, 
and $nh_n^3\to\infty$.
\item\label{a2d}  
\begin{enumerate}
\item \label{a2d1} There exists $p\ge 2$ such that $\int u^j K(u) du=0$ for all $j\in\{1,\ldots,p-1\}$ and $\int u^p K(u) du <\infty$.
\item\label{a2d11} The following hold for $p^*$ equal to the smallest odd integer greater than or equal to the constant $p+1$ specified 
in part~\ref{a2d1} of this 
assumption: 
\begin{enumerate}
\item\label{a2d2} The joint conditional density $g_{U,V|\bm{x},\bm{z}}(\cdot,\cdot)$ specified in Assumption~\ref{a1}.\ref{a1d} is $p^*$-times 
continuously differentiable in both arguments almost everywhere on $\mathbb{R}$ for all 
$[\begin{array}{cc} \bm{x}^{\top} & \bm{z}^{\top}\end{array}]$ in the support of $[\begin{array}{cc} \bm{X}^{\top} & \bm{Z}^{\top}\end{array}]$.
\item\label{a2d3} Similarly, the conditional density $g_{V|\bm{x},\bm{z}}(\cdot)$ specified in Assumption~\ref{a1}.\ref{a1d} is $p^*$-times 
continuously differentiable almost everywhere on $\mathbb{R}$ for all $[\begin{array}{cc} \bm{x}^{\top} & \bm{z}^{\top}\end{array}]$ in the support of 
$[\begin{array}{cc} \bm{X}^{\top} & \bm{Z}^{\top}\end{array}]$.
\end{enumerate}
\end{enumerate}
\end{enumerate}
\end{assumption}

\noindent The conditions of Assumption~\ref{a1} are largely standard and notably suffice for the selection parameter $\bm{\gamma}_0$ 
to be identified up to the particular location and scale normalization imposed by Assumption~\ref{a1}.\ref{a1b}.  
Assumption~\ref{a1} also does not restrict the components $V$ and $\bm{Z}^{\top}\bm{\gamma}_0$ of the selection equation to be 
independent. 

Assumption~\ref{a1} plays a crucial role in controlling the asymptotic bias of the proposed estimator $\hat{\theta}_n$.  In 
particular, identification of $\bm{\gamma}_0$ subject to the conditions of Assumption~\ref{a1}, along with the 
differentiability conditions of Assumption~\ref{a2}.\ref{a2d2}--\ref{a2}.\ref{a2d3}, imply a smoothness restriction on 
the conditional mean function
\begin{equation}\label{m0F}
m_{F_0}(q)\equiv E\left[D\left(Y-\bm{X}^{\top}\bm{\beta}_0\right)\left|F_0\left(\bm{Z}^{\top}\bm{\gamma}_0\right)=q\right.\right],
\end{equation}
This smoothness restriction takes the form of differentiability of $m_{F_0}(q)$ for $q\in (0,1)$ up to order no less than $p$, where $p\ge 2$ is the 
constant specified in Assumption~\ref{a2}.\ref{a2d}, along with finiteness of the left-hand limit of 
$\left(d^p/dq^p\right) E\left[D\left(Y-\bm{X}^{\top}\bm{\beta}_0\right)\left|F_0\left(\bm{Z}^{\top}\bm{\gamma}_0\right)=q\right.\right]$ 
at $q=1$.  This smoothness restriction, in other words, corresponds to a standard assumption in the literature on kernel estimation of 
conditional mean functions.  On the other hand, the differentiability to $p^*$-order in the second argument of 
$g_{U,V|\bm{x},\bm{z}}(\cdot,\cdot)$ is slightly stronger than the usual assumption of differentiability to 
order $p$.  This slight strengthening of the standard differentiability condition is used in the rate optimality arguments developed below in 
Section~\ref{opt}.  Details are contained in the proof of Theorem~\ref{ubthm} below. 

The smoothness restriction on $m_{F_0}(q)$ given in (\ref{m0F}) can also be seen to be implied by the identification of 
$\bm{\gamma}_0$ up to a location and scale normalization and by the smoothness conditions imposed in Assumption~\ref{a2}.\ref{a2d} on 
the conditional densities $g_{U,V|\bm{x},\bm{z}}(\cdot,\cdot)$ and $g_{V|\bm{x},\bm{z}}(\cdot,\cdot)$ for any 
$[\begin{array}{cc} \bm{x}^{\top} & \bm{z}^{\top}\end{array}]$ in the support of $[\begin{array}{cc} \bm{X}^{\top} & \bm{Z}^{\top}\end{array}]$.  In 
particular, one can show that under the conditions of Assumption~\ref{a1} and \ref{a2}.\ref{a2d}, $U$ has a conditional distribution given 
$F_0(V)=q$ for any $q\in [0,1]$ that is absolutely continuous with density
\begin{equation}\label{condl_dens}
r_{U|Q}(u|q)\equiv\frac{g_{UV}\left(u,F_0^{-1}(q)\right)}{g_V\left(F^{-1}_0(q)\right)},
\end{equation}
where $g_{UV}(\cdot,\cdot)$ and $g_V(\cdot)$ are respectively the joint density of $[\begin{array}{cc} U & V\end{array}]$ and the marginal density of 
$V$.  The conditional density $r_{U|Q}(u|q)$ is, given the absolute continuity of $F_0$ and the differentiability conditions on $g_{UV}$ and $g_V$ 
implied by Assumption~\ref{a2}.\ref{a2d}, $(p+1)$-times differentiable in $q$ on $(0,1)$ for any $u\in\mathbb{R}$.  The $(p+1)$-times 
differentiability of $r_{U|Q}(u|q)$ in $q$ on $(0,1)$ in turn implies the finiteness of 
$\left.\left(\partial^p/\partial q^p\right) r_{U|Q}(u|q)\right|_{q=1}$ for any $u\in\mathbb{R}$.  It is the finiteness of 
$\left.\left(\partial^p/\partial q^p\right) r_{U|Q}(u|q)\right|_{q=1}$ that implies the smoothness restrictions on $m_{F_0}(q)$ mentioned above.  
Further details are supplied below in Appendix~\ref{lemmabound}.

Let $\sigma^2_{U|F_0(\bm{Z}^{\top}\bm{\gamma}_0)}(q)\equiv E\left[U^2\left|F_0\left(\bm{Z}^{\top}\bm{\gamma}_0\right)=q\right.\right]$, where $U$ is 
the disturbance in the outcome equation (\ref{outcome}).  The following result summarizes the large-sample behaviour to first order of the proposed 
estimator:

\begin{theorem}\label{mainthm}
Under the conditions of Assumptions~\ref{a1} and \ref{a2}, the estimator $\hat{\theta}_n$ given above in 
(\ref{mhat}) satisfies
\[
\sqrt{nh_n}\left(\hat{\theta}_n-\theta_0-\frac{h_n^p}{p!}\int u^p K(u) du\cdot m^{(p)}_{F_0}(1)\right)\stackrel{d}{\to} N\left(0,\sigma^2_{U|F_0(\bm{Z}^{\top}\bm{\gamma}_0)}(1)\int K^2(u) du\right)
\]
as $n\to\infty$, where $m^{(p)}_{F_0}(1)=\lim_{q\uparrow 1}\left.\left(d^p/(dq^{\prime p}\right) m_{F_0}(q^{\prime})\right|_{q^{\prime}=q}$ for $m_{F_0}(\cdot)$ as given above in (\ref{m0F}).
\end{theorem}

It follows from Theorem~\ref{mainthm} that the rate of convergence of $\hat{\theta}_n$ to its limiting normal distribution 
is unaffected by the dependence, if any, between the disturbance terms $U$ and $V$ in (\ref{outcome}) and (\ref{selection}), 
respectively.  The rate of convergence of $\hat{\theta}_n$ is also unaffected 
by the relative upper tail thicknesses of the distributions of $V$ and of the selection index 
$\bm{Z}^{\top}\bm{\gamma}_0$.  

We also have from the statement of Theorem~\ref{mainthm} that a necessary condition for the consistency of 
$\hat{\theta}_n$ is the finiteness of the derivative $m^{(p)}_{F_0}(1)$.  
The finiteness of $m^{(p)}_{F_0}(1)$, as discussed above, is implied by the identification of $\bm{\gamma}_0$ up to a location 
and scale normalization as well as by the differentiability conditions 
specified in Assumptions~\ref{a2}.\ref{a2d2}--\ref{a2}.\ref{a2d3}.  From this it follows that the consistency of the 
proposed estimator is implied by natural restrictions on the joint distribution of 
$[\begin{array}{ccc} U & V & \bm{Z}^{\top}\bm{\gamma}_0\end{array}]$.  These distributional restrictions correspond 
collectively to a standard smoothness restriction in the literature on kernel estimation of conditional mean functions.

The presence of $m^{(p)}_{F_0}(1)$ in the bias term appearing in Theorem~\ref{mainthm}, however, 
indicates that the approximate large-sample bias of $\hat{\theta}_n$ does depend on the extent to which $U$ is mean 
dependent on $V$.  In particular, the conditional mean derivative $m^{(p)}_{F_0}(1)$ depends on the smoothness of the conditional 
mean $E\left[U\left|F_0(V)=q\right.\right]$ as a function of $q$ for values of $q$ near one; Appendix~\ref{lemmabound} below
contains further discussion.  It is worth noting in this connection that $m^{(p)}_{F_0}(1)=0$ when the selection mechanism is exogenous to the extent 
that $U$ is mean independent of $V$, i.e., when $P\left[E[U|V]= E[U]\right]=1$.  

The dependence of the approximate large-sample bias of $\hat{\theta}_n$ on the joint distribution of 
$[\begin{array}{ccc} U & V & \bm{Z}^{\top}\bm{\gamma}_0\end{array}]$ through the conditional mean derivative 
$m^{(p)}_{F_0}(1)$ can be ameliorated in practice by a judicious choice of variable
bandwidth; see Corollary~\ref{cormain} below and the corresponding discussion and simulation 
evidence presented in Section~\ref{mc}.  
Theorem~\ref{mainthm} in any case indicates that the approximate large-sample bias, but not the variance, of the proposed estimator 
depends on the joint distribution of $[\begin{array}{ccc} U & V & \bm{Z}^{\top}\bm{\gamma}_0\end{array}]$.  
Theorem~\ref{mainthm} as such distinguishes the asymptotic behaviour of $\hat{\theta}_n$ from those of existing estimators 
of $\theta_0$ \citep[e.g.,][]{Heckman90,Lewbel07} whose biases and variances both depend on the joint distribution 
of $[\begin{array}{ccc} U & V & \bm{Z}^{\top}\bm{\gamma}_0\end{array}]$.

The following corollary is immediate from Theorem~\ref{mainthm}:

\begin{corollary}\label{cormain}
The following hold under the conditions of Theorem~\ref{mainthm}:
\begin{enumerate}
\item\label{cor1} If the additional condition that $nh_n^{2p+1}\to 0$ holds, we have
\[
\sqrt{nh_n}\left(\hat{\theta}_n-\theta_0\right)\stackrel{d}{\to} N\left(0,\sigma^2_{U|F_0\left(\bm{Z}^{\top}\bm{\gamma}_0\right)}(1)\int K^2(u) du\right)
\]
as $n\to\infty$.
\item\label{cor2} The theoretical bandwidth $h^*_n$ minimizing the asymptotic mean-squared error of $\hat{\theta}_n$ is given by
\[
h^*_n=\left[\frac{(p!)^2\sigma^2_{U|F_0\left(\bm{Z}^{\top}\bm{\gamma}_0\right)}(1)\int K^2(u) du}{2p\left(\int u^p K(u) du\right)^2\left(m^{(p)}_{F_0}(1)\right)^2 \cdot n}\right]^{\frac{1}{2p+1}}.
\]
\end{enumerate}
\end{corollary}

\section{Rate Optimality}\label{opt}

\noindent Continue to let $p\ge 2$ be as specified above in Assumption~\ref{a2} in the 
previous section, and $n$ the sample size.  This section shows that under the conditions of Assumptions~\ref{a1} 
and \ref{a2}, the rate $n^{-p/(2p+1)}$ is the fastest achievable, or {\it optimal}, rate of convergence of an 
estimator of the intercept $\theta_0$ in (\ref{outcome}).  The optimality in question is relative to the 
convergence rates of all other estimators of $\theta_0$ and excludes by definition those estimators that are asymptotically 
superefficient at particular points in the underlying parameter space.  The exclusion of 
superefficient estimators in this context notably rules out estimators that converge at the parametric 
rate of $n^{-1/2}$ under conditions stronger than those given above in Assumptions~\ref{a1} and \ref{a2}.  For 
example, the OLS estimator of $\theta_0$ based on observations for which $D_i=1$ is superefficient for 
specifications of (\ref{outcome})--(\ref{observation}) over submodels in which the disturbance $U$ and the selection indicator $D$ are conditionally 
mean independent given $\bm{X}$ and $\bm{Z}$, i.e., models where $P\left[E[U|D=1,\bm{X},\bm{Z}]=E[U|D=0,\bm{X},\bm{Z}]=0\right]=1$.  As noted in the 
Introduction, the OLS estimator of $\theta_0$ converges at the standard rate of $n^{-1/2}$ when $U$ and $D$ are conditionally mean independent given 
$\bm{X}$ and $\bm{Z}$ but is otherwise inconsistent.  

The approach to optimality taken here follows that of \citet{Horowitz93}, which was in turn based on the 
approach of \citet{Stone80}.  In particular, let 
$\left\{\Psi_n:\,n=1,2,3,\ldots\right\}$ denote a sequence of sets of the form 
\begin{equation}\label{Psin}
\Psi_n=\left\{\psi:\,\psi=(\bm{\psi}_1,g)\right\},
\end{equation} 
where $\bm{\psi}_1=(\theta,\bm{\beta}^{\top},\bm{\gamma}^{\top})^{\top}\in\mathbb{R}^{1+k+l}$ and where 
$g$ denotes the joint conditional density given $\bm{X}$ and $\bm{Z}$ of the disturbances $U$ and $V$ appearing above in (\ref{outcome}) 
and (\ref{selection}), respectively.  The quantity $\psi$ may depend generically on $n$. 

Consider the observable random variables $D$, $Y$, $\bm{X}$ and $\bm{Z}$ appearing above in 
(\ref{outcome})--(\ref{observation}).  Suppose that for each $n$, the joint conditional distribution of the 
vector $(D,Y)$ given $\bm{X}$ and $\bm{Z}$ is indexed by some $\psi\in\Psi_n$.  Let 
$P_{\psi}[\cdot]\equiv P_{(\bm{\psi}_1,g)}[\cdot]$ denote the corresponding probability measure.  Following 
\citet{Stone80}, one may in this context define a constant $\rho>0$ to be an {\it upper bound} on the rate of 
convergence of estimators of the intercept parameter $\theta_0$ if for every estimator sequence 
$\left\{\theta_n\right\}$,
\begin{equation}\label{ub1}
\liminf_{n\to\infty}\sup_{\psi\in\Psi_n} P_{\psi}\left[\left|\theta_n-\theta\right|> sn^{-\rho}\right]>0
\end{equation}
for all $s>0$, and if 
\begin{equation}\label{ub2}
\lim_{s\to 0}\liminf_{n\to\infty}\sup_{\psi\in\Psi_n} P_{\psi}\left[\left|\theta_n-\theta\right|>sn^{-\rho}\right]=1,
\end{equation}
where $\theta$ as it appears in (\ref{ub1}) and (\ref{ub2}) refers to the first component of the 
finite-dimensional component $\bm{\psi}_1$ of $\psi$.  

In addition, define $\rho>0$ to be an {\it achievable} rate of convergence for the intercept parameter if there 
exists an estimator sequence $\left\{\theta_n\right\}$ such that
\begin{equation}\label{ach}
\lim_{s\to\infty}\limsup_{n\to\infty}\sup_{\psi\in\Psi_n} P_{\psi}\left[\left|\theta_n-\theta\right|> sn^{-\rho}\right]=0.
\end{equation}
One calls $\rho>0$ the {\it optimal} rate of convergence for estimation of the intercept parameter if it is both 
an upper bound on the rate of convergence and achievable.  In what follows, I first show that for large $n$, 
$p/(2p+1)$ is an upper bound on the rate of convergence.  I then show that there exists an implementation of 
the estimator given in (\ref{mhat}) that attains the $n^{-p/(2p+1)}$-rate of convergence uniformly 
over $\Psi_n$ as $n\to\infty$. 

The approach taken first involves the specification for each $n$ of a subset $\Psi^*_n$ of the parameter set 
$\Psi_n$ in which the finite-dimensional component $\bm{\psi}_1\equiv\bm{\psi}_{1n}$ lies in a shrinking 
neighbourhood $\Psi^*_{1n}$ of some point 
$[\begin{array}{ccc} \theta_0 & \bm{\beta}_0^{\top} & \bm{\gamma}_0^{\top}\end{array}]^{\top}\in\mathbb{R}^{1+k+l}$.  In addition, the 
infinite-dimensional component $g$ is embedded in a curve (i.e., parametrization) indexed by a scalar $\psi_{2n}$ on a shrinking 
neighbourhood $\Psi^*_{2n}$ of a bivariate density function $g_0$ satisfying all relevant conditions of Assumptions~\ref{a1} and \ref{a2} for a 
conditional density of $U$ and $V$ given $\bm{X}$ and $\bm{Z}$.  

In particular, consider a parametrization of the conditional joint density $g$ of $(U,V)$ given $\bm{X}$ and $\bm{Z}$ given by $g_{\psi_{2n}}$ for 
$\psi_2\in\Psi^*_{2n}$, where for some $\psi_{2n0}\in\Psi^*_{2n}$, we have $g_{\psi_{2n0}}(u,v|\bm{x},\bm{z})=g_0(u,v|\bm{x},\bm{z})$ for each 
$[\begin{array}{cccc} u & v & \bm{x}^{\top} & \bm{z}^{\top}\end{array}]\in\mathbb{R}^{2+k+l}$; i.e., the curve on $\Psi^*_{2n}$ given by 
$\psi_{2n}\to g_{\psi_{2n}}$ passes through the true conditional joint density $g_0$ at some point $\psi_{2n0}\in\Psi^*_{2n}$.

Now let $\Psi^*_n\equiv \Psi^*_{1n}\times\Psi^*_{2n}$.  Let $s>0$ be arbitrary, and let $\left\{\theta_n\right\}$ denote an arbitrary sequence of estimators of 
$\theta_0$.  Consider that if
\begin{equation}\label{ub1b}
\liminf_{n\to\infty}\sup_{\psi_{n}\in\Psi^*_{n}} P_{\psi_{n}}\left[n^{\frac{p}{2p+1}}\left|\theta_n-\theta\right|>s\right]>0,
\end{equation}
then (\ref{ub1}) holds with $\rho=p/(2p+1)$.  This is because the set $\Psi_n$ in (\ref{ub1}) contains the 
set over which the supremum is taken in (\ref{ub1b}).  Similarly, if 
\begin{equation}\label{ub2b}
\lim_{s\to 0}\liminf_{n\to\infty}\sup_{\psi_{n}\in\Psi^*_{n}} P_{\psi_{n}}\left[n^{\frac{p}{2p+1}}\left|\theta_n-\theta\right|>s\right]=1
\end{equation}
holds, then so does (\ref{ub2}).

It follows that proving (\ref{ub1b}) and (\ref{ub2b}) suffices to show that $p/(2p+1)$ is an upper bound on 
the rate of convergence; the key step in the proof is the specification of a suitable parametrization 
$\psi_{2n}\to g_{\psi_{2n}}$ for $\psi_{2n}\in\Psi^*_{2n}$.  This is in fact the approach taken 
in Appendix~\ref{ubthmpf}, which contains a proof of the following result:

\begin{theorem}\label{ubthm}

Under the conditions of Assumptions~\ref{a1} and \ref{a2}, (\ref{ub1b}) and (\ref{ub2b}) hold.
\end{theorem}
Theorem~\ref{ubthm} implies that $p/(2p+1)$ is an upper bound on the rate of convergence of an estimator 
sequence $\left\{\theta_n\right\}$ in the minimax sense of (\ref{ub1}) and (\ref{ub2}) above.

Next, it is shown that $p/(2p+1)$ is an achievable rate of convergence in the sense of (\ref{ach}) by 
exhibiting an estimator sequence $\left\{\theta_n\right\}$ such that (\ref{ach}) holds with $\rho=p/(2p+1)$.  
In this connection, let $\hat{\theta}_n^*$ denote the proposed estimator given above in (\ref{mhat}) implemented 
with a bandwidth $h_n^*=cn^{-1/(2p+1)}$ for some constant $c>0$.  In this case, (\ref{ach}) is satisfied with 
$\theta_n=\hat{\theta}^*_n$ and $\rho=p/(2p+1)$:

\begin{theorem}\label{achthm}

Suppose Assumptions~\ref{a1} and \ref{a2} hold.  Then (\ref{ach}) holds with $\theta_n=\hat{\theta}_n^*$ and 
$\rho=p/(2p+1)$, where $\hat{\theta}_n^*$ denotes the estimator given above in (\ref{mhat}) implemented with a bandwidth 
$h_n^*=c n^{-1/(2p+1)}$ for some constant $c>0$. 
\end{theorem}

Theorems~\ref{ubthm} and \ref{achthm} jointly imply that $p/(2p+1)$ is the optimal rate of convergence for 
estimation of $\theta_0$.

\section{Numerical Evidence}\label{mc}

\noindent This section reports the results of simulation experiments that compare the finite-sample 
behaviour of the estimator in (\ref{mhat}) to the behaviours of alternative estimators.  The simulations involved:
\begin{itemize}
\item variation in the correlation between the unobservable terms in the outcome and selection equations;
\item variation in the relative upper tail thicknesses of the selection index and the unobservable term in the 
selection equation, thus implying variation in the degree to which the parameter of interest is identified;
\item and the imposition of two different parametric families for the joint distribution of 
$[\begin{array}{ccc} U & V & \bm{Z}^{\top}\bm{\gamma}_0\end{array}]$, where $U$ is the unobservable term in the 
outcome equation, $V$ is the unobservable term in the selection equation and $\bm{Z}^{\top}\bm{\gamma}_0$ is the 
selection index.
\end{itemize}
Each simulation experiment involved 1000 replicated samples of sizes $n\in\{100,400\}$ from the model given above in 
(\ref{outcome})--(\ref{observation}), where the parameter of interest was fixed at $\theta_0=1$, the variance 
of the unobservable term in the selection equation was fixed at $Var[V]=1$ and where for some constant 
$\rho\in [-1,1]$, the unobservable term in the outcome equation was specified as $U=\rho V+E$ for a random variable 
$E$ independent of $V$, where $E\sim N\left(0,1-\rho^2\right)$.  The parameter $\rho$ in this case is by construction 
the correlation coefficient between $U$ and $V$.  The simulations considered the settings 
$\rho\in\{0,.25,.50,.75,.95\}$. 

In addition, the vector $\bm{Z}$ of observable predictors of selection was taken to be $l$-variate with 
$\bm{Z}=[\begin{array}{ccc} Z_1 & \cdots & Z_l\end{array}]^{\top}$, 
while the vector $\bm{X}$ of outcome predictors was specified to be $k$-variate with $k<l$ and 
$\bm{X}=[\begin{array}{ccc} Z_1 & \cdots & Z_k\end{array}]^{\top}$.  The coefficient 
vector attached to $\bm{X}$ was set to $\bm{\beta}_0=\bm{\iota}_k$, i.e., the $k$-dimensional 
unit vector.  The simulations imposed the settings $l=7$, $k=4$, which were intentionally set to equal the 
dimensions of the corresponding vectors appearing in the empirical example used below in Section~\ref{ee}.

The selection index $\bm{Z}^{\top}\bm{\gamma}_0$ and the selection-equation disturbance term were simulated from two 
data-generating processes (DGPs), considered in turn.  The distributions of $\bm{Z}^{\top}\bm{\gamma}_0$ under both 
DGPs, as required by Assumption~\ref{a1}, are 
absolutely continuous.  In addition, the parameter $\alpha>0$ used in the specification of 
both models is defined so as to index the degree to which the parameter of interest $\theta_0$ is identified.  In particular, $\alpha\ge 1$ can be 
seen in this context to be a necessary condition for the identification of $\theta_0$, with $\alpha\ge 1$ corresponding to the case where the 
transformation $F_0(V)$ has an absolutely continuous distribution with density given by the ratio $r_Q(q)$ defined in (\ref{marg}) below.  Values of 
$\alpha<1$, on the other hand, correspond to the non-identifiability of $\theta_0$, with $\alpha\in (0,1)$ in the context of 
either of the following two DGPs implying failure of the condition that $F_0(V)$ have an absolutely continuous 
distribution supported on $[0,1]$.  In particular, $\alpha\in (0,1)$ in the following two DGPs implies that the quantity 
$r_Q(q)$ given in (\ref{marg}) below has the property that $r_Q(1)=\infty$:
 
\begin{itemize}
\item (DGP1) I take
\[
\left[\begin{array}{cc}  \bm{Z} & \bm{0} \\ \bm{0} & V \end{array}\right]\sim N\left(\bm{0},\bm{I}_{l+1}\right).
\]
In addition, the selection parameter $\bm{\gamma}_0$ is set to $\bm{\gamma}_0=[\begin{array}{ccc} \sqrt{\alpha/l} & \cdots & \sqrt{\alpha/l}\end{array}]^{\top}$ for 
a constant $\alpha>0$.  In this way we have 
\[
\left[\begin{array}{c}\bm{Z}^{\top}\bm{\gamma}_0 \\ V\end{array}\right]\sim N\left(\bm{0},\left[\begin{array}{cc} \alpha & 0 \\ 0 & 1\end{array}\right]\right).
\] 
\item (DGP2)  $Z_1,\ldots,Z_l$ are iid standard Cauchy and jointly mutually independent of $V$, while 
$V$ is absolutely continuous on $[1,\infty)$ with Pareto type-I density given by 
\[
g_{0V}(v)=\alpha v^{-\alpha-1},
\] 
where $\alpha>0$ is a constant.  In addition, the selection parameter $\bm{\gamma}_0$ is set to 
$\bm{\gamma}_0=[\begin{array}{cc} \bm{0}_{l-1}^{\top} & 1\end{array}]^{\top}$.  In this way the selection index 
$\bm{Z}^{\top}\bm{\gamma}_0$ is standard Cauchy and independent of $V$.
\end{itemize}
The simulations under both DGPs considered the settings $\alpha\in\{2.00,1.50,1.25,1.00\}$.  The effect of variation in 
the correlation coefficient $\rho$ and the parameter $\alpha$ on estimation of the intercept was a primary focus of 
these simulations.  The outcome-equation nuisance parameter $\bm{\beta}_0$ and the selection parameter $\bm{\gamma}_0$ 
were accordingly fixed at their true values in these simulations in order to provide a clearer picture of the effects of 
variation in $\rho$ and $\alpha$ on the behaviour of the various intercept estimators considered.

The proposed intercept estimator given above in (\ref{mhat}) was implemented with a standard 
(i.e., second-order) Epanechnikov kernel.  The bandwidth used to implement $\hat{\theta}_n$ was initially 
set to the sample analogue of the theoretical asymptotic MSE-optimal bandwidth $h^*_n$ specified above in 
Corollary~\ref{cormain}.  In particular, the simulations involved the bandwidth $\hat{h}^*_n$, where $\hat{h}^*_n$ was taken to be 
the sample analogue of $h^*_n$.  As such, $\hat{h}^*_n$ was set 
to decay at the MSE-optimal rate of $n^{-1/5}$ corresponding to the order of kernel employed (i.e., $p=2$), while its 
leading constant was specified as the sample analogue of the leading constant appearing in $h^*_n$.  The unknown 
parameters appearing in the leading constant of $h^*_n$ were estimated via auxiliary locally cubic regressions as 
described in \citet[\S 4.3]{FanGijbels96}.  The sensitivity of the proposed estimator's sampling behaviour to the choice 
of bandwidth was also assessed by considering implementations using the bandwidth settings $h_n=(2/3)\hat{h}_{n}^*$ and 
$h_n=(3/2)\hat{h}_{n}^*$.  

Comparisons of the corresponding sampling behaviours in terms of squared bias, standard deviation and root mean-squared 
error (RMSE) over 1000 Monte Carlo replications for values of 
$(\rho,\alpha)\in\{0,.25,.50,.75,.95\}\times\{2.00,1.50,1.25,1.00\}$  are presented below for samples of size $n=100$ in 
Tables~\ref{dgp1n100mhat1.p2} and \ref{dgp2n100mhat1.p2} for DGP1 and DGP2, respectively.  The corresponding results for samples of size $n=400$ 
appear in Tables~\ref{dgp1n400mhat1.p2} and \ref{dgp2n400mhat1.p2}.  The RMSE figures displayed 
in these tables are multiplied by $\sqrt{n}$ in order to provide a clearer indication of 
the rate of convergence of the proposed estimator.  The increases in $\sqrt{n}\times$ RMSE as one moves from simulated 
samples of size $n=100$ to those of $n=400$ indicate the slower-than-parametric rate of convergence of the proposed 
estimator regardless of the precise setting of $(\rho,\alpha)$.  It is also clear that the rate of convergence of 
the proposed estimator is slower for settings of $(\rho,\alpha)$ with one or both of $\rho$ and $\alpha$ close to one.  
In addition, and as predicted by Theorem~\ref{mainthm} above, one can see that the effect of variation in 
$(\rho,\alpha)$ on the squared bias of the proposed estimator is more pronounced than the corresponding effect 
on the variance; indeed the standard deviation of the proposed estimator tends to be relatively stable  
over the various settings of $(\rho,\alpha)$ used in the simulations, particularly for values of 
$(\rho,\alpha)\in\{.25,.50,.75\}\times\{2.00,1.50,1.25\}$.  Finally, the tabulated results indicate that the sampling performance of $\hat{\theta}_n$ 
is not sensitive to moderate variations in bandwidth.  

\FloatBarrier
\begin{table}[H]
\pagenumbering{gobble}
{\scriptsize
\centering
\caption{DGP1 (bivariate normal), $n=100$, 1000 replications. Proposed estimator with second-order Epanechnikov kernel ($p=2$).  (RMSE is multiplied by $\sqrt{n}$.)}\label{dgp1n100mhat1.p2}
\begin{tabular}{ccccccccccccc}
  \hline
\multirow{2}{*}{$\rho$} &\multicolumn{3}{c}{$\alpha=2.00$} &\multicolumn{3}{c}{$\alpha=1.50$} &\multicolumn{3}{c}{$\alpha=1.25$} &\multicolumn{3}{c}{$\alpha=1.00$}\\  
                        & sq bias   & sd     & RMSE            & sq bias   & sd     & RMSE            & sq bias   & sd     & RMSE            & sq bias   & sd     & RMSE \\ 
  \hline
\multicolumn{13}{c}{(optimal bandwidth)} \\
0.0000 & 0.0081 & 0.1859 & 2.0645 & 0.0018 & 0.1872 & 1.9196 & 0.0008 & 0.1951 & 1.9723 & 0.0006 & 0.1900 & 1.9154 \\ 
0.2500 & 0.0003 & 0.1890 & 1.8985 & 0.0004 & 0.1857 & 1.8674 & 0.0023 & 0.1911 & 1.9695 & 0.0064 & 0.1802 & 1.9724 \\ 
0.5000 & 0.0014 & 0.1812 & 1.8503 & 0.0071 & 0.1861 & 2.0432 & 0.0127 & 0.1802 & 2.1253 & 0.0227 & 0.1779 & 2.3324 \\ 
0.7500 & 0.0092 & 0.1768 & 2.0125 & 0.0216 & 0.1784 & 2.3125 & 0.0353 & 0.1761 & 2.5750 & 0.0459 & 0.1687 & 2.7272 \\ 
0.9500 & 0.0215 & 0.1821 & 2.3373 & 0.0404 & 0.1665 & 2.6105 & 0.0548 & 0.1730 & 2.9114 & 0.0812 & 0.1626 & 3.2811 \\ 
\multicolumn{13}{c}{($2/3\times$ optimal bandwidth)} \\
0.0000 & 0.0081 & 0.1860 & 2.0657 & 0.0018 & 0.1873 & 1.9206 & 0.0008 & 0.1952 & 1.9733 & 0.0006 & 0.1901 & 1.9164 \\ 
0.2500 & 0.0003 & 0.1891 & 1.8996 & 0.0004 & 0.1858 & 1.8682 & 0.0023 & 0.1912 & 1.9702 & 0.0064 & 0.1803 & 1.9727 \\ 
0.5000 & 0.0014 & 0.1813 & 1.8507 & 0.0071 & 0.1862 & 2.0430 & 0.0126 & 0.1803 & 2.1243 & 0.0227 & 0.1780 & 2.3314 \\ 
0.7500 & 0.0091 & 0.1770 & 2.0109 & 0.0215 & 0.1786 & 2.3100 & 0.0351 & 0.1762 & 2.5724 & 0.0457 & 0.1689 & 2.7242 \\ 
0.9500 & 0.0212 & 0.1822 & 2.3336 & 0.0401 & 0.1666 & 2.6058 & 0.0545 & 0.1732 & 2.9068 & 0.0809 & 0.1628 & 3.2764 \\ 
\multicolumn{13}{c}{($3/2\times$ optimal bandwidth)} \\
0.0000 & 0.0081 & 0.1858 & 2.0640 & 0.0018 & 0.1872 & 1.9192 & 0.0008 & 0.1951 & 1.9719 & 0.0006 & 0.1900 & 1.9150 \\ 
0.2500 & 0.0003 & 0.1890 & 1.8980 & 0.0004 & 0.1857 & 1.8670 & 0.0023 & 0.1910 & 1.9692 & 0.0064 & 0.1802 & 1.9723 \\ 
0.5000 & 0.0014 & 0.1811 & 1.8502 & 0.0071 & 0.1861 & 2.0433 & 0.0127 & 0.1802 & 2.1257 & 0.0228 & 0.1779 & 2.3328 \\ 
0.7500 & 0.0093 & 0.1768 & 2.0132 & 0.0217 & 0.1784 & 2.3136 & 0.0354 & 0.1760 & 2.5762 & 0.0460 & 0.1687 & 2.7285 \\ 
0.9500 & 0.0216 & 0.1820 & 2.3389 & 0.0406 & 0.1664 & 2.6126 & 0.0550 & 0.1729 & 2.9134 & 0.0814 & 0.1625 & 3.2831 \\    
\hline
\end{tabular}
} 
\end{table}

\begin{table}[H]
\pagenumbering{gobble}
{\scriptsize
\centering
\caption{DGP2 (non-normal), $n=100$, 1000 replications. Proposed estimator with second-order Epanechnikov kernel ($p=2$) and optimal bandwidth.  (RMSE is multiplied by $\sqrt{n}$.)}\label{dgp2n100mhat1.p2}
\begin{tabular}{ccccccccccccc}
  \hline
\multirow{2}{*}{$\rho$} &\multicolumn{3}{c}{$\alpha=2.00$} &\multicolumn{3}{c}{$\alpha=1.50$} &\multicolumn{3}{c}{$\alpha=1.25$} &\multicolumn{3}{c}{$\alpha=1.00$}\\  
                        & sq bias   & sd     & RMSE            & sq bias   & sd     & RMSE            & sq bias   & sd     & RMSE            & sq bias   & sd     & RMSE \\ 
  \hline
\multicolumn{13}{c}{(optimal bandwidth)} \\
0.0000 & 0.1541 & 0.1719 & 4.2854 & 0.1892 & 0.1805 & 4.7097 & 0.2268 & 0.1708 & 5.0593 & 0.2714 & 0.1612 & 5.4535 \\ 
0.2500 & 0.0202 & 0.2119 & 2.5523 & 0.0394 & 0.2070 & 2.8669 & 0.0445 & 0.2229 & 3.0684 & 0.0636 & 0.2280 & 3.3994 \\ 
0.5000 & 0.0079 & 0.2391 & 2.5503 & 0.0071 & 0.2741 & 2.8687 & 0.0006 & 0.2942 & 2.9521 & 0.0000 & 0.3010 & 3.0102 \\ 
0.7500 & 0.1113 & 0.2796 & 4.3536 & 0.0972 & 0.3386 & 4.6030 & 0.0735 & 0.3288 & 4.2615 & 0.0585 & 0.4767 & 5.3448 \\ 
0.9500 & 0.2747 & 0.3174 & 6.1273 & 0.2510 & 0.3605 & 6.1721 & 0.2317 & 0.4509 & 6.5956 & 0.1950 & 0.4976 & 6.6528 \\ 
\multicolumn{13}{c}{($2/3\times$ optimal bandwidth)} \\
0.0000 & 0.1522 & 0.1725 & 4.2652 & 0.1871 & 0.1813 & 4.6906 & 0.2247 & 0.1715 & 5.0405 & 0.2693 & 0.1619 & 5.4357 \\ 
0.2500 & 0.0190 & 0.2125 & 2.5336 & 0.0376 & 0.2078 & 2.8422 & 0.0426 & 0.2242 & 3.0466 & 0.0613 & 0.2293 & 3.3752 \\ 
0.5000 & 0.0091 & 0.2396 & 2.5796 & 0.0084 & 0.2754 & 2.9025 & 0.0010 & 0.2960 & 2.9769 & 0.0000 & 0.3032 & 3.0319 \\ 
0.7500 & 0.1178 & 0.2799 & 4.4284 & 0.1036 & 0.3404 & 4.6850 & 0.0789 & 0.3307 & 4.3392 & 0.0632 & 0.4805 & 5.4233 \\ 
0.9500 & 0.2874 & 0.3178 & 6.2324 & 0.2633 & 0.3623 & 6.2816 & 0.2435 & 0.4536 & 6.7023 & 0.2055 & 0.5015 & 6.7598 \\
\multicolumn{13}{c}{($3/2\times$ optimal bandwidth)} \\ 
0.0000 & 0.1550 & 0.1716 & 4.2942 & 0.1902 & 0.1801 & 4.7180 & 0.2277 & 0.1705 & 5.0675 & 0.2724 & 0.1609 & 5.4612 \\ 
0.2500 & 0.0208 & 0.2116 & 2.5607 & 0.0401 & 0.2066 & 2.8777 & 0.0453 & 0.2224 & 3.0780 & 0.0646 & 0.2274 & 3.4100 \\ 
0.5000 & 0.0073 & 0.2389 & 2.5382 & 0.0066 & 0.2736 & 2.8547 & 0.0005 & 0.2934 & 2.9420 & 0.0001 & 0.3000 & 3.0015 \\ 
0.7500 & 0.1086 & 0.2795 & 4.3217 & 0.0945 & 0.3379 & 4.5683 & 0.0712 & 0.3280 & 4.2285 & 0.0565 & 0.4750 & 5.3115 \\ 
0.9500 & 0.2693 & 0.3173 & 6.0825 & 0.2458 & 0.3598 & 6.1255 & 0.2268 & 0.4497 & 6.5502 & 0.1906 & 0.4960 & 6.6072 \\ 
   \hline
\end{tabular}
} 
\end{table}

\FloatBarrier
\begin{table}[H]
\pagenumbering{gobble}
{\scriptsize
\centering
\caption{DGP1 (bivariate normal), $n=400$, 1000 replications. Proposed estimator with second-order Epanechnikov kernel ($p=2$).  (RMSE is multiplied by $\sqrt{n}$.)}\label{dgp1n400mhat1.p2}
\begin{tabular}{ccccccccccccc}
  \hline
\multirow{2}{*}{$\rho$} &\multicolumn{3}{c}{$\alpha=2.00$} &\multicolumn{3}{c}{$\alpha=1.50$} &\multicolumn{3}{c}{$\alpha=1.25$} &\multicolumn{3}{c}{$\alpha=1.00$}\\  
                        & sq bias   & sd     & RMSE            & sq bias   & sd     & RMSE            & sq bias   & sd     & RMSE            & sq bias   & sd     & RMSE \\ 
  \hline
\multicolumn{13}{c}{(optimal bandwidth)} \\
0.0000 & 0.0073 & 0.0972 & 2.5875 & 0.0021 & 0.0917 & 2.0524 & 0.0007 & 0.0942 & 1.9547 & 0.0000 & 0.0933 & 1.8683 \\ 
0.2500 & 0.0007 & 0.0924 & 1.9180 & 0.0001 & 0.0966 & 1.9371 & 0.0017 & 0.0916 & 2.0082 & 0.0052 & 0.0920 & 2.3367 \\ 
0.5000 & 0.0011 & 0.0962 & 2.0330 & 0.0057 & 0.0908 & 2.3585 & 0.0112 & 0.0959 & 2.8557 & 0.0199 & 0.0917 & 3.3640 \\ 
0.7500 & 0.0082 & 0.0919 & 2.5766 & 0.0194 & 0.0924 & 3.3416 & 0.0293 & 0.0894 & 3.8594 & 0.0460 & 0.0889 & 4.6434 \\ 
0.9500 & 0.0178 & 0.0922 & 3.2404 & 0.0345 & 0.0878 & 4.1101 & 0.0505 & 0.0840 & 4.7995 & 0.0731 & 0.0845 & 5.6663 \\ 
\multicolumn{13}{c}{($2/3\times$ optimal bandwidth)} \\
0.0000 & 0.0073 & 0.0972 & 2.5890 & 0.0021 & 0.0918 & 2.0535 & 0.0007 & 0.0942 & 1.9557 & 0.0000 & 0.0934 & 1.8694 \\ 
0.2500 & 0.0007 & 0.0924 & 1.9204 & 0.0001 & 0.0966 & 1.9382 & 0.0017 & 0.0916 & 2.0077 & 0.0052 & 0.0921 & 2.3354 \\ 
0.5000 & 0.0011 & 0.0963 & 2.0315 & 0.0056 & 0.0908 & 2.3546 & 0.0111 & 0.0960 & 2.8519 & 0.0198 & 0.0918 & 3.3594 \\ 
0.7500 & 0.0080 & 0.0920 & 2.5692 & 0.0192 & 0.0924 & 3.3330 & 0.0291 & 0.0894 & 3.8510 & 0.0458 & 0.0891 & 4.6352 \\ 
0.9500 & 0.0175 & 0.0923 & 3.2283 & 0.0342 & 0.0879 & 4.0973 & 0.0502 & 0.0841 & 4.7867 & 0.0727 & 0.0846 & 5.6532 \\ 
\multicolumn{13}{c}{($3/2\times$ optimal bandwidth)} \\ 
0.0000 & 0.0073 & 0.0971 & 2.5869 & 0.0021 & 0.0917 & 2.0520 & 0.0007 & 0.0941 & 1.9542 & 0.0000 & 0.0933 & 1.8678 \\ 
0.2500 & 0.0007 & 0.0924 & 1.9170 & 0.0001 & 0.0965 & 1.9367 & 0.0017 & 0.0916 & 2.0084 & 0.0052 & 0.0920 & 2.3372 \\ 
0.5000 & 0.0011 & 0.0962 & 2.0337 & 0.0057 & 0.0908 & 2.3602 & 0.0112 & 0.0959 & 2.8574 & 0.0199 & 0.0917 & 3.3660 \\ 
0.7500 & 0.0082 & 0.0919 & 2.5799 & 0.0195 & 0.0923 & 3.3453 & 0.0293 & 0.0893 & 3.8632 & 0.0461 & 0.0889 & 4.6471 \\ 
0.9500 & 0.0179 & 0.0921 & 3.2458 & 0.0346 & 0.0878 & 4.1158 & 0.0507 & 0.0840 & 4.8052 & 0.0733 & 0.0844 & 5.6721 \\ 
\hline
\end{tabular}
} 
\end{table}

\begin{table}[H]
\pagenumbering{gobble}
{\scriptsize
\centering
\caption{DGP2 (non-normal), $n=400$, 1000 replications. Proposed estimator with second-order Epanechnikov kernel ($p=2$).  (RMSE is multiplied by $\sqrt{n}$.)}\label{dgp2n400mhat1.p2}
\begin{tabular}{ccccccccccccc}
  \hline
\multirow{2}{*}{$\rho$} &\multicolumn{3}{c}{$\alpha=2.00$} &\multicolumn{3}{c}{$\alpha=1.50$} &\multicolumn{3}{c}{$\alpha=1.25$} &\multicolumn{3}{c}{$\alpha=1.00$}\\  
                        & sq bias   & sd     & RMSE            & sq bias   & sd     & RMSE            & sq bias   & sd     & RMSE            & sq bias   & sd     & RMSE \\ 
  \hline
\multicolumn{13}{c}{(optimal bandwidth)} \\
0.0000 & 0.1460 & 0.0908 & 7.8539 & 0.1864 & 0.0884 & 8.8133 & 0.2179 & 0.0881 & 9.5014 & 0.2658 & 0.0831 & 10.4445 \\ 
0.2500 & 0.0184 & 0.1045 & 3.4218 & 0.0317 & 0.1053 & 4.1369 & 0.0426 & 0.1153 & 4.7277 & 0.0651 & 0.1158 & 5.6037 \\ 
0.5000 & 0.0095 & 0.1255 & 3.1783 & 0.0048 & 0.1261 & 2.8790 & 0.0025 & 0.1498 & 3.1582 & 0.0004 & 0.1898 & 3.8157 \\ 
0.7500 & 0.1296 & 0.1433 & 7.7487 & 0.1058 & 0.1769 & 7.4044 & 0.0912 & 0.1956 & 7.1956 & 0.0686 & 0.2002 & 6.5943 \\ 
0.9500 & 0.2897 & 0.1551 & 11.2027 & 0.2707 & 0.1828 & 11.0287 & 0.2632 & 0.2245 & 11.1992 & 0.2507 & 0.2913 & 11.5847 \\ 
\multicolumn{13}{c}{($2/3\times$ optimal bandwidth)} \\
0.0000 & 0.1428 & 0.0913 & 7.7763 & 0.1829 & 0.0891 & 8.7369 & 0.2143 & 0.0888 & 9.4263 & 0.2621 & 0.0837 & 10.3754 \\ 
0.2500 & 0.0163 & 0.1050 & 3.3076 & 0.0289 & 0.1061 & 4.0066 & 0.0392 & 0.1164 & 4.5947 & 0.0611 & 0.1170 & 5.4683 \\ 
0.5000 & 0.0120 & 0.1259 & 3.3398 & 0.0068 & 0.1271 & 3.0284 & 0.0040 & 0.1515 & 3.2815 & 0.0010 & 0.1918 & 3.8896 \\ 
0.7500 & 0.1426 & 0.1434 & 8.0783 & 0.1180 & 0.1784 & 7.7409 & 0.1022 & 0.1968 & 7.5083 & 0.0780 & 0.2022 & 6.8948 \\ 
0.9500 & 0.3146 & 0.1552 & 11.6389 & 0.2946 & 0.1837 & 11.4600 & 0.2865 & 0.2260 & 11.6209 & 0.2723 & 0.2939 & 11.9783 \\ 
\multicolumn{13}{c}{($3/2\times$ optimal bandwidth)} \\ 
0.0000 & 0.1473 & 0.0905 & 7.8874 & 0.1879 & 0.0882 & 8.8462 & 0.2195 & 0.0878 & 9.5336 & 0.2674 & 0.0828 & 10.4742 \\ 
0.2500 & 0.0193 & 0.1042 & 3.4714 & 0.0329 & 0.1049 & 4.1929 & 0.0440 & 0.1149 & 4.7847 & 0.0668 & 0.1153 & 5.6615 \\ 
0.5000 & 0.0085 & 0.1254 & 3.1136 & 0.0041 & 0.1256 & 2.8209 & 0.0020 & 0.1491 & 3.1118 & 0.0002 & 0.1890 & 3.7895 \\ 
0.7500 & 0.1243 & 0.1432 & 7.6116 & 0.1009 & 0.1763 & 7.2650 & 0.0868 & 0.1950 & 7.0661 & 0.0649 & 0.1993 & 6.4701 \\ 
0.9500 & 0.2797 & 0.1551 & 11.0222 & 0.2610 & 0.1824 & 10.8502 & 0.2537 & 0.2239 & 11.0246 & 0.2419 & 0.2902 & 11.4212 \\ 
   \hline
\end{tabular}
} 
\end{table}

I next consider the simulated performances over 1000 Monte Carlo replications across DGPs, sample sizes and settings 
of $(\rho,\alpha)$ of several alternative estimators of the intercept $\theta_0$.  The results for samples of size $n=100$ are summarized below in 
Tables~\ref{dgp1n100balt} and \ref{dgp2n100balt} for DGPs 1 and 2, respectively.  The corresponding results for samples of size $n=400$ 
appear in Tables~\ref{dgp1n400balt} and \ref{dgp2n400balt}.  The standard Heckman 2-step estimator is found under DGP1 to have a performance in terms of RMSE that is comparable to that of the 
proposed estimator in (\ref{mhat}).  The proposed estimator under DGP2, on the other hand, is found to dominate in terms of RMSE the performance 
of the following alternative estimators under all combinations of $(\rho,\alpha)$ considered:  

\begin{itemize}
\item (OLS)  The ordinary least squares estimator of the intercept parameter using only those 
observations for which $D=1$.  

These results are consistent with well established theory.  In particular, Tables~\ref{dgp1n100balt}--\ref{dgp2n100balt} indicate the good performance 
of OLS when $\rho=0$ and the poor performance of OLS when $\rho>0$.  In addition, the decrease in $\sqrt{n}\times$ RMSE for the OLS estimator when $\rho=0$ as 
one moves from $n=100$ to $n=400$ is suggestive of superefficiency, while at the same time the increases in $\sqrt{n}\times$RMSE 
when $\rho>0$ is consistent with OLS being inconsistent under $\rho>0$.

\item (2-step)  The estimator of the intercept based on the well known procedure of 
\citet{Heckman76,Heckman79}, which is known to be $\sqrt{n}$-consistent if $[\begin{array}{cc} U & V\end{array}]$ is 
bivariate normal (i.e., generated according to DGP1).  

The results for DGP1 given in Table~\ref{dgp1n100balt} below are consistent with expectations; in particular, the 
2-step procedure exhibits an RMSE that is stable across the various configurations of $(\rho,\alpha)$ that were tried.  In addition, a comparison 
of the relevant sections of Table~\ref{dgp1n100balt} and Table~\ref{dgp1n400balt} highlights the stability 
of $\sqrt{n}\times$ RMSE as one moves from $n=100$ to $n=400$, which is consistent with the $\sqrt{n}$-consistency of the procedure under DGP1.

The results for DGP2 appearing in Table~\ref{dgp2n100balt}, on the other hand, show that the performance of 
the 2-step procedure can vary dramatically with $(\rho,\alpha)$.  A comparison of the $\sqrt{n}\times$RMSE figures in 
Table~\ref{dgp2n400balt} with those in Table~\ref{dgp2n100balt} also suggests that the 2-step procedure under DGP2 is 
superefficient at $\rho=0$ and converges at a slower-than-parametric rate for model specifications with $\rho>0$.

\item (H90) The intercept estimator suggested by \citet{Heckman90}, which in the context of the model specified in (\ref{outcome})--(\ref{observation}) 
has the form
\begin{equation}\label{h90}
\hat{\theta}_{H90}\equiv\frac{\sum_{i=1}^n D_i\left(Y_i-\bm{X}^{\top}_i\hat{\bm{\beta}}\right)1\left\{\bm{Z}_i\hat{\bm{\gamma}}>b_n\right\}}{\sum_{i=1}^n D_i 1\left\{\bm{Z}_i\hat{\bm{\gamma}}>b_n\right\}}
\end{equation}
for some sequence of positive constants $\{b_n\}$ with $b_n\to\infty$ as $n\to\infty$.  I present the results of 
simulations in which the nuisance-parameter estimators $\hat{\bm{\beta}}$ and $\hat{\bm{\gamma}}$ are fixed at 
the true values of the corresponding estimands.  These results appear below in Tables~\ref{dgp1n100balt}, \ref{dgp2n100balt}, \ref{dgp1n400balt} and 
\ref{dgp2n400balt} for $b_n$ equal to the sample .95-quantile of $\bm{Z}^{\top}_i\bm{\gamma}_0$.  

Table~\ref{dgp1n100balt} below indicates that the performance of $\hat{\theta}_{H90}$ is comparable to that of the 2-step procedure under DGP1 in that its 
RMSE is stable over changes in $(\rho,\alpha)$.  The stability in the $\sqrt{n}\times$RMSE figures in these tables as one moves from $n=100$ to 
$n=400$, which is evident from a comparison of Table~\ref{dgp1n100balt} with Table~\ref{dgp1n400balt}, also suggests that 
$\hat{\theta}_{H90}$ may be $\sqrt{n}$-consistent under DGP1.

Table~\ref{dgp2n100balt}, on the other hand, shows that the performance of $\hat{\theta}_{H90}$ can deteriorate dramatically as $\rho$ moves away from 
zero, although its performance under DGP2 appears to be unaffected by variation in $\alpha$ for any given value of $\rho$.  The $\sqrt{n}\times$RMSE 
figures in Table~\ref{dgp2n100balt} and Table~\ref{dgp2n400balt} indicate that $\hat{\theta}_{H90}$ has a 
slower-than-parametric rate of convergence under DGP2 that is highly sensitive to variation in $\rho$ but relatively insensitive to variation in 
$\alpha$.

\item (AS98)  The intercept estimator developed by \citet{AndrewsSchafgans98} as a generalization of the procedure of 
\citet{Heckman90}.  The AS98 estimator in the context of the 
model given above in (\ref{outcome})--(\ref{observation}) has the form
\begin{equation}\label{as98}
\hat{\theta}_{AS98}\equiv\frac{\sum_{i=1}^n D_i\left(Y_i-\bm{X}^{\top}_i\hat{\bm{\beta}}\right) s\left(\bm{Z}_i^{\top}\hat{\bm{\gamma}}-b_n\right)}{\sum_{i=1}^n D_i s\left(\bm{Z}_i^{\top}\hat{\bm{\gamma}}_n-b_n\right)},
\end{equation} 
where, following \citet[eq. (4.1)]{AndrewsSchafgans98}, I set
\begin{equation}\label{as98s}
s(u)=\left\{\begin{array}{ccc} 1-\exp\left(-\frac{u}{\tau-u}\right) &,&\,x\in (0,\tau)\\
								0 &,&\,x\le 0\\
								1 &,&\,x\ge \tau\end{array}.\right.
\end{equation}
Note that the setting $\tau=0$ reduces $\hat{\theta}_{AS98}$ to $\hat{\theta}_{H90}$ as given earlier in (\ref{h90}).  
In addition, the tuning parameter $b_n$ in (\ref{as98}), as it does for $\hat{\theta}_{H90}$ in (\ref{h90}) above, refers 
to a sequence of positive constants with $b_n\to\infty$ as $n\to\infty$.

I present, in common with other simulations reported here, results for $\hat{\theta}_{AS98}$ in which the 
nuisance-parameter estimators $\hat{\bm{\beta}}$ and $\hat{\bm{\gamma}}$ are fixed at the true values of the corresponding 
estimands.  These simulations also involve setting the nuisance parameter $\tau$ in (\ref{as98s}) to the sample 
median of $\bm{Z}^{\top}_i\bm{\gamma}_0$ and the smoothing parameter $b_n$ in (\ref{as98}) to the sample .95-quantile 
of $\bm{Z}^{\top}_i\bm{\gamma}_0$.  The corresponding results appear below in 
Tables~\ref{dgp1n100balt}--\ref{dgp2n100balt} and also in Tables~\ref{dgp1n400balt}--\ref{dgp2n400balt}.

It is clear from Tables~\ref{dgp1n100balt}--\ref{dgp2n100balt} below that $\hat{\theta}_{AS}$ is numerically 
unstable under DGP1 but numerically stable under DGP2.  Table~\ref{dgp2n100balt} also indicates the 
sensitivity of the performance of $\hat{\theta}_{AS}$ to variation in $(\rho,\alpha)$.  A comparison of Table~\ref{dgp2n100balt} with 
Table~\ref{dgp2n400balt} also underscores the slower-than-parametric rate of convergence of $\hat{\theta}_{AS98}$.
\end{itemize}

\FloatBarrier
\begin{table}[H]
\pagenumbering{gobble}
{\scriptsize
\centering
\caption{DGP1 (bivariate normal), $n=100$, 1000 replications.  Alternative estimators.  (RMSE is multiplied by $\sqrt{n}$.)}\label{dgp1n100balt}
\begin{tabular}{ccccccccccccc}
  \hline
\multirow{2}{*}{$\rho$} &\multicolumn{3}{c}{$\alpha=2.00$} &\multicolumn{3}{c}{$\alpha=1.50$} &\multicolumn{3}{c}{$\alpha=1.25$} &\multicolumn{3}{c}{$\alpha=1.00$}\\  
                        & sq bias   & sd     & RMSE            & sq bias   & sd     & RMSE            & sq bias   & sd     & RMSE            & sq bias   & sd     & RMSE \\ 
  \hline
\multicolumn{13}{c}{(OLS)} \\
0.0000 & 0.0000 & 0.1788 & 1.7882 & 0.0000 & 0.1693 & 1.6936 & 0.0000 & 0.1728 & 1.7280 & 0.0001 & 0.1746 & 1.7495 \\ 
0.2500 & 0.0219 & 0.1727 & 2.2745 & 0.0256 & 0.1698 & 2.3326 & 0.0265 & 0.1795 & 2.4239 & 0.0306 & 0.1651 & 2.4053 \\ 
0.5000 & 0.0936 & 0.1656 & 3.4796 & 0.1029 & 0.1593 & 3.5820 & 0.1152 & 0.1605 & 3.7550 & 0.1218 & 0.1565 & 3.8253 \\ 
0.7500 & 0.2162 & 0.1582 & 4.9117 & 0.2378 & 0.1609 & 5.1352 & 0.2559 & 0.1464 & 5.2658 & 0.2712 & 0.1431 & 5.4004 \\ 
0.9500 & 0.3443 & 0.1458 & 6.0461 & 0.3887 & 0.1362 & 6.3815 & 0.4124 & 0.1362 & 6.5651 & 0.4361 & 0.1325 & 6.7356 \\ 
\multicolumn{13}{c}{(Heckman 2-step)} \\
0.0000 & 0.0000 & 0.3007 & 3.0068 & 0.0001 & 0.3299 & 3.3007 & 0.0000 & 0.3503 & 3.5025 & 0.0003 & 0.3886 & 3.8898 \\ 
0.2500 & 0.0011 & 0.3143 & 3.1596 & 0.0023 & 0.3265 & 3.2994 & 0.0009 & 0.3564 & 3.5757 & 0.0016 & 0.3650 & 3.6722 \\ 
0.5000 & 0.0031 & 0.2963 & 3.0150 & 0.0046 & 0.3320 & 3.3892 & 0.0058 & 0.3338 & 3.4248 & 0.0064 & 0.3712 & 3.7973 \\ 
0.7500 & 0.0083 & 0.2969 & 3.1056 & 0.0106 & 0.3148 & 3.3116 & 0.0111 & 0.3270 & 3.4354 & 0.0111 & 0.3509 & 3.6635 \\ 
0.9500 & 0.0171 & 0.2678 & 2.9806 & 0.0169 & 0.2952 & 3.2264 & 0.0131 & 0.3221 & 3.4179 & 0.0207 & 0.3052 & 3.3747 \\ 
\multicolumn{13}{c}{(\citet{Heckman90} ($b_n=\hat{F}^{-1}_{\bm{Z}^{\top}\bm{\gamma}_0}(.95)$))} \\
0.0000 & 0.0003 & 0.4509 & 4.5124 & 0.0002 & 0.4524 & 4.5261 & 0.0004 & 0.4557 & 4.5611 & 0.0001 & 0.4689 & 4.6895 \\ 
0.2500 & 0.0004 & 0.4488 & 4.4929 & 0.0006 & 0.4507 & 4.5135 & 0.0003 & 0.4488 & 4.4918 & 0.0005 & 0.4513 & 4.5180 \\ 
0.5000 & 0.0000 & 0.4481 & 4.4815 & 0.0011 & 0.4454 & 4.4661 & 0.0015 & 0.4423 & 4.4403 & 0.0002 & 0.4342 & 4.3443 \\ 
0.7500 & 0.0011 & 0.4516 & 4.5283 & 0.0001 & 0.4372 & 4.3732 & 0.0025 & 0.4379 & 4.4076 & 0.0040 & 0.4315 & 4.3615 \\ 
0.9500 & 0.0030 & 0.4531 & 4.5641 & 0.0020 & 0.4363 & 4.3863 & 0.0030 & 0.4326 & 4.3604 & 0.0060 & 0.4327 & 4.3953 \\ 
\multicolumn{13}{c}{(\citet{AndrewsSchafgans98} ($b_n=\hat{F}^{-1}_{\bm{Z}^{\top}\bm{\gamma}_0}(.95)$))} \\
0.0000 & Inf    & Inf    & Inf    & Inf    & Inf    & Inf    & Inf    & Inf    & Inf    & Inf    & Inf    & Inf \\ 
0.2500 & Inf    & Inf    & Inf    & Inf    & Inf    & Inf    & Inf    & Inf    & Inf    & Inf    & Inf    & Inf \\ 
0.5000 & 0.0000 & 0.5650 & 5.6499 & Inf    & Inf    & Inf    & Inf    & Inf    & Inf    & 0.0000 & 0.5579 & 5.5796 \\ 
0.7500 & 0.0019 & 0.5757 & 5.7739 & Inf    & Inf    & Inf    & 0.0025 & 0.5465 & 5.4883 & Inf    & Inf    & Inf \\ 
0.9500 & Inf    & Inf    & Inf    & 0.0017 & 0.5772 & 5.7865 & 0.0007 & 0.5559 & 5.5656 & 0.0036 & 0.5365 & 5.3982 \\ 
   \hline
\end{tabular}
} 
\end{table}

\begin{table}[H]
\pagenumbering{gobble}
{\scriptsize
\centering
\caption{DGP2 (non-normal), $n=100$, 1000 replications. Alternative estimators.  (RMSE is multiplied by $\sqrt{n}$.)}\label{dgp2n100balt}
\begin{tabular}{ccccccccccccc}
  \hline
\multirow{2}{*}{$\rho$} &\multicolumn{3}{c}{$\alpha=2.00$} &\multicolumn{3}{c}{$\alpha=1.50$} &\multicolumn{3}{c}{$\alpha=1.25$} &\multicolumn{3}{c}{$\alpha=1.00$}\\  
                        & sq bias   & sd     & RMSE            & sq bias   & sd     & RMSE            & sq bias   & sd     & RMSE            & sq bias   & sd     & RMSE \\ 
  \hline
\multicolumn{13}{c}{(OLS)} \\
0.0000 & 0.0001 & 0.4821 & 4.8221 & 0.0000 & 0.3909 & 3.9091 & 0.0014 & 1.0831 & 10.8368 & 0.0005 & 0.6892 & 6.8959 \\ 
0.2500 & 0.1510 & 0.3288 & 5.0909 & 0.1736 & 0.4121 & 5.8603 & 0.2394 & 0.7670 & 9.0977 & 0.2633 & 2.1317 & 21.9257 \\ 
0.5000 & 0.5682 & 0.3763 & 8.4250 & 0.7799 & 0.7149 & 11.3620 & 0.8475 & 0.4896 & 10.4266 & 1.1118 & 1.1903 & 15.9015 \\ 
0.7500 & 1.2924 & 0.2869 & 11.7248 & 1.7087 & 0.4380 & 13.7860 & 1.8913 & 1.8977 & 23.4365 & 2.3993 & 1.4258 & 21.0529 \\ 
0.9500 & 2.1300 & 0.3085 & 14.9169 & 2.6775 & 0.5903 & 17.3954 & 3.2066 & 1.2800 & 22.0111 & 3.6936 & 4.8277 & 51.9616 \\ 
\multicolumn{13}{c}{(Heckman 2-step)} \\
0.0000 & 0.0000 & 0.5162 & 5.1620 & 0.0002 & 0.5989 & 5.9913 & 0.0020 & 1.6850 & 16.8557 & 0.0055 & 1.1861 & 11.8844 \\ 
0.2500 & 0.1849 & 0.5086 & 6.6596 & 0.2433 & 0.6251 & 7.9627 & 0.3461 & 0.7657 & 9.6563 & 0.4278 & 1.4736 & 16.1227 \\ 
0.5000 & 0.7718 & 0.6221 & 10.7649 & 1.2359 & 1.9152 & 22.1445 & 1.3865 & 0.9240 & 14.9673 & 2.3721 & 2.9882 & 33.6172 \\ 
0.7500 & 1.6186 & 0.5040 & 13.6844 & 1.0515 & 16.7408 & 167.7214 & 3.2286 & 2.1692 & 28.1672 & 4.6267 & 2.7391 & 34.8276 \\ 
0.9500 & 2.7113 & 0.5127 & 17.2457 & 3.8929 & 0.9294 & 21.8101 & 5.4307 & 3.6028 & 42.9077 & 8.1794 & 4.7219 & 55.2045 \\
\multicolumn{13}{c}{(\citet{Heckman90} ($b_n=\hat{F}^{-1}_{\bm{Z}^{\top}\bm{\gamma}_0}(.95)$))} \\
0.0000 & 0.0001 & 0.4388 & 4.3889 & 0.0001 & 0.4741 & 4.7418 & 0.0004 & 0.4621 & 4.6254 & 0.0000 & 0.4743 & 4.7428 \\ 
0.2500 & 0.2121 & 0.4677 & 6.5640 & 0.3145 & 0.5099 & 7.5793 & 0.4265 & 0.6050 & 8.9023 & 0.5743 & 0.6945 & 10.2790 \\ 
0.5000 & 0.8606 & 0.4819 & 10.4540 & 1.3574 & 0.8578 & 14.4677 & 1.5614 & 0.8493 & 15.1087 & 2.1035 & 0.9589 & 17.3869 \\ 
0.7500 & 1.8488 & 0.5264 & 14.5806 & 2.9232 & 1.0078 & 19.8465 & 3.4719 & 0.9811 & 21.0580 & 5.0749 & 1.8515 & 29.1596 \\ 
0.9500 & 3.0237 & 0.7279 & 18.8508 & 4.3979 & 1.0647 & 23.5192 & 5.6603 & 1.6073 & 28.7121 & 7.8382 & 1.9172 & 33.9323 \\
\multicolumn{13}{c}{(\citet{AndrewsSchafgans98} ($b_n=\hat{F}^{-1}_{\bm{Z}^{\top}\bm{\gamma}_0}(.95)$))} \\
0.0000 & 0.0000 & 0.9234 & 9.2336 & 0.0015 & 1.0125 & 10.1326 & 0.0017 & 1.0037 & 10.0452 & 0.0001 & 0.9807 & 9.8079 \\ 
0.2500 & 0.2035 & 1.0138 & 11.0963 & 0.3754 & 1.1914 & 13.3967 & 0.6937 & 1.8209 & 20.0236 & 1.0490 & 2.1058 & 23.4170 \\ 
0.5000 & 0.8978 & 1.0810 & 14.3748 & 2.0336 & 2.9706 & 32.9515 & 2.3156 & 3.0879 & 34.4245 & 3.2494 & 2.6734 & 32.2432 \\ 
0.7500 & 1.8968 & 1.1576 & 17.9912 & Inf & Inf & Inf & 4.9692 & 2.9053 & 36.6200 & 10.1655 & 7.9701 & 85.8419 \\ 
0.9500 & 3.3558 & 2.1412 & 28.1792 & 5.0293 & 2.8442 & 36.2197 & 7.8162 & 3.3294 & 43.4752 & 13.5506 & 7.1477 & 80.3995 \\ 
   \hline
\end{tabular}
} 
\end{table}

\begin{table}[H]
\pagenumbering{gobble}
{\scriptsize
\centering
\caption{DGP1 (bivariate normal), $n=400$, 1000 replications. Alternative estimators.  (RMSE is multiplied by $\sqrt{n}$.)}\label{dgp1n400balt}
\begin{tabular}{ccccccccccccc}
  \hline
\multirow{2}{*}{$\rho$} &\multicolumn{3}{c}{$\alpha=2.00$} &\multicolumn{3}{c}{$\alpha=1.50$} &\multicolumn{3}{c}{$\alpha=1.25$} &\multicolumn{3}{c}{$\alpha=1.00$}\\  
                        & sq bias   & sd     & RMSE            & sq bias   & sd     & RMSE            & sq bias   & sd     & RMSE            & sq bias   & sd     & RMSE \\ 
  \hline
\multicolumn{13}{c}{(OLS)} \\
0.0000 & 0.0000 & 0.0866 & 1.7318 & 0.0000 & 0.0809 & 1.6194 & 0.0000 & 0.0833 & 1.6676 & 0.0000 & 0.0756 & 1.5130 \\ 
0.2500 & 0.0232 & 0.0840 & 3.4805 & 0.0258 & 0.0823 & 3.6111 & 0.0289 & 0.0829 & 3.7818 & 0.0295 & 0.0779 & 3.7712 \\ 
0.5000 & 0.0935 & 0.0805 & 6.3227 & 0.1071 & 0.0772 & 6.7232 & 0.1149 & 0.0771 & 6.9534 & 0.1211 & 0.0750 & 7.1190 \\ 
0.7500 & 0.2163 & 0.0758 & 9.4248 & 0.2390 & 0.0712 & 9.8808 & 0.2530 & 0.0712 & 10.1604 & 0.2719 & 0.0684 & 10.5179 \\ 
0.9500 & 0.3417 & 0.0695 & 11.7732 & 0.3830 & 0.0664 & 12.4488 & 0.4087 & 0.0634 & 12.8487 & 0.4366 & 0.0624 & 13.2739 \\ 
\multicolumn{13}{c}{(Heckman 2-step)} \\
0.0000 & 0.0000 & 0.1552 & 3.1042 & 0.0000 & 0.1640 & 3.2810 & 0.0000 & 0.1784 & 3.5672 & 0.0000 & 0.1832 & 3.6645 \\ 
0.2500 & 0.0000 & 0.1522 & 3.0456 & 0.0001 & 0.1591 & 3.1876 & 0.0001 & 0.1720 & 3.4472 & 0.0001 & 0.1797 & 3.6009 \\ 
0.5000 & 0.0002 & 0.1525 & 3.0627 & 0.0005 & 0.1646 & 3.3198 & 0.0003 & 0.1709 & 3.4357 & 0.0002 & 0.1770 & 3.5502 \\ 
0.7500 & 0.0009 & 0.1435 & 2.9301 & 0.0006 & 0.1559 & 3.1585 & 0.0003 & 0.1641 & 3.2977 & 0.0011 & 0.1708 & 3.4785 \\ 
0.9500 & 0.0007 & 0.1368 & 2.7857 & 0.0008 & 0.1468 & 2.9915 & 0.0012 & 0.1558 & 3.1939 & 0.0010 & 0.1614 & 3.2917 \\ 
\multicolumn{13}{c}{(\citet{Heckman90} ($b_n=\hat{F}^{-1}_{\bm{Z}^{\top}\bm{\gamma}_0}(.95)$))} \\
0.0000 & 0.0000 & 0.1623 & 3.2457 & 0.0000 & 0.1627 & 3.2548 & 0.0000 & 0.1653 & 3.3058 & 0.0000 & 0.1608 & 3.2164 \\ 
0.2500 & 0.0001 & 0.1585 & 3.1769 & 0.0001 & 0.1594 & 3.1950 & 0.0002 & 0.1566 & 3.1470 & 0.0005 & 0.1557 & 3.1449 \\ 
0.5000 & 0.0005 & 0.1582 & 3.1931 & 0.0010 & 0.1556 & 3.1755 & 0.0016 & 0.1647 & 3.3916 & 0.0028 & 0.1696 & 3.5550 \\ 
0.7500 & 0.0007 & 0.1569 & 3.1826 & 0.0016 & 0.1607 & 3.3140 & 0.0032 & 0.1557 & 3.3125 & 0.0070 & 0.1502 & 3.4401 \\ 
0.9500 & 0.0009 & 0.1539 & 3.1347 & 0.0027 & 0.1458 & 3.0967 & 0.0052 & 0.1487 & 3.3033 & 0.0098 & 0.1515 & 3.6168 \\ 
\multicolumn{13}{c}{(\citet{AndrewsSchafgans98} ($b_n=\hat{F}^{-1}_{\bm{Z}^{\top}\bm{\gamma}_0}(.95)$))} \\
0.0000 & Inf & Inf & Inf & 0.0000 & 0.3118 & 6.2369 & Inf & Inf & Inf & 0.0000 & 0.2991 & 5.9817 \\ 
0.2500 & 0.0000 & 0.3043 & 6.0860 & 0.0000 & 0.3056 & 6.1114 & Inf & Inf & Inf & Inf & Inf & Inf \\ 
0.5000 & Inf & Inf & Inf & 0.0001 & 0.2973 & 5.9497 & Inf & Inf & Inf & 0.0004 & 0.3065 & 6.1438 \\ 
0.7500 & 0.0001 & 0.2888 & 5.7792 & 0.0000 & 0.3041 & 6.0817 & Inf & Inf & Inf & Inf & Inf & Inf \\ 
0.9500 & Inf & Inf & Inf & Inf & Inf & Inf & Inf & Inf & Inf & 0.0015 & 0.3021 & 6.0908 \\ 
   \hline
\end{tabular}
} 
\end{table}

\begin{table}[H]
\pagenumbering{gobble}
{\scriptsize
\centering
\caption{DGP2 (non-normal), $n=400$, 1000 replications. Alternative estimators.  (RMSE is multiplied by $\sqrt{n}$.)}\label{dgp2n400balt}
\begin{tabular}{ccccccccccccc}
  \hline
\multirow{2}{*}{$\rho$} &\multicolumn{3}{c}{$\alpha=2.00$} &\multicolumn{3}{c}{$\alpha=1.50$} &\multicolumn{3}{c}{$\alpha=1.25$} &\multicolumn{3}{c}{$\alpha=1.00$}\\  
                        & sq bias   & sd     & RMSE            & sq bias   & sd     & RMSE            & sq bias   & sd     & RMSE            & sq bias   & sd     & RMSE \\ 
  \hline
\multicolumn{13}{c}{(OLS)} \\
0.0000 & 0.0000 & 0.1225 & 2.4496 & 0.0000 & 0.1295 & 2.5902 & 0.0000 & 0.1356 & 2.7151 & 0.0000 & 0.1423 & 2.8469 \\ 
0.2500 & 0.1464 & 0.1253 & 8.0533 & 0.1855 & 0.1377 & 9.0428 & 0.2246 & 0.1549 & 9.9707 & 0.2653 & 0.1651 & 10.8174 \\ 
0.5000 & 0.5746 & 0.1254 & 15.3660 & 0.7346 & 0.1398 & 17.3680 & 0.8895 & 0.2137 & 19.3406 & 1.0931 & 0.2617 & 21.5551 \\ 
0.7500 & 1.3338 & 0.1182 & 23.2184 & 1.6652 & 0.1752 & 26.0457 & 1.9547 & 0.2239 & 28.3186 & 2.3994 & 0.2910 & 31.5215 \\ 
0.9500 & 2.1328 & 0.1128 & 29.2955 & 2.6396 & 0.1762 & 32.6839 & 3.1089 & 0.2503 & 35.6178 & 4.0168 & 0.4762 & 41.1997 \\ 
\multicolumn{13}{c}{(Heckman 2-step)} \\
0.0000 & 0.0000 & 0.1900 & 3.8004 & 0.0000 & 0.2107 & 4.2133 & 0.0002 & 0.2368 & 4.7430 & 0.0000 & 0.2641 & 5.2828 \\ 
0.2500 & 0.1966 & 0.1881 & 9.6325 & 0.2873 & 0.2342 & 11.6982 & 0.4031 & 0.3386 & 14.3913 & 0.6158 & 0.4610 & 18.2023 \\ 
0.5000 & 0.7901 & 0.1930 & 18.1916 & 1.1622 & 0.2667 & 22.2114 & 1.6124 & 0.5121 & 27.3834 & 2.4130 & 0.6840 & 33.9460 \\ 
0.7500 & 1.8144 & 0.2124 & 27.2728 & 2.6411 & 0.4291 & 33.6168 & 3.4999 & 0.5283 & 38.8794 & 5.4115 & 0.8450 & 49.4998 \\ 
0.9500 & 2.8707 & 0.2172 & 34.1632 & 4.1789 & 0.4178 & 41.7300 & 5.6244 & 0.6344 & 49.0992 & 9.1834 & 1.2123 & 65.2785 \\ 
\multicolumn{13}{c}{(\citet{Heckman90} ($b_n=\hat{F}^{-1}_{\bm{Z}^{\top}\bm{\gamma}_0}(.95)$))} \\
0.0000 & 0.0000 & 0.1625 & 3.2499 & 0.0000 & 0.1668 & 3.3373 & 0.0001 & 0.1698 & 3.4017 & 0.0000 & 0.1703 & 3.4069 \\ 
0.2500 & 0.1881 & 0.1647 & 9.2789 & 0.2496 & 0.1838 & 10.6470 & 0.3165 & 0.2183 & 12.0686 & 0.4023 & 0.2231 & 13.4475 \\ 
0.5000 & 0.7417 & 0.1732 & 17.5695 & 0.9928 & 0.2013 & 20.3307 & 1.2631 & 0.3301 & 23.4275 & 1.6212 & 0.4082 & 26.7415 \\ 
0.7500 & 1.7112 & 0.1877 & 26.4303 & 2.2915 & 0.3216 & 30.9514 & 2.7843 & 0.3655 & 34.1634 & 3.5582 & 0.4415 & 38.7456 \\ 
0.9500 & 2.7126 & 0.1885 & 33.1547 & 3.6019 & 0.2978 & 38.4220 & 4.4380 & 0.4261 & 42.9862 & 6.0276 & 0.7106 & 51.1176 \\  
\multicolumn{13}{c}{(\citet{AndrewsSchafgans98} ($b_n=\hat{F}^{-1}_{\bm{Z}^{\top}\bm{\gamma}_0}(.95)$))} \\
0.0000 & 0.0000 & 0.4867 & 9.7345 & 0.0001 & 0.5337 & 10.6766 & 0.0001 & 0.5335 & 10.6730 & 0.0004 & 0.5566 & 11.1378 \\ 
0.2500 & 0.2272 & 0.5277 & 14.2224 & 0.4256 & 0.7807 & 20.3484 & 0.7890 & 1.5216 & 35.2385 & 1.1932 & 1.4466 & 36.2534 \\ 
0.5000 & 0.9464 & 0.7039 & 24.0155 & 1.6530 & 1.0790 & 33.5688 & 2.9148 & 2.8347 & 66.1824 & 5.3105 & 3.1842 & 78.6128 \\ 
0.7500 & 2.2403 & 0.9842 & 35.8267 & 3.9629 & 2.3050 & 60.9133 & 6.0330 & 2.7449 & 73.6684 & 10.2806 & 3.3502 & 92.7458 \\ 
0.9500 & 3.5332 & 0.9977 & 42.5610 & 5.7880 & 1.8579 & 60.7943 & 9.2207 & 3.1294 & 87.2098 & 21.6148 & 5.9786 & 151.4702 \\  
   \hline
\end{tabular}
} 
\end{table}

In summary, the simulations presented here show the potential of the proposed estimator to exhibit good performance in 
terms of RMSE across two different parametric families of data-generating process, across variation in the degree to 
which the errors $U$ and $V$ are dependent and across variation in the extent to which the parameter of interest is identified.  This assessment is 
unaffected by moderate variation in the estimated MSE-optimal bandwidth used to 
implement the proposed estimator.  Tables~\ref{dgp1n100mhat1.p2} and \ref{dgp2n100mhat1.p2} also support the conclusion of Theorem~\ref{mainthm}  
in indicating the sensitivity of the bias of the proposed estimator to variation in the parameter $(\rho,\alpha)$ 
under both DGP1 and DGP2.  These results also show that the estimated MSE-optimal bandwidth used to implement the proposed 
estimator was effective in limiting the extent to which the RMSE of the proposed estimator was sensitive to variation in 
$(\rho,\alpha)$.  In particular, the RMSE of the proposed estimator was found under DGP2 to dominate those of the other 
estimators considered.

\section{Empirical Example}\label{ee}

\noindent  This section reconsiders individual labour-market data from Malaysia that were originally 
analyzed by \citet{Schafgans00}.  The estimator developed above is applied to the problem of 
estimating the extent of plausible gender wage discrimination in Malaysia using data from 
the Second Malaysian Family Life Survey (MFLS2) conducted between August 1988 and January 1989.  Inferences 
available from the proposed estimator are compared with those obtained via 
the same alternative estimators considered in Section~\ref{mc}.  The proposed estimator was found to 
generate inferences that differ significantly from those obtainable via the alternative estimators used in 
the simulation experiments described in Section~\ref{mc}.  All estimators applied to the MFLS2 data considered 
in this section were implemented in precisely the same way in which they were implemented in the simulation 
experiments presented earlier.

I consider a decomposition of the female--male log-wage difference for ethnic Malay workers.  The 
consideration of gender wage gaps for Malaysian workers of the same ethnicity is potentially important because of the differential treatment of 
Malays in the labour force after 1970 \citep[see e.g.,][and references cited]{Schafgans00}.  I follow \citet{Schafgans00} by analyzing gender wage gaps in the MFLS2 using the basic decomposition 
technique of \citet{Oaxaca73}.  In particular, suppose that the generic model given above in 
(\ref{outcome})--(\ref{observation}) holds for both men and women, 
i.e., 
\begin{eqnarray}
Y^*_j &=& \theta_{j}+\bm{X}^{\top}_j\bm{\beta}_{j}+U_j,\label{outcomej}\\
D_j &=& 1\left\{\bm{Z}^{\top}_j\bm{\gamma}_{j}\ge V_j\right\},\label{selectionj}\\
Y_j &=& DY^*_j\label{observationj}
\end{eqnarray}
where $Y^*_j$ is the natural logarithm of the offered average hourly wage, and where the index $j\in\{0,1\}$ denotes a given gender.  For $j\in\{0,1\}$ let 
$\bar{Y}_j\equiv E\left[\left.Y_j\right|D_j=1\right]$, and let $\bar{\bm{X}}_j$ denote the average 
``endowments'' of wage-determining attributes for workers of gender $j$.  The observed log-wage gap 
$\bar{Y}_1-\bar{Y}_0$ between the two genders can then be decomposed as
\begin{eqnarray}
\bar{Y}_1-\bar{Y}_0 &=&\left[\left(\theta_1-\theta_0\right)+\bar{\bm{X}}^{\top}_0\left(\bm{\beta}_1-\bm{\beta}_0\right)\right]+(\bar{\bm{X}}_1-\bar{\bm{X}}_0)^{\top}\bm{\beta}_1\nonumber\\
	& &+\left(E\left[\left.U_1\right|D_1=1\right]-E\left[\left.U_0\right|D_0=1\right]\right)\label{decomp1}\\
	&=&\left[\left(\theta_1-\theta_0\right)+\bar{\bm{X}}^{\top}_1\left(\bm{\beta}_1-\bm{\beta}_0\right)\right]+(\bar{\bm{X}}_1-\bar{\bm{X}}_0)^{\top}\bm{\beta}_0\nonumber\\
	& &+\left(E\left[\left.U_1\right|D_1=1\right]-E\left[\left.U_0\right|D_0=1\right]\right)\label{decomp2}\\
	&\equiv & A+B+C,\label{decomp}
\end{eqnarray}
where $A$ is that part of the gap due to differences in wage structures between genders;
$B$ is due to observable differences between men and women in wage-determining 
characteristics and $C$ is the contribution of differential self-selection into the labour force.  Following 
\citet{Schafgans00} the quantity $\bar{Y}_1-\bar{Y}_0-C=A+B$ is referred to as the 
{\it selection-corrected log-wage gap}.

Wage discrimination in favor of members of gender $j=1$ is empirically plausible if the overall log-wage gap 
$\bar{Y}_1-\bar{Y}_0$ cannot be entirely explained by differential self-selection into paid work, differences in 
observed endowments or by differing returns to those endowments.  Moreover, given the definitions of the 
quantities $A$ and $B$ appearing above in (\ref{decomp}), the extent of plausible wage discrimination favoring gender 1 may be equated with the difference in 
intercepts $\theta_1-\theta_0$.  

The analysis that follows considers a subset of the sample taken from the MFLS2 of 1988--89 that 
was analyzed by \citet{Schafgans00}.  This particular dataset is publicly available from the {\it Journal of Applied Econometrics} Data Archive 
at \url{http://qed.econ.queensu.ca/jae/1998-v13.5/schafgans/}.  Each observation in this sample corresponds to a member of the labour force.  I specifically consider ethnic 
Malays residing in non-urban settings who were observed to have some level of unearned household income in terms of dividends, interests or rents, 
and who were also observed to have passed the highest level of schooling (i.e., primary on the one hand, or secondary or above) corresponding to the 
number of years of schooling observed.  This subset of the MFLS2 consisted of 965 women and 878 men.

I also use the same variable specifications used by \citet{Schafgans00}.  In particular, the outcome variable $Y^*_j$ is LWAGE, the log hourly real 
wage in the local currency deflated using the 1985 consumer price index.  The selection variable $D_j$ is the indicator PAIDWORK for whether the 
individual in question is in fact a wage worker.  The exogenous variables appearing in the selection equations for each gender include UNEARN, a 
measure of household unearned income in terms of dividends, interest and rents; HOUSEH, the value of household 
real estate owned, computed as the product of an indicator variable for house ownership and the cost of the house 
owned; and AMTLAND, the extent of household landholding in hundreds of acres.  In addition, selection into wage 
work is also assumed to be determined by AGE, in years; AGESQ, the square of AGE divided by 100; YPRIM, years of 
primary schooling and YSEC, years of schooling at the secondary level or above.  The variables appearing on the 
right-hand side of the outcome equations for each gender or ethnic group are AGE, AGESQ, YPRIM and YSEC.  
\citet[Section~4]{Schafgans00} contains further details regarding variable definitions.

For each gender $j\in\{\mbox{female, male}\}$, the proposed estimator and that of \citet{Heckman90} and \citet{AndrewsSchafgans98} rely on 
the preliminary procedure described in \citet[p. 484--487]{Schafgans98} to estimate the nuisance parameters 
$\bm{\beta}_j$ and $\bm{\gamma}_j$ appearing in (\ref{outcomej}) and (\ref{selectionj}), respectively.  This 
involves estimating the selection equation for each group via the method of \citet{KleinSpady93} and estimating 
the slope parameters in each outcome equation using the method of \citet{Robinson88}.  This is followed by 
estimation of the intercept parameter in each outcome equation via the proposed estimator.  Standard errors are 
calculated by bootstrapping with replacement with $B=200$ replications.  

Estimates of the outcome-equation 
parameters obtained via the proposed estimator are given in which the proposed estimator is implemented using the same 
kernel and estimated MSE-optimal bandwidth $\hat{h}^*_n$ used in the simulations reported in Section~\ref{mc}.  In 
common with the results given earlier in Section~\ref{mc}, I also considered implementations of the proposed 
estimator in which the bandwidth was set to $h_n=(2/3)\hat{h}^*_n$ and $h_n=(3/2)\hat{h}^*_n$.

The decomposition of the observed gender log-wage gaps for ethnic Malay workers is presented in Table~\ref{eedecomp1110}.  In keeping with the theory 
developed above, the focus is on the extent of plausible gender wage discrimination, which is identified with the difference between the estimated 
intercepts.  A striking result is the evidence provided by the proposed estimator of positive wage discrimination in 
favor of women.  In particular, Table~\ref{eedecomp1110} indicates that all three implementations of the proposed estimator imply a large, positive 
and significant difference in intercepts, while the OLS and 2-step procedures generated estimated intercept differences 
that were both insignificant.  The implementation of the H90 procedure with $b_n$ set to the .90-quantile of the estimated 
selection index generated a similarly insignificant estimate of the difference in intercepts.  The other implementation of the H90 procedure, 
along with the AS98 procedure, proved to be numerically unstable.  Table~\ref{eedecomp1110} indicates that there exists a clear difference in 
inferences regarding the extent of gender wage discrimination amongst ethnic Malay workers between estimates generated by the 
proposed estimator and those generated by established procedures.  

\FloatBarrier
\begin{landscape}
\begin{table}[H]
\pagenumbering{gobble}
{\scriptsize
\caption{Female--male log-wage decomposition, Malays.  Standard errors in parentheses}\label{eedecomp1110} 
\centering
\hspace*{-30pt}\setlength\tabcolsep{3pt}
\begin{tabular}{l|c|c|c|c|c|c|c|c}
  \hline\hline
\rule{0pt}{4ex}\multirow{2}{*}{Wage gap (overall)}  & \multicolumn{8}{c}{-0.2882} \\ 
                                                    & \multicolumn{8}{c}{(0.0425)} \\[4pt] \hline 
\rule{0pt}{4ex}\multirow{2}{*}{Female (endowment)}  & \multicolumn{8}{c}{-0.0638} \\ 
                                                    & \multicolumn{8}{c}{(0.0337)} \\[4pt] \hline 
\rule{0pt}{4ex}\multirow{2}{*}{Male (endowment)}    & \multicolumn{8}{c}{-0.0369} \\ 
                                                    & \multicolumn{8}{c}{(0.0377)}\\[4pt] 
\hline
\rule{0pt}{4ex}\multirow{2}{*}{ }             & \multicolumn{3}{c|}{\rule{0pt}{4ex} $\hat{\theta}_n$}  & \multirow{2}{*}{OLS} & \multirow{2}{*}{2-step} & \multicolumn{2}{c|}{H90} & AS98 \\[2pt] \cline{2-4}\cline{7-9} 
                               &\rule{0pt}{4ex} ($h_n=\hat{h}^*_n$) & ($h_n=(2/3)\hat{h}^*_n$) & ($h_n=(3/2)\hat{h}^*_n$) & & & ($b_n=\hat{F}^{-1}_{\bm{Z}^{\top}\hat{\bm{\gamma}}_n}(.90)$) & ($b_n=\hat{F}^{-1}_{\bm{Z}^{\top}\hat{\bm{\gamma}}_n}(.95)$) &($b_n=\hat{F}^{-1}_{\bm{Z}^{\top}\hat{\bm{\gamma}}_n}(.95)$) \\[4pt]
\hline\rule{0pt}{4ex}\multirow{2}{*}{Wage gap (selection-corrected)} & 0.9283 & 0.9328 & 0.9263 & -1.337 & -10.1628 & -0.2709 & -0.0516 & -0.1015 \\ 
                                                               & (0.7516) & (0.7515) & (0.7517) & (0.0812) & (1.0385) & (1.0617) & (NaN) & (NaN)  \\[4pt]\hline 
\rule{0pt}{4ex}\multirow{2}{*}{Female (coefficients)}          &\multicolumn{3}{c|}{0.1098}   & -1.2553  & -8.0955  &\multicolumn{2}{c|}{0.1098} & 0.1098 \\ 
                                                               &\multicolumn{3}{c|}{(0.6483)} & (0.6977) & (2.2356) &\multicolumn{2}{c|}{(0.6483)} & (0.6483) \\[4pt] \hline 
\rule{0pt}{4ex}\multirow{2}{*}{Male (coefficients)}            &\multicolumn{3}{c|}{0.0829}   & -1.2951 & -7.988    &\multicolumn{2}{c|}{0.0829} & 0.0829 \\ 
                                                               &\multicolumn{3}{c|}{(0.6457)} & (0.6989) & (2.2296) &\multicolumn{2}{c|}{(0.6457)} & (0.6457) \\[4pt] \hline 
\rule{0pt}{4ex}\multirow{2}{*}{Difference in intercepts}       & 0.8823   & 0.8867   & 0.8803   & -0.0364  & -2.1082  & -0.3169  & -0.0976 & -0.1475 \\ 
                                                               & (0.3787) & (0.3784) & (0.3789) & (0.7526) & (3.1712) & (0.8401) & (NaN)   & (NaN)  \\ 
   \hline
\end{tabular}
}
\end{table}
\end{landscape}

\section{Conclusion}\label{concl}

\noindent This paper has developed a new estimator of the intercept of a sample-selection model in which the 
joint distribution of the unobservables and the selection index is unspecified.  It has been shown that the new 
estimator can be made under mild conditions to converge in probability at an $n^{-p/(2p+1)}$-rate, where $p\ge 2$ is an 
integer that indexes the strength of certain smoothness assumptions as given above in Assumption~\ref{a2}.\ref{a2d}.  
This rate of convergence is shown to be the optimal rate of convergence for estimation of the intercept parameter in terms 
of a minimax criterion.  The new estimator is under mild conditions consistent and asymptotically normal with a rate of 
convergence that is the same regardless of the joint distribution of the unobservables and the selection 
index.  This differs from other proposals in the literature and is convenient in practice, as the extent to which 
selection is endogenous is typically unknown in applications.  In addition, the rate of convergence of the new estimator, 
unlike those of better known estimators, does not depend on assumptions regarding the relative tail behaviours of 
the determinants of selection beyond those necessary for the identification of the estimand.  This similarly 
facilitates statistical inference regarding the intercept.  Simulations presented above show the potential accuracy of 
the proposed estimator relative to that of established procedures across different model specifications.  An 
empirical example using individual labour-market data from Malaysia shows the potential of the proposed estimator to 
generate inferences regarding the extent of plausible gender wage discrimination that differ from those available from  
better known estimators.

\appendix

\section{Appendix}

\subsection{Further discussion of the finiteness of $\left.\left(\partial^p/\partial q^{p}\right) r_{U|Q}(u|q)\right|_{q=1}$ for any $u\in\mathbb{R}$}\label{lemmabound}

\noindent This appendix supplies details regarding the assertion made earlier that the conditional density $r_{U|Q}(u|q)$ given in 
(\ref{condl_dens}) satisfies $\left.\left(\partial^p/\partial q^{p}\right) r_{U|Q}(u|q)\right|_{q=1}$ for any $u\in\mathbb{R}$.  Recall in this connection that the finiteness of 
$\left.\left(\partial^p/\partial q^{p}\right) r_{U|Q}(u|q)\right|_{q=1}$ implies in turn the previously stated differentiability condition regarding the conditional mean 
function $m_{F_0}(q)$ given in (\ref{m0F}).  In particular, the differentiability of $m_{F_0}(q)$ on $(0,1)$ to $p$th order, where $p\ge 2$ is a 
constant specified in Assumption~\ref{a2}.\ref{a2d}, along with the left-continuity of the $p$th derivative of $m_{F_0}(q)$ at $q=1$, is sufficient 
to control the asymptotic bias of the proposed estimator $\hat{\theta}_n$; see Appendix~\ref{mainthmpf} below for details.

The finiteness of $\left.\left(\partial^p/\partial q^{p}\right) r_{U|Q}(u|q)\right|_{q=1}$ is a consequence firstly of the fact, developed in Lemma~\ref{bound1} in what follows, that 
identification of $\bm{\gamma}_0$ subject to Assumption~\ref{a1}.\ref{a1b} implies that $r_{U|Q}(u|q)$ is the conditional density of $U$ given $F_0(V)=q$ for any $q\in [0,1]$:

\begin{lemma}\label{bound1}
Identification of $\bm{\gamma}_0$ subject to the conditions of Assumption~\ref{a1}.\ref{a1b} implies that the random variable $F_0(V)$ 
satisfies the following:
\begin{enumerate} 
\item The distribution of $F_0(V)$ has support equal to $[0,1]$;
\item the conditional distribution of $U$ given $F_0(V)=q$ for any $q\in [0,1]$ is absolutely continuous with density given by $r_{U|Q}(u|q)$ in (\ref{condl_dens}).
\end{enumerate}
\end{lemma}
\begin{proof}
\begin{enumerate}
\item Begin by observing that $\bm{\gamma}_0$ is identified up to the location and scale normalization specified in Assumption~\ref{a1}.\ref{a1b} iff 
the mapping $\tilde{\bm{\gamma}}\to P\left[\left.V\le z_1+\tilde{\bm{z}}^{\top}\tilde{\bm{\gamma}}\right|\bm{X}=\bm{x},\bm{Z}=\bm{z}\right]$ is 1--1 on $\mathbb{R}^{l-1}$ for any $\bm{x}$ in the support of $\bm{X}$ and any
$\bm{z}=[\begin{array}{cc} z_1 & \tilde{\bm{z}}^{\top}\end{array}]^{\top}$ in the support $Supp[\bm{Z}]$ of $\bm{Z}$.  We have 
\[
P\left[\left.V\le z_1+\tilde{\bm{z}}^{\top}\tilde{\bm{\gamma}}\right|\bm{X}=\bm{x}, \bm{Z}=\bm{z}\right]=P\left[\left.F_0(V)\le F_0\left(z_1+\tilde{\bm{z}}^{\top}\tilde{\bm{\gamma}}\right)\right|\bm{X}=\bm{x},\bm{Z}=\bm{z}\right]
\] 
since $F_0$ is a distribution function, so identification of $\bm{\gamma}_0$ subject to the conditions of Assumption~\ref{a1}.\ref{a1b} holds iff the mapping 
$\tilde{\bm{\gamma}}\to P\left[\left.F_0(V)\le F_0\left(z_1+\tilde{\bm{z}}^{\top}\tilde{\bm{\gamma}}\right)\right|\bm{X}=\bm{x},\bm{Z}=\bm{z}\right]$ is 1--1 on $\mathbb{R}^{l-1}$ for any 
$[\begin{array}{cc} \bm{x}^{\top} & \bm{z}^{\top}\end{array}]$ in the support of $\bm{X}$ and $\bm{Z}$.

Let the support of the conditional distribution of $V$ given $[\begin{array}{cc}\bm{X}^{\top} & \bm{Z}^{\top}\end{array}]=[\begin{array}{cc}\bm{x}^{\top} & \bm{z}^{\top}\end{array}]$ be given by the interval 
$[v_1,v_2]$ for constants $-\infty\le v_1 < v_2\le\infty$.  Suppose that $F_0(v_2)<1$.  Then writing $\bm{z}=[\begin{array}{cc} z_1 & \tilde{\bm{z}}^{\top}\end{array}]^{\top}$, there exists a 
$\tilde{\bm{\gamma}}^{\prime}\in\mathbb{R}^{l-1}$ with $\tilde{\bm{\gamma}}^{\prime}\neq\bm{0}$ such that 
$F_0\left(z_1+\tilde{\bm{z}}^{\top}\tilde{\bm{\gamma}}^{\prime}\right)\ge F_0(v_2)$, which implies that
\begin{eqnarray*}
 & & P\left[\left.F_0(V)\le F\left(z_1+\tilde{\bm{z}}^{\top}\tilde{\bm{\gamma}}^{\prime}\right)\right|\bm{X}=\bm{x},\bm{Z}=\bm{z}\right]\\
 &=& P\left[\left.F_0(V)\le F_0\left(z_1+\tilde{\bm{z}}^{\top}\left(2\tilde{\bm{\gamma}}^{\prime}\right)\right)\right|\bm{X}=\bm{x},\bm{Z}=\bm{z}\right]\\
 &=& 1,
\end{eqnarray*} 
from which it follows that $\bm{\gamma}_0$ is not identified.  A failure of identification accordingly ensues when $F_0(v_2)<1$.

Similarly, if $F_0(v_1)>0$ we have that $F_0\left(z_1+\tilde{\bm{z}}^{\top}\tilde{\bm{\gamma}}^{\prime\prime}\right)< F_0(v_1)$ for 
some $\tilde{\bm{\gamma}}^{\prime\prime}\in\mathbb{R}^{l-1}$ with $\tilde{\bm{\gamma}}^{\prime\prime}\neq\bm{0}$, so that 
\begin{eqnarray*}
 & & P\left[\left.F_0(V)\le F_0\left(z_1+\tilde{\bm{z}}^{\top}\tilde{\bm{\gamma}}^{\prime\prime}\right)\right|\bm{X}=\bm{x},\bm{Z}=\bm{z}\right]\\
 &=& P\left[\left.F_0(V)\le F_0\left(z_1+\tilde{\bm{z}}^{\top}\left(.5\tilde{\bm{\gamma}}^{\prime\prime}\right)\right)\right|\bm{X}=\bm{x},\bm{Z}=\bm{z}\right]\\
 &=& 0.
\end{eqnarray*}  
This implies a similar failure of identification when $F_0(v_1)>0$.  

It follows that identification of $\bm{\gamma}_0$ subject to Assumption~\ref{a1}.\ref{a1b} implies that the support of the 
conditional distribution of $F_0(V)$ given $[\begin{array}{cc}\bm{X}^{\top} & \bm{Z}^{\top}\end{array}]=[\begin{array}{cc} \bm{x}^{\top} & \bm{z}^{\top}\end{array}]$, for any 
$[\begin{array}{cc} \bm{x}^{\top} & \bm{z}^{\top}\end{array}]$ in the support of $\bm{X}$ and $\bm{Z}$, is $[0,1]$.  The support of the conditional distribution given 
$[\begin{array}{cc}\bm{X}^{\top} & \bm{Z}^{\top}\end{array}]=[\begin{array}{cc} \bm{x}^{\top} & \bm{z}^{\top}\end{array}]$ coincides with that of the marginal 
distribution.

\item Part~1 of this lemma shows that the distribution of $F_0(V)$ has support equal to $[0,1]$.  Assumption~\ref{a1}.\ref{a1b4} implies that the mapping 
$v\to F_0(v)$ is 1--1 and strictly monotone on the support $Supp[V]$ of $V$.  A standard argument accordingly 
shows that the density of $Q\equiv F_0(V)$ is given by 
\begin{equation}\label{marg}
r_Q(q)\equiv\frac{g_{V}\left(F_0^{-1}(q)\right)}{f_0\left(F^{-1}_0(q)\right)},
\end{equation}  
where $g_V(\cdot)$ denotes the marginal density of $V$.  Similarly, the joint density of $[\begin{array}{cc} U & Q\end{array}]$ is given by 

\begin{equation}\label{joint}
r_{UQ}(u,q)\equiv\frac{g_{UV}\left(u,F^{-1}_0(q)\right)}{f_0\left(F^{-1}_0(q)\right)},
\end{equation}  
where $g_{UV}(\cdot,\cdot)$ denotes the joint density of $U$ and $V$.

That the conditional density of $U$ given $F_0(V)=q$ for any $q\in [0,1]$ has the desired form is immediate.
\end{enumerate}
\end{proof}

It should be noted that the marginal distribution of the disturbance term $V$ in the selection equation is restricted to have a right 
tail that is related to that of the selection index $\bm{Z}^{\top}\bm{\gamma}_0$ via the finiteness for all $q\in [0,1]$ of the marginal density $r_Q(q)$ in (\ref{marg}).
For example, in the case where $V$ and $\bm{Z}^{\top}\bm{\gamma}_0$ are both normally distributed with the scale normalization $Var[V]=1$, the condition $r_Q(1)<\infty$ implies that the variance of $V$ 
is no greater than that of $\bm{Z}^{\top}\bm{\gamma}_0$, i.e., $r_Q(1)<\infty$ in this case implies that $Var\left[\bm{Z}^{\top}\bm{\gamma}_0\right]\ge 1$.  More generally, 
the condition $r_Q(1)<\infty$ rules out situations where the marginal distribution of $V$ has an upper tail that is strictly heavier than that of $\bm{Z}^{\top}\bm{\gamma}_0$.

This restriction on the relative upper-tail behaviours of the distributions of $V$ and $\bm{Z}^{\top}\bm{\gamma}_0$ 
is weaker than the restrictions on the joint distribution of $[\begin{array}{cc} V & \bm{Z}^{\top}\bm{\gamma}_0\end{array}]$ that feature in e.g., \citet{AndrewsSchafgans98} 
or \citet{Lewbel07}.  In particular, Assumption~\ref{a1} does not imply restrictions on the relative upper tail 
thicknesses of the distributions of $V$ and the selection index $\bm{Z}^{\top}\bm{\gamma}_0$ 
that are beyond those necessary for the identification of $\bm{\gamma}_0$.

Lemma~\ref{bound1} and the differentiability conditions of Assumption~\ref{a2}.\ref{a2d} imply the desired finiteness of $\left.\left(\partial^p/\partial q^{p}\right) r_{U|Q}(u|q)\right|_{q=1}$:
\begin{lemma}\label{bound2}
The conditional density $r_{U|Q}(u|q)$ in (\ref{condl_dens}) satisfies $\left.\left(\partial^p/\partial q^{p}\right) r_{U|Q}(u|q)\right|_{q=0}<\infty$ and 
$\left.\left(\partial^p/\partial q^{p}\right) r_{U|Q}(u|q)\right|_{q=1}<\infty$ under the conditions of Assumptions~\ref{a1} and \ref{a2}.\ref{a2d} for any $u\in\mathbb{R}$.
\end{lemma}
\begin{proof}

Recall from the proof of Lemma~\ref{bound1} that $r_{U|Q}(u|q)=r_{UQ}(u,q)/r_{Q}(q)$, where $r_Q(q)$ and $r_{UQ}(u,q)$ are as given above in 
(\ref{marg}) and (\ref{joint}), respectively.  We have for any $u\in\mathbb{R}$ that $r_{UQ}(u,\cdot)$ is right- and left-continuous at $q=0$ and $q=1$, 
respectively, given the conclusion of Lemma~\ref{bound1} that the distribution of $F_0(V)$ is absolutely continuous with support $[0,1]$.  Similarly, 
the marginal density $r_Q(q)$ is bounded away from zero for all $q\in [0,1]$ by virtue of the conclusion of Lemma~\ref{bound1} that $Q$ has support 
$[0,1]$.  It follows that for any $u\in\mathbb{R}$, the conditional density $r_{U|Q}(u|q)$ is right- and left-continuous as a function of $q$ at $q=0$ and $q=1$, 
respectively.

Next, observe that for any $u\in\mathbb{R}$, the conditional density $r_{U|Q}(u|q)$ is $(p+1)$-times differentiable in $q$ on 
$(0,1)$ by virtue of Assumption~\ref{a2}.\ref{a2d}.  The desired conclusion is a special case of the following argument.  Let $\phi(\cdot)$ 
denote a differentiable function on $(0,1)$ so that $\sup_{x\in (0,1)}\left|\phi^{\prime}(x)\right|<\infty$.  Assume that $\phi(\cdot)$ is 
right- and left-continuous at 0 and 1, respectively.  Then
\begin{eqnarray}
 & & \phi(1)-\phi(0)\label{firsta}\\
 &=& \phi(1-)-\phi(0+)\label{firstb}\\
 &=& \int_0^1 \phi^{\prime}(x) dx\nonumber\\
 &\le & \sup_{x\in (0,1)}\left|\phi^{\prime}(x)\right|\nonumber\\
 &<&\infty,\nonumber
 \end{eqnarray}
where (\ref{firsta})--(\ref{firstb}) follows by the right- and left-continuity of $\phi(\cdot)$ at 0 and 1, respectively.

\end{proof}

\subsection{Proof of Theorem~\ref{mainthm}}\label{mainthmpf}

\noindent Begin by recalling the definition of $\hat{\eta}_n(\cdot)$ given above in (\ref{etahatn}).  Define 
in addition
\begin{equation}\label{etahat0}
\hat{\eta}_0(\bm{z})\equiv\frac{1}{n}\sum_{i=1}^n 1\left\{(\bm{Z}_i-\bm{z})^{\top}\bm{\gamma}_0\le 0\right\}.
\end{equation}
Next, let $F_n(\cdot)$ denote the cdf of $\bm{Z}^{\top}\hat{\bm{\gamma}}_n$, and define 
\begin{equation}\label{etan}
\eta_n(\bm{z})\equiv F_n\left(\bm{z}^{\top}\hat{\bm{\gamma}}_n\right)
\end{equation}
and 
\begin{equation}\label{eta0}
\eta_0(\bm{z})\equiv F_0\left(\bm{z}^{\top}\bm{\gamma}_0\right).
\end{equation}

\noindent Consider the following preliminary result that will be used repeatedly in the sequel:

\begin{lemma}\label{prelim}
Under the conditions of Assumptions~\ref{a1} and \ref{a2},
\[
\sup_{\bm{z}\in\mathbb{R}^l}\left|\sqrt{n}\left[\left(\hat{\eta}_n(\bm{z})-\eta_n(\bm{z})\right)-\left(\hat{\eta}_0(\bm{z})-\eta_0(\bm{z})\right)\right]\right|=o_p(1).
\]
\end{lemma}\label{lem41}
\begin{proof}

\noindent Lemma~\ref{prelim} involves an application of \citet[Theorem~2.1]{vanderVaartWellner07}.  In particular, let 
$\delta\in\mathbb{R}$, $\tilde{\bm{\gamma}}\in\mathbb{R}^{l-1}$ and $\tilde{\bm{z}}\in\mathbb{R}^{l-1}$ be fixed, and define the 
function $g_{\delta,\tilde{\bm{\gamma}},\tilde{\bm{z}}}:\mathbb{R}^{l-1}\to\mathbb{R}$ as 
$g_{\delta,\tilde{\bm{\gamma}},\tilde{\bm{z}}}(\bm{w})=\delta+(\bm{w}-\tilde{\bm{z}})^{\top}\tilde{\bm{\gamma}}$.  
Consider the corresponding function class $\mathcal{G}\equiv\left\{g_{\delta,\tilde{\bm{\gamma}},\tilde{\bm{z}}}:\,\delta\in\mathbb{R},\tilde{\bm{\gamma}}\in\mathbb{R}^{l-1},\tilde{\bm{z}}\in\mathbb{R}^{l-1}\right\}$.  
Observe that $\mathcal{G}$ is contained in a finite-dimensional vector space.  To see this, note that for an arbitrary non-zero constant $\lambda\in\mathbb{R}$,
$\lambda g_{\delta,\tilde{\bm{\gamma}},\tilde{\bm{z}}}(\bm{w})=g_{\lambda\delta,\lambda\tilde{\bm{\gamma}},\tilde{\bm{z}}}(\bm{w})$, while for fixed
$[\begin{array}{ccc}\delta_1 & \tilde{\bm{\gamma}}^{\top}_1 & \tilde{\bm{z}}^{\top}_1\end{array}],\,[\begin{array}{ccc}\delta_2 & \tilde{\bm{\gamma}}^{\top}_2 & \tilde{\bm{z}}^{\top}_2\end{array}]\in\mathbb{R}^{2l-1}$,
\begin{eqnarray*}
 & & g_{\delta_1,\tilde{\bm{\gamma}}_1,\tilde{\bm{z}}_1}(\bm{w})+g_{\delta_2,\tilde{\bm{\gamma}}_2,\tilde{\bm{z}}_2}(\bm{w})\\
 &=& \left(\delta_1+\delta_2+\tilde{\bm{z}}_1^{\top}\tilde{\bm{\gamma}}_2+\tilde{\bm{z}}^{\top}_2\tilde{\bm{\gamma}}_1\right)+\left[\bm{w}-\left(\tilde{\bm{z}}_1+\tilde{\bm{z}}_2\right)\right]^{\top}\left(\tilde{\bm{\gamma}}_1+\tilde{\bm{\gamma}}_2\right)\\
 &=& g_{\delta_1+\delta_2+\tilde{\bm{z}}_1^{\top}\tilde{\bm{\gamma}}_2+\tilde{\bm{z}}^{\top}_2\tilde{\bm{\gamma}}_1,\tilde{\bm{\gamma}}_1+\tilde{\bm{\gamma}}_2,\tilde{\bm{z}}_1+\tilde{\bm{z}}_2}(\bm{w}).
 \end{eqnarray*}
It follows that $\mathcal{G}$ is a VC-class, which implies that its negativity sets also constitute a VC-class \citep[e.g.,][Lemma~2.6.18]{vanderVaartWellner96}.
 
As such, it follows that the class of indicator functions of $\left\{\bm{w}\in\mathbb{R}^l:\,g_{\delta,\tilde{\bm{\gamma}},\tilde{\bm{z}}}(\bm{w})\le 0\right\}$, 
indexed by $[\begin{array}{ccc} \delta & \tilde{\bm{\gamma}}^{\top} & \tilde{\bm{z}}^{\top}\end{array}]\in\mathbb{R}^{2l-1}$, is a Donsker class.  
Lemma~\ref{prelim} follows immediately from 
\citet[Theorem~2.1]{vanderVaartWellner07}.
\end{proof}

Recall the definition of $\hat{W}_i$ given above in (\ref{what}) and define $W_i\equiv D_i\left(Y_i-\bm{X}^{\top}_i\bm{\beta}_0\right)$ for each $i\in\{1,\ldots,n\}$.  
Recall in addition the definitions of $\bm{S}_i$ and $K_i$ given above in (\ref{si}) and (\ref{ki}), respectively.  The proposed estimator given above 
in (\ref{mhat}) may be rewritten as
\begin{eqnarray}
 & &\hat{\theta}_n\nonumber\\
 &=&\bm{e}^{\top}_1\left(\sum_{i=1}^n\bm{S}_i K_i\bm{S}_i^{\top}\right)^{-1}\sum_{i=1}^n\bm{S}_iK_i W_i\nonumber\\
 & &-\bm{e}^{\top}_1\left(\sum_{i=1}^n\bm{S}_i K_i\bm{S}_i^{\top}\right)^{-1}\sum_{i=1}^n\bm{S}_iK_i\bm{X}^{\top}_i\left(\hat{\bm{\beta}}_n-\bm{\beta}_0\right)\nonumber\\
	&\equiv & \hat{m}_{n1}(1)+\hat{m}_{n2}(1).\label{thetahat}
\end{eqnarray}

Consider $\hat{m}_{n1}(1)$.  For each $i\in\{1,\ldots,n\}$ we can write
\begin{eqnarray}
W_i &=& m_{F_0}\left(\hat{\eta}_n(\bm{Z}_i)\right)+\left(W_i-m_{F_0}\left(\hat{\eta}_n(\bm{Z}_i)\right)\right)\nonumber\\
	&\equiv & m_{F_0}\left(\hat{\eta}_n(\bm{Z}_i)\right)+\zeta_{ni}\nonumber\\
	&=& m_{F_0}(1)+\sum_{j=1}^{p-1}\frac{1}{j!}\left(\hat{\eta}_n(\bm{Z}_i)-1\right)^j m^{(j)}_{F_0}(1)\nonumber\\
	& &+\frac{1}{p!}\left(\hat{\eta}_n(\bm{Z}_i)-1\right)^{p}\cdot m^{(p)}_{F_0}\left(\hat{\eta}_{ni}^*\right)+\zeta_{ni}\nonumber\\
	&\equiv & \bm{S}^{\top}_i\left[\begin{array}{c} \theta_0 \\ \theta_0^{(1)}\end{array}\right]+\sum_{j=2}^{p-1}\frac{1}{j!}\left(\hat{\eta}_n(\bm{Z}_i)-1\right)^j m^{(j)}_{F_0}(1)\nonumber\\
	& &+\frac{1}{p!}\left(\hat{\eta}_n(\bm{Z}_i)-1\right)^{p}\bar{\theta}_{ni}^{(p)}+\zeta_{ni},\label{Wi}
\end{eqnarray}
where $m^{(j)}_{F_0}(1)$ for each $j\in\{1,\ldots,p\}$ denotes the left-hand limit of $m^{(j)}_{F_0}(q)$ as $q\uparrow 1$, i.e., $m^{(j)}_{F_0}(1)=\lim_{q\uparrow 1}\left.\left(d^j/dq^j\right) m_{F_0}(q^{\prime})\right|_{q^{\prime}=q}$.  
In addition, for each $i\in\{1,\ldots,n\}$, $\hat{\eta}_{ni}^*$ denotes a point between $\hat{\eta}_n(\bm{Z}_i)$ and one. 	It follows that
\begin{eqnarray}
\hat{m}_{n1}(1) &=& \theta_0+\bm{e}^{\top}_1\left(\sum_{i=1}^n\bm{S}_i K_i\bm{S}_i^{\top}\right)^{-1}\sum_{i=1}^n\bm{S}_iK_i\sum_{j=2}^{p-1}\frac{1}{j!}\left(\hat{\eta}_n(\bm{Z}_i)-1\right)^j m^{(j)}_{F_0}(1)\nonumber\\
	& &+\frac{1}{p!}\bm{e}^{\top}_1\left(\sum_{i=1}^n\bm{S}_i K_i\bm{S}_i^{\top}\right)^{-1}\sum_{i=1}^n \bm{S}_i K_i\left(\hat{\eta}_n(\bm{Z}_i)-1\right)^{p}\bar{\theta}_{ni}^{(p)}\nonumber\\
	& &+\bm{e}^{\top}_1\left(\sum_{i=1}^n\bm{S}_i K_i\bm{S}_i^{\top}\right)^{-1}\sum_{i=1}^n\bm{S}_i K_i\zeta_{ni}.\label{mhat1}
\end{eqnarray}
Observe that for each $i\in\{1,\ldots,n\}$,
\begin{eqnarray}
 & &\hat{\eta}_n(\bm{Z}_i)-1\nonumber\\
 &=&\eta_0(\bm{Z}_i)-1+\left[\left(\hat{\eta}_n(\bm{Z}_i)-\eta_n(\bm{Z}_i)\right)-\left(\hat{\eta}_0(\bm{Z}_i)-\eta_0(Z_i)\right)\right]\nonumber\\
 & &+\left(\eta_n(\bm{Z}_i)-\eta_0(\bm{Z}_i)\right)+\left(\hat{\eta}_0(\bm{Z}_i)-\eta_0(\bm{Z}_i)\right)\nonumber\\
	&\equiv & \eta_0(\bm{Z}_i)-1+R_{ni1}+R_{ni2}+R_{ni3}.\label{rns}
\end{eqnarray}
Notice that 
\begin{equation}\label{setupstarp3}
\max_{1\le i\le n}\left|R_{ni1}\right|=o_p\left(n^{-\frac{1}{2}}\right)
\end{equation} 
by Lemma~\ref{prelim} and that 
\begin{equation}\label{setup2starp3}
\max_{1\le i\le n}\left|R_{ni2}\right|=O_p\left(n^{-\frac{1}{2}}\right)
\end{equation} 
given the $\sqrt{n}$-consistency of $\hat{\bm{\gamma}}_n$ and the assumption that $F_0(\cdot)$ has a bounded derivative on the 
support of $\bm{Z}^{\top}\bm{\gamma}_0$.  Finally, we have 
\begin{equation}\label{setup3starp3}
\max_{1\le i\le n}\left|R_{ni3}\right|=O_p\left(n^{-\frac{1}{2}}\right)
\end{equation} 
by Donsker's theorem.

Now consider $\left(nh_n\right)^{-1}\sum_{i=1}^n\bm{S}_i K_i\bm{S}_i^{\top}$.  We have
\begin{eqnarray}
 & &\frac{1}{nh_n}\sum_{i=1}^n\bm{S}_i K_i\bm{S}_i^{\top}\nonumber\\
 &=&\frac{1}{nh_n}\sum_{i=1}^n\nonumber\\
 & &\left[\begin{array}{cc} K\left(\frac{1}{h_n}\left(\hat{\eta}_n(\bm{Z}_i)-1\right)\right) & \left(\hat{\eta}_n(\bm{Z}_i)-1\right)K\left(\frac{1}{h_n}\left(\hat{\eta}_n(\bm{Z}_i)-1\right)\right) \\ \left(\hat{\eta}_n(\bm{Z}_i)-1\right)K\left(\frac{1}{h_n}\left(\hat{\eta}_n(\bm{Z}_i)-1\right)\right) & \left(\hat{\eta}_n(\bm{Z}_i)-1\right)^2 K\left(\frac{1}{h_n}\left(\hat{\eta}_n(\bm{Z}_i)-1\right)\right)\end{array}\right].\nonumber\\
 & &\label{lldagp1}
\end{eqnarray}
We have for each $i\in\{1,\ldots,n\}$ that
 \begin{eqnarray}
  & & K_i\nonumber\\
  &=&K\left(\frac{1}{h_n}\left(\eta_0(\bm{Z}_i)-1\right)\right)\nonumber\\
  & &+\frac{1}{h_n}\left(R_{ni1}+R_{ni2}+R_{ni3}\right) K^{(1)}\left(\frac{1}{h_n}\left(\eta_0(\bm{Z}_i)-1\right)\right)\nonumber\\
  & &+\frac{1}{h_n^2}\left(R_{ni1}+R_{ni2}+R_{ni3}\right)^2 K^{(2)}\left(\Delta_{ni}\right),\label{ll2dagp1}
  \end{eqnarray}
 where $R_{ni1}$, $R_{n2}$ and $R_{ni3}$ are as given above in (\ref{rns}), and where $\Delta_{ni}$ is a point between 
 $h_n^{-1}\left(\hat{\eta}_n(\bm{Z}_i)-1\right)$ and $h_n^{-1}\left(\eta_0(\bm{Z}_i)-1\right)$.  It follows 
 that
  \begin{eqnarray}
  & &\frac{1}{nh_n}\sum_{i=1}^n K_i\nonumber\\
  &=&\frac{1}{nh_n}\sum_{i=1}^n K\left(\frac{1}{h_n}\left(\eta_0(\bm{Z}_i)-1\right)\right)\nonumber\\
  & &+\frac{1}{nh_n^2}\sum_{i=1}^n\left(R_{ni1}+R_{ni2}+R_{ni3}\right) K^{(1)}\left(\frac{1}{h_n}\left(\eta_0(\bm{Z}_i)-1\right)\right)\nonumber\\
  & &+\frac{1}{2nh_n^3}\sum_{i=1}^n\left(R_{ni1}+R_{ni2}+R_{ni3}\right)^2 K^{(2)}\left(\Delta_{ni}\right)\nonumber\\
  &=&\frac{1}{nh_n}\sum_{i=1}^n K\left(\frac{1}{h_n}\left(\eta_0(\bm{Z}_i)-1\right)\right)+O_p\left(\frac{1}{\sqrt{n}}\right)+O_p\left(\frac{1}{nh_n^3}\right)\nonumber\\
  &=&\frac{1}{nh_n}\sum_{i=1}^n K\left(\frac{1}{h_n}\left(\eta_0(\bm{Z}_i)-1\right)\right)+o_p(1)\label{llstarp1}
  \end{eqnarray}
given the results (\ref{setupstarp3})--(\ref{setup3starp3}) above and the assumptions that $nh_n^3\to\infty$ and that $K^{(2)}(\cdot)$ is 
bounded. 

From (\ref{lldagp1}) and (\ref{llstarp1}) one can use standard calculations \citep[e.g.,][]{RuppertWand94} to deduce that
\[
\frac{1}{nh_n}\sum_{i=1}^n \bm{S}_i K_i\bm{S}_i^{\top}=\left[\begin{array}{cc} 1 & 0 \\ 0 & h_n^{p}\int u^{p} K(u) du\end{array}\right]+o_p(1),
\]
which implies that
\begin{equation}\label{ll3dagp1}
\bm{e}^{\top}_1\left(\frac{1}{nh_n}\sum_{i=1}^n\bm{S}_i K_i\bm{S}_i^{\top}\right)^{-1}=\bm{e}^{\top}_1+o_p(1).
\end{equation}
Similar calculations yield
\begin{eqnarray} 
 & &\frac{1}{nh_n}\sum_{i=1}^n \bm{S}_i K_i\left(\hat{\eta}_n(\bm{Z}_i)-1\right)^j m^{(j)}_{F_0}(1)\nonumber\\
 &=&\frac{1}{nh_n}\sum_{i=1}^n\left[\begin{array}{c} \left(\hat{\eta}_n(\bm{Z}_i)-1\right)^j K_i\cdot m^{(j)}_{F_0}(1)\\ \left(\hat{\eta}_n(\bm{Z}_i)-1\right)^{j+1} K_i\cdot m^{(j)}_{F_0}(1)\end{array}\right]\nonumber\\
 &=&\left[\begin{array}{c} 0 \\ h_n^{j+1}\int u^{j+1} K(u) du\cdot m^{(j)}_{F_0}(1)\end{array}\right]+o_p(1)\label{ll2dagp2}
\end{eqnarray}
for each $j\in \{2,\ldots,p-1\}$, and
\begin{eqnarray} 
 & &\frac{1}{nh_n}\sum_{i=1}^n \bm{S}_i K_i\left(\hat{\eta}_n(\bm{Z}_i)-1\right)^{p} m^{(p)}_{F_0}(1)\nonumber\\
 &=&\frac{1}{nh_n}\sum_{i=1}^n\left[\begin{array}{c} \left(\hat{\eta}_n(\bm{Z}_i)-1\right)^{p} K_i \bar{\theta}_{ni}^{(p)}\\ \left(\hat{\eta}_n(\bm{Z}_i)-1\right)^{p+1} K_i \bar{\theta}_{ni}^{(p)}\end{array}\right]\nonumber\\
	&=&\left[\begin{array}{c} h_n^{p}\int u^{p} K(u) du\cdot m^{(p)}_{F_0}(1) \\ h_n^{p+1}\int u^{p+1} K(u) du\cdot m^{(p+1)}_{F_0}(1)\end{array}\right]+o_p(1).\label{lldagp2}
\end{eqnarray}
Combining (\ref{ll3dagp1}) with (\ref{ll2dagp2}) and (\ref{lldagp2}) we get
\begin{eqnarray}
 & &\bm{e}^{\top}_1\left(\sum_{i=1}^n\bm{S}_i K_i\bm{S}_i^{\top}\right)^{-1}\sum_{i=1}^n\bm{S}_i K_i\sum_{j=2}^{p-1}\frac{1}{j!}\left(\hat{\eta}_n(\bm{Z}_i)-1\right)^j m^{(j)}_{F_0}(1)\nonumber\\
 & &+\frac{1}{p!}\bm{e}^{\top}_1\left(\sum_{i=1}^n\bm{S}_i K_i\bm{S}_i^{\top}\right)^{-1}\sum_{i=1}^n\bm{S}_i K_i\left(\hat{\eta}_n(\bm{Z}_i)-1\right)^{p+1}\bar{\theta}_{ni}^{(p)}\nonumber\\
 &=&\frac{h_n^{p}}{p!}\int u^{p} K(u) du \cdot m^{(p)}_{F_0}(1)+o_p(1).\label{ll2starp2}
\end{eqnarray}

Now consider $\zeta_{ni}=W_i-m_{F_0}\left(\hat{\eta}_n(\bm{Z}_i)\right)$ for each $i\in\{1,\ldots,n\}$.  We have 
\begin{eqnarray*}
 & & m_{F_0}\left(\hat{\eta}_n(\bm{Z}_i)\right)\\
 &=& m_{F_0}\left(\eta_0(\bm{Z}_i)\right)\\
 & &+\left\{\left[\left(\eta_n(\bm{Z}_i)-\eta_n(\bm{Z}_i)\right)+\left(\hat{\eta}_0(\bm{Z}_i)-\eta_0(\bm{Z}_i)\right)\right]+\left(\eta_n(\bm{Z}_i)-\eta_0(\bm{Z}_i)\right)\right.\\
 & &\left.+\left(\hat{\eta}_0(\bm{Z}_i)-\eta_0(\bm{Z}_i)\right)\right\}\cdot m^{(1)}_{F_0}\left(\hat{\eta}_{ni}^{**}\right)\\
 &=&m_{F_0}\left(\eta_0(\bm{Z}_i)\right)+\left(R_{ni1}+R_{ni2}+R_{ni3}\right)\cdot m^{(1)}_{F_0}\left(\hat{\eta}_{ni}^{**}\right),
\end{eqnarray*}
where $\hat{\eta}_{ni}^{**}$ is an intermediate value.  It follows that
\begin{eqnarray}
\zeta_{ni} &=& W_i-m_{F_0}\left(\eta_0(\bm{Z}_i)\right)+\left(R_{ni1}+R_{ni2}+R_{ni3}\right)\cdot m^{(1)}_{F_0}\left(\hat{\eta}_{ni}^{**}\right)\nonumber\\
	&\equiv & \zeta_i+\left(R_{ni1}+R_{n2}+R_{ni3}\right)\bar{\theta}_{ni}^{(1)},\label{zetai}
\end{eqnarray}
and so	
\begin{eqnarray}
 & &\frac{1}{\sqrt{nh_n}}\sum_{i=1}^n \bm{S}_i K_i\zeta_{ni}\nonumber\\
 &=&\frac{1}{\sqrt{nh_n}}\sum_{i=1}^n \bm{S}_i K_i\zeta_i+\frac{1}{\sqrt{nh_n}}\sum_{i=1}^n\left[\begin{array}{c} K_i\left(R_{ni1}+R_{n2}+R_{ni3}\right)\bar{\theta}_{ni}^{(1)} \\ K_i\left(\hat{\eta}_n(\bm{Z}_i)-1\right)\left(R_{ni1}+R_{n2}+R_{ni3}\right)\bar{\theta}_{ni}^{(1)}\end{array}\right].\nonumber\\
 & &\label{llstarp3}
 \end{eqnarray}
 
Recall that identification of $\bm{\gamma}_0$ subject to the conditions of Assumption~\ref{a1} and the smoothness conditions in 
Assumptions~\ref{a2}.\ref{a2d1}--\ref{a2}.\ref{a2d2} jointly imply that $m^{(1)}_{F_0}(q)$ is  bounded for all $q\in(0,1)$.  It follows that there exists 
a constant $C_1\in (0,\infty)$ such that
 \begin{eqnarray*}
 \left|\frac{1}{\sqrt{nh_n}}\sum_{i=1}^n K_i\left(R_{ni1}+R_{ni2}+R_{ni3}\right)\bar{\theta}_{ni}^{(1)}\right| & \le &\frac{1}{\sqrt{nh_n}}\cdot C_1 n^{-\frac{1}{2}}\cdot\left(nh_n\right)\cdot\frac{1}{nh_n}\sum_{i=1}^n \left|K_i\right|\\
	&=&O_p\left(\sqrt{h_n}\right)\\
	&=& o_p(1).
\end{eqnarray*}
Similar calculations show that the second component of the second term in (\ref{llstarp3}) is $o_p(1)$.  It follows that
\begin{equation}\label{lldagp3}
\frac{1}{\sqrt{nh_n}}\sum_{i=1}^n \bm{S}_i K_i\zeta_{ni}=\frac{1}{\sqrt{nh_n}}\sum_{i=1}^n \bm{S}_i K_i\zeta_i+o_p(1).
\end{equation}
Combining (\ref{ll3dagp1}), (\ref{ll2starp2}) and (\ref{lldagp3}) yields
\begin{equation}\label{ll2dagp3}
\hat{m}_{n1}(1)=\theta_0+\frac{h_n^{p}}{p!}\int u^{p} K(u) du\cdot m^{(p)}_{F_0}(1)+\bm{e}^{\top}_1\cdot\frac{1}{nh_n}\sum_{i=1}^n\bm{S}_i K_i\zeta_i+o_p(1).
\end{equation}
Exploiting the decomposition in (\ref{ll2dagp1}) produces the result
\begin{eqnarray}
 & &\bm{e}^{\top}_1\cdot\frac{1}{nh_n}\sum_{i=1}^n\bm{S}_i K_i\zeta_i\nonumber\\
 &=&\frac{1}{nh_n}\sum_{i=1}^n\zeta_i\nonumber\\
 &=&\frac{1}{nh_n}\sum_{i=1}^n K\left(\frac{1}{h_n}\left(\eta_0(\bm{Z}_i)-1\right)\right)\zeta_i+O_p\left(n^{-\frac{1}{2}}\right)+O_p\left(\frac{1}{nh_n^{3}}\right)\nonumber\\
 &=&\frac{1}{nh_n}\sum_{i=1}^n K\left(\frac{1}{h_n}\left(\eta_0(\bm{Z}_i)-1\right)\right)\zeta_i+o_p(1)\label{ll2starp3}
\end{eqnarray}
under the condition that $nh_n^{3}\to\infty$.  It follows from (\ref{ll2dagp3}) and (\ref{ll2starp3}) that
\begin{eqnarray}
 & &\hat{m}_{n1}(1)\nonumber\\
 &=&\theta_0+\frac{h_n^{p}}{p!}\int u^{p} K(u) du\cdot m^{(p)}_{F_0}(1)\nonumber\\
 & &+\frac{1}{nh_n}\sum_{i=1}^n K\left(\frac{1}{h_n}\left(\eta_0(\bm{Z}_i)-1\right)\right)\zeta_i+o_p(1).\label{llstarp4}
\end{eqnarray}

Next, consider the term $\hat{m}_{n2}(1)=-\bm{e}^{\top}_1\left(\sum_{i=1}^n\bm{S}_i K_i\bm{S}_i^{\top}\right)^{-1}\sum_{i=1}^n\bm{S}_i K_i\bm{X}^{\top}_i\left(\hat{\bm{\beta}}_n-\bm{\beta}_0\right)$.  
We have
\[
\bm{e}^{\top}_1\cdot\frac{1}{nh_n}\sum_{i=1}^n\bm{S}_i K_i\bm{X}^{\top}_i\left(\hat{\bm{\beta}}_n-\bm{\beta}_0\right)=\frac{1}{nh_n}\sum_{i=1}^n K_i\bm{X}^{\top}_i\left(\hat{\bm{\beta}}_n-\bm{\beta}_0\right),
\]
where
\begin{equation}\label{ll2dagp4}
\left|\frac{1}{nh_n}\sum_{i=1}^n K_i\bm{X}^{\top}_i\left(\hat{\bm{\beta}}_n-\bm{\beta}_0\right)\right|\le\frac{1}{nh_n}\sum_{i=1}^n\left|K_i\right|\left\|\bm{X}_i\right\|\cdot\left\|\hat{\bm{\beta}}_n-\bm{\beta}_0\right\|=O_p\left(n^{-\frac{1}{2}}\right),
\end{equation}
and where the decomposition appearing above in (\ref{ll2dagp1}) has been applied, along with the assumptions that 
$E\left[\left\|\bm{X}_1\right\|\right]<\infty$ 
and $\left\|\hat{\bm{\beta}}_n-\bm{\beta}_0\right\|=O_p\left(n^{-1/2}\right)$.  Combining (\ref{ll3dagp1}) with (\ref{ll2dagp4}) yields the result
\begin{equation}\label{ll2starp4}
\hat{m}_{n2}(1)=O_p\left(n^{-\frac{1}{2}}\right),
\end{equation}
while combining (\ref{ll2starp4}) with (\ref{llstarp4}) produces
\begin{eqnarray*}
 & &\hat{\theta}_n\\
 &=&\theta_0+\frac{h_n^{p}}{p!}\int u^{p} K(u) du\cdot m^{(p)}_{F_0}(1)\\
 & &+\frac{1}{nh_n}\sum_{i=1}^n K\left(\frac{1}{h_n}\left(\eta_0(\bm{Z}_i)-1\right)\right)\zeta_i+o_p(1).
\end{eqnarray*}
It follows that 
\begin{eqnarray}
 & &\sqrt{nh_n}\left(\hat{\theta}_n-\theta_0-\frac{h_n^{p}}{p!}\int u^{p} K(u) du\cdot m^{(p)}_{F_0}(1)\right)\nonumber\\
 &=&\frac{1}{\sqrt{nh_n}}\sum_{i=1}^n K\left(\frac{1}{h_n}\left(\eta_0(\bm{Z}_i)-1\right)\right)\zeta_i+o_p(1),\label{ll3dagp4}
\end{eqnarray}
where the leading term is asymptotically normal mean-zero with variance
\begin{eqnarray}
 & &\frac{1}{h_n} E\left[K^2\left(\frac{1}{h_n}\left(\eta_0(\bm{Z}_1)-1\right)\right)\zeta_1^2\right]\\
 &=&\frac{1}{h_n}E\left[ K^2\left(\frac{1}{h_n}\left(\eta_0(\bm{Z}_1)-1\right)\right)E\left[\left.\zeta^2_1\right|\eta_0(\bm{Z}_1)\right]\right]\\
 &\to & E\left[\left.U^2_1\right|F_0\left(\bm{Z}^{\top}\bm{\gamma}_0\right)=1\right]\int K^2(u) du\\
 &=& \sigma^2_{U|F_0\left(\bm{Z}^{\top}\bm{\gamma}_0\right)}\left(1\right)\int K^2(u) du\label{asyvarleading}
\end{eqnarray}
The conclusion of Theorem~\ref{mainthm} is immediate.

\subsection{Proof of Theorem~\ref{ubthm}}\label{ubthmpf}

\noindent For any $s>0$ and $t\in\mathbb{R}$, define $L_s(t)\equiv 1\{|t|>s\}$.  Let 
$\bm{\psi}_{10}\equiv (\theta_0,\bm{\beta}_0^{\top},\bm{\gamma}_0^{\top})$ denote a point in $\mathbb{R}^{1+k+l}$, and let 
$\bm{\psi}_{1n}$ denote a generic vector in the corresponding set $\Psi^*_{1n}\equiv\Theta_n\times B_n\times\Gamma_n$, where 
\begin{eqnarray}
\Theta_n &=& \left\{\theta\in\mathbb{R}:\,n^{\frac{p}{2p+1}}|\theta-\theta_0|\le\kappa_1\right\},\label{Thetan}\\
B_n &=&\left\{\bm{\beta}\in\mathbb{R}^k:\,\sqrt{n}\left\|\bm{\beta}-\bm{\beta}_0\right\|\le\kappa_2\right\},\label{Bn}\\
\Gamma_n &=&\left\{\bm{\gamma}\in\mathbb{R}^l:\,\sqrt{n}\left\|\bm{\gamma}-\bm{\gamma}_0\right\|\le\kappa_3\right\}\label{Gamman}
\end{eqnarray}
for some positive constants $\kappa_1$, $\kappa_2$ and $\kappa_3$.  Let $g_{\psi_{2n}}$ denote a joint conditional density for $(U,V)$ 
given $\bm{X}$ and $\bm{Z}$ lying on some curve in a shrinking neighbourhood $\Psi^*_{2n}$ of a bivariate density $g_0$ satisfying 
all the relevant conditions of Assumptions~\ref{a1} and \ref{a2} for a conditional density of $(U,V)$ given $\bm{X}=\bm{x}$ and $\bm{Z}=\bm{z}$, and such that $g_{\psi_{2n0}}=g_0$ for some $\psi_{2n0}\in\Psi^*_{2n}$.  Let 
$E_{\bm{\psi}_{1n},g_{\psi_{2n}}}[\cdot]$ denote expectation under the corresponding point 
$(\bm{\psi}_{1n},g_{\psi_{2n}})\in\Psi_n$.  Begin by noting that (\ref{ub1b}) and (\ref{ub2b}) may be rewritten as
\begin{equation}\label{ub1c}
\liminf_{n\to\infty}\sup_{\bm{\psi}_{1n}\in\Psi^*_{1n},\,\psi_{2n}\in\Psi^*_{2n}} E_{\bm{\psi}_{1n},g_{\psi_{2n}}}\left[L_s\left(n^{\frac{p}{2p+1}}\left(\theta_n-\theta\right)\right)\right]>0
\end{equation}
and
\begin{equation}\label{ub2c}
\lim_{s\to 0}\liminf_{n\to\infty}\sup_{\bm{\psi}_{1n}\in\Psi^*_{1n},\,\psi_{2n}\in\Psi^*_{2n}} E_{\bm{\psi}_{1n},g_{\psi_{2n}}}\left[L_s\left(n^{\frac{p}{2p+1}}\left(\theta_n-\theta\right)\right)\right]=1,
\end{equation}
respectively.  

Consider the generalization of the H\'{a}jek--Le Cam asymptotic minimax theorem \citep[e.g.,][Theorem~12.1]{IbragimovHasminskii81} 
given in \citet[inequality (II.12.18)]{IbragimovHasminskii81}.  One can deduce from 
\citet[inequality (II.12.18)]{IbragimovHasminskii81} that (\ref{ub1c})--(\ref{ub2c}) hold, thus implying 
(\ref{ub1b})--(\ref{ub2b}) and (\ref{ub1})--(\ref{ub2}), if for some $\bm{\psi}_{10}\in\mathbb{R}^{1+k+l}$ and $\psi_{20}\in\mathbb{R}$, there 
exists a parametrization $\psi_{2n}\to g_{\psi_{2n}}$ on $\Psi^*_{2n}$ with $g_{\psi_{20}}=g_0$ such that 
the conditional joint distribution of $(D,Y)$ given $\bm{X}$ and $\bm{Z}$ is locally asymptotically normal (LAN) at 
$[\begin{array}{cc}\bm{\psi}^{\top}_{10} & \psi_{20}\end{array}]$ in the sense of Condition~\ref{lancond} given below.  

Let $\psi_{2n}\to g_{\psi_{2n}}$ be such a parametrization of the conditional joint density of $[\begin{array}{cc} U & V\end{array}]$ given 
$\bm{X}$ and $\bm{Z}$, and let $l\left(\left.\bm{\psi}_{1n},\psi_{2n};D,Y\right|\bm{X},\bm{Z}\right)$ denote the conditional log-likelihood 
of $(D,Y)$ given $\bm{X}$ and $\bm{Z}$ evaluated at the point 
$(\bm{\psi}_{1n},g_{\psi_{2n}})\in\Psi^*_{1n}\times\Psi^*_{2n}$.  Let 
\[
\left\{(D_i,Y_i,\bm{X}^{\top}_i,\bm{Z}^{\top}_i):\,i=1,\ldots,n\right\}
\]
denote a sample of ordered $(2+k+l)$-tuples generated by (\ref{outcome})--(\ref{observation}).  The LAN condition is specified as follows:

\begin{condition}[LAN]\label{lancond}
The point $(\bm{\psi}_{10},g_{0})$, where 
$\bm{\psi}_{10}=(\theta_0,\bm{\beta}_0^{\top},\bm{\gamma}_0^{\top})^{\top}\in\mathbb{R}^{1+k+l}$, identifies the 
conditional joint distribution of each $(D_i,Y_i)$ given $\bm{X}_i$ and $\bm{Z}_i$.  In addition, the point
$\bm{\psi}_{1n}\equiv (\theta_n,\bm{\beta}_n^{\top},\bm{\gamma}_n^{\top})^{\top}$ is such that 
$n^{p/(2p+1)}\left(\theta_n-\theta_0\right)\to\omega^{p^*}_1$ for some $\omega_1\neq 0$, where $p^*\ge 3$ is the odd integer specified above in Assumption~\ref{a2}.\ref{a2d11}.  
In addition, $\sqrt{n}[(\bm{\beta}_n-\bm{\beta}_0)^{\top},(\bm{\gamma}_n-\bm{\gamma}_0)^{\top}]^{\top}\to\bm{\omega}_2$ 
for some $\bm{\omega}_2\neq\bm{0}$, while $\psi_{2n}$ is such that $\sqrt{n}\psi_{2n}\to\omega_3$ for some constant 
$\omega_3\neq 0$.

There exists a random $(2+k+l)$-vector $\bm{S}_{n0}$ and a $[(2+k+l)\times (2+k+l)]$-matrix $\bm{I}_0$ of full rank such that the conditional 
distribution $\left.\bm{S}_{n0}\right|(\bm{X}^{\top},\bm{Z}^{\top})\stackrel{d}{\to} N_{2+k+l}\left(\bm{0},\bm{I}_0\right)$ and
\begin{eqnarray*}
 & &\sum_{i=1}^n\left( l\left(\left.\bm{\psi}_{1n},\psi_{2n};D_i,Y_i\right|\bm{X}_i,\bm{Z}_i\right)-l\left(\left.\bm{\psi}_{10},0;Y_i,D_i\right|\bm{X}_i,\bm{Z}_i\right)\right)\\
 &=&[\begin{array}{ccc} \omega_1^{p^*} &\bm{\omega}_2 & \omega_3\end{array}]\bm{S}_{n0}-\frac{1}{2}[\begin{array}{ccc} \omega_1^{p^*} &\bm{\omega}_2 & \omega_3\end{array}]\bm{I}_0\left[\begin{array}{c} \omega_1^{p^*} \\ \bm{\omega}_2 \\ \omega_3\end{array}\right]+o_p(1).
\end{eqnarray*}
\end{condition}

It follows that Theorem~\ref{ubthm} is proved if for a given $\bm{\psi}_{10}\in\mathbb{R}^{1+k+l}$, one can exhibit a 
parametrization $\psi_{2n}\to g_{\psi_{2n}}$ on a shrinking neighbourhood $\Psi^*_{2n}$ of $g_0$ such that  
the corresponding conditional log likelihood of $(D,Y)$ given $[\begin{array}{cc} \bm{X}^{\top} & \bm{Z}^{\top}\end{array}]$ satisfies 
Condition~\ref{lancond} at the point $[\begin{array}{cc}\bm{\psi}^{\top}_{10} & 0\end{array}]$.

In this connection consider arbitrary points $\theta, \theta_0\in\mathbb{R}$ and $\bm{\beta},\bm{\beta}_0\in\mathbb{R}^k$.  Let $\delta_1\equiv\theta-\theta_0$ and $\bm{\delta}_2\equiv\bm{\beta}-\bm{\beta}_0$, and let 
$\Delta_{1n}$ and $\Delta_{2n}$ denote neighbourhoods of the origin given by
\begin{eqnarray}
\Delta_{1n} &=& \left\{\delta_1:\,n^{\frac{p}{2p+1}}\left|\delta_1\right|\le \kappa_1\right\},\label{Delta1n}\\
\Delta_{2n} &=& \left\{\bm{\delta}_2:\,n^{\frac{1}{2}}\left\|\bm{\delta}_2\right\|\le \kappa_2\right\}\label{Delta2n}
\end{eqnarray}
for positive constants $\kappa_1$ and $\kappa_2$.  Next, let $g_{\bm{0}U|V,\bm{X},\bm{Z}}(\cdot|\cdot)$ denote a conditional joint density of $U$ given 
$[\begin{array}{ccc} V & \bm{X}^{\top} & \bm{Z}^{\top}\end{array}]$ that satisfies all relevant conditions of Assumptions~\ref{a1} and \ref{a2}.  Let $\eta_1(u|\bm{x},\bm{z})$ be a non-constant measurable function such 
that 
\[
E\left[\left.\eta_1\left(\left.U\right|\bm{X},\bm{Z}\right)\right|D=1,\bm{X},\bm{Z}\right] = 0
\]
and 
\[
E\left[\left.\eta^2_1\left(\left.U\right|\bm{X},\bm{Z}\right)\right|D=1,\bm{X},\bm{Z}\right] < \infty
\]
with probability one.  Let $\Delta_{3n}$ be a neighbourhood of the origin on $\mathbb{R}$ given by
\begin{equation}\label{Delta3n}
\Delta_{3n} = \left\{\delta_3:\,n^{\frac{1}{2}}\left|\delta_3\right|\le \kappa_3\right\}
\end{equation}
for some positive constant $\kappa_3$.  Define the following curve on $\Delta_{1n}\times\Delta_{2n}\times\Delta_{3n}$ parameterized by 
$(\delta_1,\bm{\delta}^{\top}_2,\delta_3)$ and passing through $g_{\bm{0}U|V,\bm{X},\bm{Z}}(y-\theta_0-\bm{x}^{\top}\bm{\beta}_0|v,\bm{x},\bm{z})$:
\begin{eqnarray}
 & & g_{\bm{\delta}U|V}\left(\left.y-\theta_0-\bm{x}^{\top}\bm{\beta}_0\right|v,\bm{x},\bm{z}\right))\nonumber\\
 &=& \left(1+\delta_3\eta_1(y-\theta_0-\bm{x}^{\top}\bm{\beta}_0|\bm{x},\bm{z})\right)\nonumber\\
 & &\cdot g_{\bm{0}U|V}\left(\left.y-\theta_0-\frac{1}{p^*!}\eta_2\left(\left.y-\theta_0-\bm{x}^{\top}\beta_0\right|\bm{x},\bm{z}\right)\delta_1^{p^*}\right.\right.\nonumber\\
 & & -\bm{x}^{\top}(\bm{\beta}_0+\bm{\delta}_2)\left|v-\frac{1}{p^*!}\eta_2(y-\theta_0-\bm{x}^{\top}\beta_0|\bm{x},\bm{z})\delta_1^{p^*},\bm{x},\bm{z}\right),\label{ratestarp1}
\end{eqnarray}
where $\eta_2(u|\bm{x},\bm{z})$ is a non-constant function such that $E\left[\left.\eta^2_2\left(\left.U\right|\bm{X},\bm{Z}\right)\right|D=1,\bm{X},\bm{Z}\right]<\infty$ and where 
\[
E\left[\left.\eta_2(U|\bm{X},\bm{Z})\cdot\left.\frac{\partial^{p^*}}{\partial\delta_1^{p^*}}l(\bm{\delta})\right|_{\bm{\delta}=\bm{0}}\,\right|D=1,\bm{X},\bm{Z}\right]=0
\]
for $l(\bm{\delta})$ denoting the sub-model conditional likelihood function given below in (\ref{sublike}).

Now let $(d,y,\bm{x}^{\top},\bm{z}^{\top})\in\{0,1\}\times\mathbb{R}^{2+k+l}$ be a point in the support of $(D,Y,\bm{X}^{\top},\bm{Z}^{\top})$.  Let $\bm{\gamma},\bm{\gamma}_0\in\mathbb{R}^l$ be arbitrary points, and 
define $\bm{\delta}_4\equiv\bm{\gamma}-\bm{\gamma}_0$.  Let $\Delta_{4n}$ be a neighbourhood of the origin given by
\begin{equation}\label{Delta4n}
\Delta_{4n} = \left\{\bm{\delta}_4:\,n^{\frac{1}{2}}\left\|\bm{\delta}_4\right\|\le\kappa_4\right\}
\end{equation}
for some positive constant $\kappa_4$, and let $\Delta_n\equiv\Delta_{1n}\times\Delta_{2n}\times\Delta_{3n}\times\Delta_{4n}$ and 
$\bm{\delta}\equiv[\begin{array}{cccc}\delta_1 & \bm{\delta}_2^{\top} & \delta_3 & \bm{\delta}^{\top}_4\end{array}]^{\top}$.

Let $g_{\bm{0}V}(\cdot|\cdot)$ denote a conditional density of $V$ given $\bm{X}$ and $\bm{Z}$ satisfying all relevant conditions of Assumptions~\ref{a1} and \ref{a2}.  For a given $n$ the 
conditional log-likelihood of $(d,y)$ given $\bm{x}$ and $\bm{z}$ of the submodel indexed by $\left(\bm{\delta},g_{\bm{\delta} U|V,\bm{X},\bm{Z}}\right)$, where 
$g_{\bm{\delta} U|V,\bm{X},\bm{Z}}(y-\theta_0-\bm{x}^{\top}\bm{\beta}_0|v,\bm{x},\bm{z})$ is as given above in (\ref{ratestarp1}), is 
\begin{eqnarray}
 & & l(\bm{\delta})\nonumber\\ 
 &\equiv & l\left(\left.\bm{\delta};d,y\right|\bm{x},\bm{z}\right)\nonumber\\
	&\equiv & d\log \int_{-\infty}^{\bm{z}^{\top}\left(\bm{\gamma}_0+\bm{\delta}_4\right)} g_{\bm{\delta} U|V}\left(\left.y-\theta_0-\bm{x}^{\top}\bm{\beta}_0\right|v,\bm{x},\bm{z}\right)g_{\bm{0}V}(v|\bm{x},\bm{z})dv\nonumber\\
	& &+(1-d)\log\int_{\bm{z}^{\top}\left(\bm{\gamma}_0+\bm{\delta}_4\right)}^{\infty} g_{\bm{0}V}(v|\bm{x},\bm{z})dv.\label{sublike}
\end{eqnarray}

Observe from (\ref{sublike}) that $\left.\left(\partial^m/\partial\delta_1^m\right)l(\bm{\delta})\right|_{\bm{\delta}=\bm{0}}\equiv 0$ for each $m\in\{1,\ldots,p^*-1\}$ and all $(d,y,\bm{x}^{\top},\bm{z}^{\top})$, while 
$s^{(p^*)}_{\delta_1}(\bm{0})\equiv\left.\left(\partial^{p^*}/\partial\delta_1^{p^*}\right)l(\bm{\delta})\right|_{\bm{\delta}=\bm{0}}$ is both nonzero with positive probability and linearly independent, with probability one, 
of the submodel scores corresponding to $\bm{\delta}_2$, $\delta_3$ and $\bm{\delta}_4$.  In particular, for 
$\bm{s}_{\delta_3}(\bm{0})\equiv\left.\left(\partial/\partial\delta_3\right)l(\bm{\delta})\right|_{\bm{\delta}=\bm{0}}$ we have 
\[
E\left[\left.s^{(p^*)}_{\delta_1}(\bm{0})\right|D=1,\bm{X},\bm{Z}\right]=E\left[\left.s_{\delta_3}(\bm{0})\right|D=1,\bm{X},\bm{Z}\right]=0
\]
and 
\[
E\left[\left.s^{(p^*)}_{\delta_1}(\bm{0})s_{\delta_3}(\bm{0})\right|D=1,\bm{X},\bm{Z}\right]=0.
\]  
Similarly, for $\bm{s}_{\bm{\delta}_2}(\bm{0})\equiv\left.\left(\partial/\partial\bm{\delta}_2\right)l(\bm{\delta})\right|_{\bm{\delta}=\bm{0}}$ and 
$\bm{s}_{\bm{\delta}_4}(\bm{0})\equiv\left.\left(\partial/\partial\bm{\delta}_4\right)l(\bm{\delta})\right|_{\bm{\delta}=\bm{0}}$, one can show that
$E\left[\left.s^{(p^*)}_{\delta_1}(\bm{0})\bm{s}_{\bm{\delta}_2}(\bm{0})\right|D=1,\bm{X},\bm{Z}\right]=\bm{0}_{k\times 1}$ and 
$E\left[\left.s^{(p^*)}_{\delta_1}(\bm{0})\bm{s}_{\bm{\delta}_4}(\bm{0})\right|D=1,\bm{X},\bm{Z}\right]=\bm{0}_{l\times 1}$, which indicates that $s^{(p^*)}_{\delta_1}(\bm{0})$ is almost surely conditionally 
uncorrelated given $\bm{X}$ and $\bm{Z}$ with the submodel scores corresponding to $\bm{\delta}_2$, $\delta_3$ and $\bm{\delta}_4$. 

In what follows, Condition~\ref{lancond} is shown to apply to a condensed version of the submodel with conditional log-likelihood 
given in (\ref{sublike}).  This simplification involves assuming that the finite-dimensional nuisance parameters $\bm{\beta}_0$ and $\bm{\gamma}_0$ are known, in which case the argument $\bm{\delta}$ appearing in 
(\ref{sublike}) reduces to the ordered pair $\bm{\delta}=[\begin{array}{cc}\delta_1 & \delta_3\end{array}]^{\top}=[\begin{array}{cc}\theta-\theta_0 & \delta_3\end{array}]^{\top}$.  In addition, the set $\Delta_n$ is 
understood to have the form $\Delta_n=\Delta_{1n}\times\Delta_{3n}$, where $\Delta_{1n}$ and $\Delta_{3n}$ are as given above in (\ref{Delta1n}) and (\ref{Delta3n}), respectively.  It is  
shown that the the family of conditional joint distributions of $(d,y)$ given $(\bm{x}^{\top},\bm{z}^{\top})$ and indexed by 
$(\bm{\delta},g_{\bm{\delta}U|V})$ for $\bm{\delta}\in\Delta_n$ is LAN at the point $\bm{\delta}=\bm{0}_{2\times 1}$.  The 
argument for the general case in which Condition~\ref{lancond} is shown to apply to the conditional log-likelihood appearing 
in (\ref{sublike}) in which $\bm{\beta}_0$ and $\bm{\gamma}_0$ are both unknown is similar, although rather more notationally complex.

For $\bm{\delta}=(\delta_1,\delta_3)$ as discussed above and $(j_1,j_2)$ denoting an ordered pair of non-negative integers, 
define the derivatives $l^{(j_1,j_2)}(\bm{\delta})\equiv\left(\partial^{j_1+j_2}/\partial\delta_1^{j_1}\partial\delta_3^{j_2}\right) l(\bm{\delta})$ 
and also $l_0^{(j_1,j_2)}\equiv l^{(j_1,j_2)}(\bm{0})$, where $l(\bm{\delta})$ is now taken to be the analogue of the conditional log-likelihood given in (\ref{sublike}) corresponding to the 
submodel in which $\bm{\delta}_2$ and $\bm{\delta}_4$ are both fixed.  Suppose $\delta_{1n}$ and $\delta_{3n}$ 
are such that $[\begin{array}{cc}(n^{p/(2p+1)}\delta_{1n})^{1/p^*} & \sqrt{n}\delta_{3n}\end{array}]\to [\begin{array}{cc} \omega_1 & \omega_3\end{array}]$ for some 
$[\begin{array}{cc} \omega_1 & \omega_3\end{array}]\neq\bm{0}$.  Let 
$\bm{\delta}_{n}\equiv [\begin{array}{cc}\delta_{1n} & \delta_{3n}\end{array}]^{\top}$.  For $j_1,j_2\ge 0$ with $j_1+j_2=2p^*+1$, define 
$R_0^{(j_1+j_2)}\equiv l^{(j_1,j_2)}(\bar{\bm{\delta}})-l_0^{(j_1,j_2)}$, for some point $\bar{\bm{\delta}}\in\Delta_n$ 
such that $\left\|\bar{\bm{\delta}}\right\|<\left\|\bm{\delta}_{n}\right\|$.  

Observe from previous discussion that $l_0^{(1,0)}=\cdots = l_0^{(p^*-1,0)}=0$.    It follows that
\begin{eqnarray*}
 & &l(\bm{\delta}_{n})-l(\bm{0})\\
 &=&\omega_1^{p^*}\left[n^{-\frac{1}{2}}\cdot\frac{l_0^{(p^*,0)}}{p^*!}+n^{-\frac{1}{2p^*}}\left(n^{-\frac{1}{2}}\cdot\frac{l_0^{(p^*+1,0)}}{(p^*+1)!}\omega_1\right)+n^{-\frac{1}{2p^*}}\left(\sum_{j_1=2}^{p^*-1} n^{-\frac{1}{2}}\frac{l_0^{(p^*+j_1,0)}}{(p^*+j_1)!}n^{\frac{1-j_1}{2p^*}}\omega_1^{j_1}\right)\right.\\
 & &\left.+n^{-1}\frac{l_0^{(2p^*,0)}}{(2p^*)!}\omega_1^{p^*}+n^{-\frac{1}{2p^*}}\left(n^{-1}\frac{l_0^{(2p^*+1,0)}}{(2p^*+1)!}\omega_1^{p^*+1}+n^{-1}\cdot\frac{R_0^{(2p^*+1,0)}}{(2p^*+1)!}\omega_1^{p^*+1}\right)\right]\\
 & &+\omega_3\left\{n^{-\frac{1}{2}}\cdot l_0^{(0,1)}+n^{-\frac{1}{2p^*}}\cdot\left[n^{-\frac{1}{2}}l_0^{(1,1)}\omega_1+\left(\sum_{j_1=2}^{p^*-1}n^{-\frac{1}{2}}\frac{l_0^{(j_1,1)}}{j_1!}n^{\frac{1-j_1}{2p^*}}\omega_1^{j_1}\right)\right]\right.\\
 & &\left.+n^{-1}\frac{l_0^{(p^*,1)}}{p^*!}\omega^{p^*}_1+n^{-\frac{1}{2p^*}}\left[n^{-1}\frac{l_0^{(p^*+1,1)}}{(p^*+1)!}\omega_1^{p^*+1}+\left(\sum_{j_1=p^*+2}^{2p^*}n^{-1}\frac{l_0^{(j_1,1)}}{j_1!}n^{\frac{1-j_1}{2p^*}}\omega_1^{j_1}\right)\right.\right.\\
 & &\left.\left.+n^{-1}\frac{R_0^{(2p^*,1)}}{(2p^*)!}n^{\frac{1-p^*}{2p^*}}\omega_1^{2p^*}\right]\right.\\
 & &\left.+n^{-1}\frac{l_0^{(0,2)}}{2}\omega_3+n^{-\frac{1}{2p^*}}\left[n^{-1}\frac{l_0^{(1,2)}}{2}\omega_1\omega_3+\left(\sum_{j_1=2}^{2p^*-1}n^{-1}\frac{l_0^{(j_1,2)}}{(2+j_1)!}n^{\frac{1-j_1}{2p^*}}\omega_1^{j_1}\omega_3{2+j_1\choose j_1}\right)\right.\right.\\
 & &\left.\left.+n^{-1}\frac{R_0^{(2p^*-1,2)}}{(2p^*+1)!}n^{\frac{2-2p^*}{2p^*}}\omega_1^{2p^*-1}\omega_3{ 2p^*+1\choose 2p^*-1}\right.\right.\\
 & &\left.\left.+\left(\sum_{m=3}^{2p^*+1}\sum_{j_1+j_2=m:\, j_2\ge 3, j_1\ge 0} n^{-1}\frac{l_0^{(j_1,j_2)}}{m!}n^{\frac{(1-j_1)+(2-j_2)p^*}{2p^*}}\omega_1^{j_1}\omega_3^{j_2-1}{m\choose j_1}\right)\right.\right.\\
 & &\left.\left.+\left(\sum_{j_1+j_2=2p^*+1:\,j_2\ge 3, j_1\ge 0} n^{-1}\frac{R_0^{(j_1,j_2)}}{(2p^*+1)!}n^{\frac{1-j_1+(2-j_2)p^*}{2p^*}}\omega_1^{j_1}\omega_3^{j_2-1}{2p^*+1\choose j_1}\right)\right]\right\}\\
 &\equiv & \omega_1^{p^*}[A_{1n}+n^{-\frac{1}{2p^*}}A_{2n}+n^{-\frac{1}{2p^*}}A_{3n}+A_{4n}+n^{-\frac{1}{2p^*}}(A_{5n}+A_{6n})]\\
 & &+\omega_3\{A_{7n}+n^{-\frac{1}{2p^*}}[A_{8n}+A_{9n}]+A_{10n}+n^{-\frac{1}{2p^*}}[A_{11n}+A_{12n}+A_{13n}]+A_{14n}\\
 & &+n^{-\frac{1}{2p^*}}[A_{15n}+A_{16n}+A_{17n}+A_{18n}+A_{19n}]\}.
\end{eqnarray*} 
Let $f_{\bm{\delta}}\left(\left.y,d\right|\bm{x},\bm{z}\right)\equiv\exp\left(l(\bm{\delta})\right)$, where $\bm{\delta}=[\begin{array}{cc} \delta_1 & \delta_3\end{array}]^{\top}$, denote the joint 
conditional density of $(Y,D)$ given $(\bm{X}^{\top},\bm{Z}^{\top})=(\bm{x}^{\top},\bm{z}^{\top})$ corresponding to the condensed version of the conditional log-likelihood in (\ref{sublike}) where 
$\bm{\delta}_2$ and $\bm{\delta}_4$ are fixed. Define
\begin{eqnarray*}
I_{011} &\equiv & E\left[\left(\left.\frac{\partial^{p^*}}{\partial\delta^{p^*}_1}\log f_{\bm{\delta}}\left(\left.Y,D\right|\bm{x},\bm{z}\right)\right|_{\bm{\delta}=\bm{0}}\right)^2\right];\\
I_{033} &\equiv & E\left[\left(\left.\frac{\partial}{\partial\delta_3}\log f_{\bm{\delta}}\left(\left.Y,D\right|\bm{x},\bm{z}\right)\right|_{\bm{\delta}=\bm{0}}\right)^2\right]
\end{eqnarray*}
and
\begin{eqnarray*}
I_{031} &\equiv & I_{013}\equiv E\left[\left.\frac{\partial^{p^*}}{\partial\delta^{p^*}_1}\log f_{\bm{\delta}}\left(\left.Y,D\right|\bm{x},\bm{z}\right)\right|_{\bm{\delta}=\bm{0}}\cdot\left.\frac{\partial}{\partial\delta_3}\log f_{\bm{\delta}}\left(\left.Y,D\right|\bm{x},\bm{z}\right)\right|_{\bm{\delta}=\bm{0}}\right],
\end{eqnarray*}
where each expectation is taken at $\bm{\delta}=\bm{0}$, and let
\[
\bm{I}_0\equiv\left[\begin{array}{cc} I_{011} & I_{013} \\ I_{031} & I_{033}\end{array}\right].
\]

For any ordered pair of nonnegative integers $(r_1,r_2)$ with $3\le r_1+r_2\le 2p^*+1$, one can exploit the form of the 
parametrization of the conditional density $g_{U|V,\bm{X},\bm{Z}}(y-\theta_0-x\beta_0|v,\bm{x},\bm{z})$ given above in (\ref{ratestarp1}) to deduce that
\[
E\left[\left(\left.\frac{\partial^{r_1+r_2}}{\partial\delta^{r_1}_1\partial\delta^{r_2}_3} f_{\bm{\delta}}\left(\left.Y,D\right|\bm{x},\bm{z}\right)\right|_{\bm{\delta}=\bm{0}}\right)^2\right]<\infty,
\]
where the expectation is also taken at $\bm{\delta}=\bm{0}$.  The parametrization of the function $g_{U|V,\bm{X},\bm{Z}}(\cdot|\cdot)$ given in 
(\ref{ratestarp1}) allows one to apply \citet[Corollary~1, p. 268]{RotnitzkyCoxBottaiRobins00} to deduce the following:
\begin{itemize}
\item $A_{3n}=O_p\left(n^{-1/(2p^*)}\right)=o_p(1)$.
\item $A_{9n}=O_p\left(n^{-1/(2p^*)}\right)=o_p(1)$.
\item $A_{12n}=O_p\left(n^{-1/(2p^*)}\right)=o_p(1)$.
\item $A_{16n}=O_p\left(n^{-1/(2p^*)}\right)=o_p(1)$.
\item $A_{18n}=O_p\left(n^{-1/p^*}\right)=o_p(1)$.
\item $A_{6n}=O_p\left(n^{-1/(2p^*)}\right)=o_p(1)$.
\item $A_{13n}=o_p\left(n^{-1/p^*}\right)=o_p(1)$.
\item $A_{17n}=o_p\left(n^{-2/p^*}\right)=o_p(1)$.
\item $A_{19n}=o_p\left(n^{-1/p^*}\right)=o_p(1)$.
\end{itemize}

Now define 
\begin{eqnarray*}
C_{011} &\equiv & E\left[\left.\frac{\partial^{p^*}}{\partial\delta^{p^*}_1}\log f_{\bm{\delta}}\left(\left.Y,D\right|\bm{x},\bm{z}\right)\right|_{\bm{\delta}=\bm{0}}\cdot\left.\frac{\partial^{p^*+1}}{\partial\delta^{p^*+1}_1}\log f_{\bm{\delta}}\left(\left.Y,D\right|\bm{x},\bm{z}\right)\right|_{\bm{\delta}=\bm{0}}\right];\\
C_{013} &\equiv & E\left[\left.\frac{\partial^{p^*}}{\partial\delta^{p^*}_1}\log f_{\bm{\delta}}\left(\left.Y,D\right|\bm{x},\bm{z}\right)\right|_{\bm{\delta}=\bm{0}}\cdot\left.\frac{\partial^{p^*-1}}{\partial\delta_1\partial\delta_3}\log f_{\bm{\delta}}\left(\left.Y,D\right|\bm{x},\bm{z}\right)\right|_{\bm{\delta}=\bm{0}}\right];\\
C_{031} &\equiv & E\left[\left.\frac{\partial}{\partial\delta_3}\log f_{\bm{\delta}}\left(\left.Y,D\right|\bm{x},\bm{z}\right)\right|_{\bm{\delta}=\bm{0}}\cdot\left.\frac{\partial^{p^*+1}}{\partial\delta^{p^*+1}_1}\log f_{\bm{\delta}}\left(\left.Y,D\right|\bm{x},\bm{z}\right)\right|_{\bm{\delta}=\bm{0}}\right];
\end{eqnarray*}
and
\begin{eqnarray*}
C_{033} &\equiv & E\left[\left.\frac{\partial}{\partial\delta_3}\log f_{\bm{\delta}}\left(\left. Y,D\right|\bm{x},\bm{z}\right)\right|_{\bm{\delta}=\bm{0}}\cdot\left.\frac{\partial^{p^*-1}}{\partial\delta_1\partial\delta_3}\log f_{\bm{\delta}}\left(\left.Y,D\right|\bm{x},\bm{z}\right)\right|_{\bm{\delta}=\bm{0}}\right],
\end{eqnarray*}
where each expectation is taken at $\bm{\delta}=\bm{0}$, and let
\[
\bm{C}_0\equiv\left[\begin{array}{cc} C_{011} & C_{013} \\ C_{031} & C_{033}\end{array}\right].
\]
A further application of \citet[Corollary~1]{RotnitzkyCoxBottaiRobins00} yields the following:
\begin{itemize}
\item $A_{4n}=\omega_1^{p^*}\left[-(1/2)\cdot I_{011}+o_p\left(n^{-1/(2p^*)}\right)\right]$.
\item $A_{5n}=\omega_1^{p^*+1}\left(-C_{011}+o_p\left(n^{-1/(2p^*)}\right)\right)$.
\item $A_{10n}=\omega_1^{p^*}\left(-I_{013}+o_p\left(n^{-1/(2p^*)}\right)\right)$.
\item $A_{11n}=\omega_1^{p^*+1}\left[-\left(C_{013}+C_{031}\right)+o_p\left(n^{-1/(2p^*)}\right)\right]$.
\item $A_{14n}=\omega_3\left[-(1/2)\cdot I_{013}+o_p\left(n^{-1/(2p^*)}\right)\right]$.
\item $A_{15n}=\omega_1\omega_3\left(-C_{033}+o_p\left(n^{-1/(2p^*)}\right)\right)$.
\end{itemize}
Collecting terms, one gets $l(\bm{\delta}_{n})-l(\bm{0})=G_{n0}\left(\omega_1^{p^*},\omega_3\right)+H_{n0}\left(\omega_1,\omega_3\right)$, 
where
\begin{eqnarray*}
G_{n0}\left(\omega_1^{p^*},\omega_3\right) &=&\omega_1^{p^*}\cdot n^{-\frac{1}{2}}\cdot\frac{l_0^{(p^*,0)}}{(p^*)!}+\omega_3\cdot n^{-\frac{1}{2}}\cdot l_0^{(0,1)}-\frac{1}{2}[\begin{array}{cc} \omega_1^{p^*} & \omega_3\end{array}]\bm{I}_0\left[\begin{array}{c} \omega_1^{p^*} \\ \omega_3\end{array}\right]\\
	&=&[\begin{array}{cc}\omega_1^{p^*} & \omega_3\end{array}]\cdot n^{-\frac{1}{2}}\left[\begin{array}{c} \frac{1}{(p^*)!}l_0^{(p^*,0)} \\ l_0^{(0,1)}\end{array}\right]-\frac{1}{2}[\begin{array}{cc} \omega_1^{p^*} & \omega_3\end{array}]\bm{I}_0\left[\begin{array}{c} \omega_1^{p^*} \\ \omega_3\end{array}\right],
\end{eqnarray*}
and where 
\[
H_{n0}\left(\omega_1,\omega_3\right)=n^{-\frac{1}{2p^*}}\omega_1\left(T_{n0}\left(\omega_1^{p^*},\omega_3\right)+o_p(1)\right)+O_p\left(n^{-\frac{1}{2}}\right),
\]
where
\[
T_{n0}\left(\omega_1^{p^*},\omega_3\right)=\omega_1^{p^*}\cdot n^{-\frac{1}{2}}\cdot\frac{l_0^{(p^*+1,0)}}{(p^*+1)!}+\omega_3\cdot n^{-\frac{1}{2}}\cdot l_0^{(1,1)}-[\begin{array}{cc} \omega_1^{p^*} & \omega_3\end{array}]\bm{C}_0\left[\begin{array}{c} \omega_1^{p^*} \\ \omega_3\end{array}\right].
\]
An application of \citet[Corollary~1]{RotnitzkyCoxBottaiRobins00} and the Cram\'{e}r--Wold device yields
\[
n^{-\frac{1}{2}}\left[\begin{array}{c} \frac{1}{p^*!}l_0^{(p^*,0)}\\ l_0^{(0,1)}\end{array}\right]\stackrel{d}{\to} N_2\left(\bm{0},\bm{I}_0\right).
\]
In addition, \citet[Corollary~1]{RotnitzkyCoxBottaiRobins00} also implies that 
\[
n^{-\frac{1}{2}}\cdot\frac{l_0^{(p^*+1,0)}}{(p^*+1)!}=O_p(1)
\]
and that
\[
n^{-\frac{1}{2}}l_0^{(1,1)}=O_p(1).
\]
It follows that $T_{n0}\left(\omega_1^{p^*},\omega_3\right)=O_p(1)$ and that $H_{n0}\left(\omega_1,\omega_3\right)=O_p\left(n^{-1/(2p^*)}\right)=o_p(1)$.

In summary, we have the first-order representation
\[
l(\bm{\delta}_{n})-l(\bm{0})=[\begin{array}{cc} \omega_1^{p^*} & \omega_3\end{array}]\cdot n^{-\frac{1}{2}}\left[\begin{array}{c}\frac{1}{p^*!}l_0^{(p^*,0)}\\ l_0^{(0,1)}\end{array}\right]-\frac{1}{2}[\begin{array}{cc}\omega_1^{p^*} & \omega_3\end{array}]\bm{I}_0\left[\begin{array}{c} \omega_1^{p^*} \\ \omega_3\end{array}\right]+o_p(1),
\]
where for $(\bm{X}^{\top},\bm{Z}^{\top})=(\bm{x}^{\top},\bm{z}^{\top})$, 
\[
n^{-\frac{1}{2}}\left[\begin{array}{c} \frac{1}{p^*!}l_0^{(p^*,0)}\\ l_0^{(0,1)}\end{array}\right]\stackrel{d}{\to} N_2\left(\bm{0},\bm{I}_0\right).
\]
It follows that the condensed version of the submodel with conditional log-likelihood given above in (\ref{sublike}) and where $\bm{\delta}_2$ and 
$\bm{\delta}_4$ are fixed is LAN at the point $[\begin{array}{cc}\delta_1 & \delta_3\end{array}]^{\top}=\bm{0}$, and as such, 
satisfies Condition~\ref{lancond}.  The general case in which $\bm{\delta}_2$ and $\bm{\delta}_4$ are unknown follows {\it mutatis mutandis}.  

\subsection{Proof of Theorem~\ref{achthm}}

\noindent The approach taken involves showing that 
$\sqrt{nh_n^*}(\hat{\theta}^*_n-\theta)=O_p(1)$ as $n\to\infty$ uniformly across 
sequences $\left\{\psi_n\right\}\equiv\left\{(\bm{\psi}_{1n},g)\right\}\subset \Psi_n$, where the set $\Psi_n$ is as given above in (\ref{Psin}).  
In particular, for each $n$, $\bm{\psi}_{1n}=(\theta,\bm{\beta}^{\top},\bm{\gamma}^{\top})^{\top}\in\mathbb{R}^{1+k+l}$, while $g$ denotes the joint 
conditional density of $(U,V)$ given $\bm{X}$ and $\bm{Z}$.
 
It suffices to show that for any $\epsilon>0$ there exists a constant $\delta(\epsilon)\in (0,\infty)$ such that
\begin{equation}\label{unifstarp1}
\lim_{n\to\infty}\sup_{\psi_n\in\Psi_n} P_{\psi_n}\left[\sqrt{nh_n^*}|\hat{\theta}_n^*-\theta|>\delta(\epsilon)\right]<\epsilon,
\end{equation}
where $P_{\psi_n}[\cdot]$ denotes probability measure under $\psi_n$.  One can show using Chebyshev's inequality that the 
following conditions suffice for (\ref{unifstarp1}) to hold with $\delta(\epsilon)=4/\epsilon$, in particular:
\begin{eqnarray}
\lim_{n\to\infty}\sqrt{nh_n^*}\sup_{\psi_n\in\Psi_n}\left|E_{\psi_n}\left[\theta^*_n\right]-\theta\right| &<& \infty;\label{unifonep1}\\
\lim_{n\to\infty} nh_n^*\cdot\sup_{\psi_n\in\Psi_n} Var_{\psi_n}\left[\theta^*_n\right] &<& \infty,\label{uniftwop1}
\end{eqnarray}
where $E_{\psi_n}[\cdot]$ and $Var_{\psi_n}[\cdot]$ respectively denote expectation and variance under a given 
$\psi_n\in\Psi_n$.

To see that (\ref{unifonep1}) and (\ref{uniftwop1}) jointly imply (\ref{unifstarp1}), note that 
$\left|\hat{\theta}_n^*-\theta\right|\le\left|\hat{\theta}^*_n-E_{\psi_n}\left[\hat{\theta}^*_n\right]\right|+\left|E_{\psi_n}\left[\hat{\theta}^*_n\right]-\theta\right|$,
so
\begin{eqnarray*}
 & &P_{\psi_n}\left[\sqrt{nh_n^*}\left|\hat{\theta}^*_n-\theta\right|>\delta(\epsilon)\right]\\
 &\le & P_{\psi_n}\left[\sqrt{nh_n^*}\left|\hat{\theta}^*_n-E_{\psi_n}\left[\hat{\theta}_n^*\right]\right|+\sqrt{nh_n^*}\left|E_{\psi_n}\left[\hat{\theta}^*_n\right]-\theta\right|>\delta(\epsilon)\right]\\
 &\le & P_{\psi_n}\left[\sqrt{nh_n^*}\left|\hat{\theta}^*_n-E_{\psi_n}\left[\hat{\theta}_n^*\right]\right|>\frac{\delta(\epsilon)}{2}\right]+P_{\psi_n}\left[\sqrt{nh_n^*}\left|E_{\psi_n}\left[\hat{\theta}^*_n\right]-\theta\right|>\frac{\delta(\epsilon)}{2}\right]\\
 &\le & \frac{4}{\delta^2(\epsilon)}Var_{\psi_n}\left[\sqrt{nh_n^*}\hat{\theta}^*_n\right]+\frac{2\sqrt{nh_n^*}\left|E_{\psi_n}\left[\hat{\theta}^*_n\right]-\theta\right|}{\delta(\epsilon)},
 \end{eqnarray*}
 from which it follows that
\begin{eqnarray*}
 & &\sup_{\psi_n\in\bar{\Psi}_n} P_{\psi_n}\left[\sqrt{nh^*_n}\left|\hat{\theta}^*_n-\theta\right|>\delta(\epsilon)\right]\\
 &\le &\frac{4nh_n^*}{\delta^2(\epsilon)}\sup_{\psi\in\Psi_n} Var_{\psi}\left[\hat{\theta}^*_n\right]+\frac{2\sqrt{nh_n^*}}{\delta(\epsilon)}\sup_{\psi_n\in\bar{\Psi}_n}\left|E_{\psi_n}\left[\hat{\theta}^*_n\right]-\theta\right|.
\end{eqnarray*}
In what follows, (\ref{unifonep1}) and (\ref{uniftwop1}) are proved in sequence.

\subsubsection{Proof of (\ref{unifonep1})}

\noindent  Recall the expression for $\hat{m}_{n1}(1)$ given above in (\ref{mhat1}).  In particular, consider the first bias 
term appearing in (\ref{mhat1}).  The assumption of a $p$th-order kernel, the uniform boundedness of 
$\hat{\eta}_n(\cdot)$ over $n\ge 1$, the boundedness of $K(\cdot)$ and of
$m^{(j)}_{F_0}(\cdot)$ for each $j\in\{0,1,\ldots,p\}$ imply via the bounded convergence 
theorem that for each sequence $\left\{\psi_n:\,\psi_n\in\Psi_n\right\}$, 
\begin{eqnarray}
 & &\sqrt{nh_n^*}\nonumber\\
 & &\cdot E_{\psi_n}\left[\bm{e}^{\top}_1\left(\sum_{i=1}^n\bm{S}_i K_i\bm{S}_i^{\top}\right)^{-1}\sum_{i=1}^n\bm{S}_iK_i\left[\sum_{j=2}^{p-1}\frac{1}{j!}\left(\hat{\eta}_n(\tilde{\bm{Z}}_i)-1\right)^j m^{(j)}_{F_0}(1)\right.\right.\nonumber\\
 & &\left.\left.+\frac{1}{p!}\left(\hat{\eta}_n(\tilde{\bm{Z}}_i)-1\right)^{p}\bar{\theta}_{ni}^{(p)}\right]\right]\nonumber\\
 &=& O(1) \label{unifstarp2}
\end{eqnarray}
in view of the assumption that 
$\sqrt{nh_n^*}\cdot (h_n^*)^{p}=\sqrt{n(h_n^*)^{2p+1}}=\sqrt{c}<\infty $.

Notice that the expectation in (\ref{unifstarp2}) does not depend on $g$, while $K_i$ is nonzero only for 
those observations $i$ such that $1-\hat{\eta}_n(\bm{Z}_i)\le h_n$.  It follows that the bound in (\ref{unifstarp2}) 
is uniform in $\Psi_n$, i.e., 
\begin{eqnarray}
 & &\lim_{n\to\infty}\sqrt{nh_n^*}\nonumber\\
 & &\cdot \sup_{\psi_n\in\Psi_n}\left|E_{\psi_n}\left[\bm{e}^{\top}_1\left(\sum_{i=1}^n\bm{S}_i K_i\bm{S}_i^{\top}\right)^{-1}\sum_{i=1}^n\bm{S}_iK_i\left[\sum_{j=2}^{p-1}\frac{1}{j!}\left(\hat{\eta}_n(\bm{Z}_i)-1\right)^j m^{(j)}_{F_0}(1)\right.\right.\right.\nonumber\\
 & &\left.\left.\left.+\frac{1}{p!}\left(\hat{\eta}_n(\bm{Z}_i)-1\right)^{p}\bar{\theta}_{ni}^{(p)}\right]\right]\right|\nonumber\\
 &<& \infty. \label{unifstarp2b}
\end{eqnarray}

Next, consider that $\left\|\hat{\bm{\beta}}_n-\bm{\beta}_0\right\|=O_p\left(n^{-1/2}\right)$ by Assumption~\ref{a2}, 
so there exists a constant $C_1\in (0,\infty)$ such that 
$\sqrt{n}\left\|\hat{\bm{\beta}}_n-\bm{\beta}_0\right\|\le C_1$ with probability 
approaching one as $n\to\infty$.  Let $A_{n1}$ denote the event in which 
$\sqrt{n}\left\|\hat{\bm{\beta}}_n-\bm{\beta}_0\right\|\le C_1$.  Let $\mathcal{A}_{n1}$ be the $\sigma$-algebra generated by 
$A_{n1}$.  We have for any sequence $\left\{\psi_n:\,\psi_n\in\Psi_n\right\}$ that
\begin{eqnarray}
 & &\sqrt{nh_n^*}\left|E_{\psi_n}\left[\left.\bm{e}^{\top}_1\left(\sum_{i=1}^n\bm{S}_i K_i\bm{S}_i^{\top}\right)^{-1}\sum_{i=1}^n\bm{S}_i K_i\bm{X}^{\top}_i\left(\hat{\bm{\beta}}_n-\bm{\beta}_0\right)\right|\mathcal{A}_{n1}\right]\right|\nonumber\\ 
 &\le & \sqrt{h^*_n} E_{\psi_n}\left[\left.\left\|\left(\sum_{i=1}^n\bm{S}_i K_i\bm{S}_i^{\top}\right)^{-1}\right\|\sum_{i=1}^n\left\|\bm{S}_i\right\|\left| K_i\right|\left\|\bm{X}_i\right\|\right|\mathcal{A}_{n1}\right]\cdot C_1 \nonumber\\
 &=& O_p\left(\sqrt{h_n^*}\right)\nonumber\\
 &=& o_p(1),\label{unifstarp3}
\end{eqnarray}
where use has been made of the assumption that $E\left[\left\|\bm{X}_1\right\|\right]<\infty$, as well as of the uniform 
boundedness of $\hat{\eta}_n(\cdot)$ over all $n$, the boundedness of $K(\cdot)$ and the bounded convergence theorem.  From (\ref{unifstarp3}) 
it follows that there exists a random variable $M_{n1}=O_p\left(\sqrt{h_n^*}\right)$ such that
\begin{eqnarray*}
 & &\sqrt{nh_n^*}\sup_{\psi_n\in\Psi_n}\left|E_{\psi_n}\left[\left.\bm{e}^{\top}_1\left(\sum_{i=1}^n\bm{S}_i K_i\bm{S}_i^{\top}\right)^{-1}\sum_{i=1}^n\bm{S}_i K_i\bm{X}^{\top}_i\left(\hat{\bm{\beta}}_n-\bm{\beta}_0\right)\right|\mathcal{A}_{n1}\right]\right|\\
 &\le & M_{n1}.
\end{eqnarray*}
But $P_{\psi_n}\left[A_{n1}\right]\to 1$ as $n\to\infty$, so
\begin{eqnarray*}
 & &\sqrt{nh_n^*}\sup_{\psi_n\in\Psi_n}\left|E_{\psi_n}\left[\bm{e}^{\top}_1\left(\sum_{i=1}^n\bm{S}_i K_i\bm{S}_i^{\top}\right)^{-1}\sum_{i=1}^n\bm{S}_i K_i\bm{X}^{\top}_i\left(\hat{\bm{\beta}}_n-\bm{\beta}_0\right)\right]\right|\\
 &=& O\left(\sqrt{h_n^*}\right)\\
 &=&o(1).
\end{eqnarray*}

Next, note that the uniform boundedness of $\hat{\eta}_n(\cdot)$ over $n$ and the boundedness of $K(\cdot)$ and 
$m^{(1)}_{F_0}(\cdot)$ 
imply via the bounded convergence theorem that
\begin{equation}\label{unifstarp4}
\sqrt{nh_n^*}E_{\psi_n}\left[\bm{e}^{\top}_1\left(\sum_{i=1}^n\bm{S}_i K_i\bm{S}_i^{\top}\right)^{-1}\sum_{i=1}^n \bm{S}_i K_i\left(\zeta_{ni}-\zeta_i\right)\right]=o(1)
\end{equation}
for all sequences $\left\{\psi_n:\,\psi_n\in\Psi_n\right\}$, where $\zeta_{ni}$ and $\zeta_i$ are as given above in (\ref{Wi}) 
and (\ref{zetai}), respectively.

The expectation in (\ref{unifstarp4}) does not depend upon the conditional joint density $g$ of $(U,V)$ given $\bm{X}$ and $\bm{Z}$, 
while $K_i$ is nonzero only for those 
observations $i$ where $1-\hat{\eta}_n(\bm{Z}_i)\le h_n$.  It follows that
\begin{equation}\label{unif2starp4}
\sqrt{nh_n^*}\sup_{\psi_n\in\Psi_n}\left|E_{\psi_n}\left[\bm{e}^{\top}_1\left(\sum_{i=1}^n\bm{S}_i K_i\bm{S}_i^{\top}\right)^{-1}\sum_{i=1}^n\bm{S}_i K_i\left(\zeta_{ni}-\zeta_i\right)\right]\right|=o(1).
\end{equation}
Next, consider that the uniform boundedness of $\hat{\eta}_n(\cdot)$ over all $n$ and the boundedness of $K(\cdot)$  
imply that there exists constants $M_{n2}, M_{n3}<(0,\infty)$ not depending on $\psi_n$ such that for every  $\psi_n\in\Psi_n$,
\begin{eqnarray}
 & &\left|E_{\psi_n}\left[\bm{e}^{\top}_1\left(\sum_{i=1}^n\bm{S}_i K_i\bm{S}_i\right)^{-1}\sum_{i=1}^n\bm{S}_i K_i\zeta_i\right]-\frac{1}{nh_n^*} E_{\psi_n}\left[\bm{e}^{\top}_1\sum_{i=1}^n\bm{S}_i K_i\zeta_i\right]\right|\nonumber\\
 &\le & M_{n2}\cdot\frac{1}{h_n^*}\left|E_{\psi_n}\left[ K\left(\frac{1}{h_n^*}\left(\eta_0(\bm{Z}_1)-1\right)\right)\zeta_1\right]\right|\nonumber\\
 &\le & M_{n3}(h_n^*)^p, \label{unifdagp4}
\end{eqnarray} 
where use has been made of the assumptions that $E\left[\left\|\bm{X}_1\right\|\right]<\infty$, $E\left[U_1\right]=0$, that 
$U_1$ and $\eta_0(\bm{Z}_1)$ are independent and that $K(\cdot)$ is a kernel of $p$-th order.  Since (\ref{unifdagp4}) 
holds for every $\psi_n\in\Psi_n$, we find that
\begin{eqnarray}
 & &\sup_{\psi_n\in\Psi_n}\left|E_{\psi_n}\left[\bm{e}^{\top}_1\left(\sum_{i=1}^n\bm{S}_i K_i\bm{S}_i^{\top}\right)^{-1}\sum_{i=1}^n\bm{S}_i K_i\zeta_i\right]-\frac{1}{nh_n^*} E_{\psi_n}\left[\bm{e}^{\top}_1\sum_{i=1}^n\bm{S}_i K_i\zeta_i\right]\right|\nonumber\\
 &\le & M_{n3}\left(h^*_n\right)^{p}.\label{unif2dagp4}
\end{eqnarray}
A similar calculation shows that there exists a constant $M_{n4}\in (0,\infty)$ such that
\begin{equation}\label{unif3starp4}
\sup_{\psi_n\in\Psi_n}\left|\frac{1}{nh_n^*} E_{\psi_n}\left[\bm{e}^{\top}_1\sum_{i=1}^n\bm{S}_i K_i\zeta_i\right]\right|\le M_{n4}\left(h^*_n\right)^{p}.
\end{equation}
Combine (\ref{unif2starp4}), (\ref{unif2dagp4}) and (\ref{unif3starp4}) with the assumption that $n(h_n^*)^{2p+1}=c<\infty$ to deduce that
\begin{equation}\label{lastconv}
\sqrt{nh_n^*}\sup_{\psi_n\in\Psi_n}\left|E_{\psi_n}\left[\bm{e}^{\top}_1\left(\sum_{i=1}^n\bm{S}_i K_i\bm{S}_i^{\top}\right)^{-1}\sum_{i=1}^n\bm{S}_i K_i\zeta_{ni}\right]\right|=0.
\end{equation}
The desired bound, i.e., (\ref{unifonep1}), follows from (\ref{unifstarp2b}) and (\ref{lastconv}). 

\subsubsection{Proof of (\ref{uniftwop1})}

\noindent Arguments similar to those used in the proof of (\ref{unifonep1}) based on the uniform boundedness of 
$\hat{\eta}_n(\cdot)$ over all $n$ and on the boundedness of $K(\cdot)$ and of $m^{(j)}_{F_0}(\cdot)$ 
for each $j\in\{0,1,\ldots,p\}$ show that
\begin{eqnarray}
 & &nh_n^*\sup_{\psi_n\in\Psi_n}\left|Var_{\psi_n}\left[\hat{\theta}_n^*\right]\right.\nonumber\\
 & &\left.-E_{\psi_n}\left[\bm{e}^{\top}_1\left(\sum_{i=1}^n\bm{S}_i K_i\bm{S}_i^{\top}\right)^{-1}\left(\sum_{i=1}^n\bm{S}_i K_i\zeta_i\right)\left(\sum_{i=1}^n\bm{S}^{\top}_i K_i\zeta_i\right)\right.\right.\nonumber\\
 & &\left.\left.\cdot\left(\sum_{i=1}^n\bm{S}_i K_i\bm{S}^{\top}_i\right)^{-1}\bm{e}_1\right]\right|\nonumber\\
 &=&o(1).\label{unifstarp5}
\end{eqnarray}
Next, consider that the uniform boundedness of $\hat{\eta}_n(\cdot)$ over all $n$ and the boundedness of 
$K(\cdot)$, and of $m^{(1)}_{F_0}(\cdot)$ imply that there exists constants $M_{n5}, M_{n6}\in (0,\infty)$ not depending on $\psi_n$ 
such that for every sequence $\left\{\psi_n:\,\psi_n\in\Psi_n\right\}$:
\begin{eqnarray}
 & &E_{\psi_n}\left[\bm{e}^{\top}_1\left(\sum_{i=1}^n\bm{S}_i K_i\bm{S}_i^{\top}\right)^{-1}\left(\sum_{i=1}^n\bm{S}_i K_i\zeta_i\right)\left(\sum_{i=1}^n\bm{S}_i^{\top} K_i\zeta_i\right)\right.\nonumber\\
 & &\left.\cdot\left(\sum_{i=1}^n\bm{S}_i K_i\bm{S}_i^{\top}\right)^{-1}\bm{e}_1\right]\nonumber\\
 &\le & M_{n5}\cdot\frac{1}{n(h_n^*)^2} E_{\psi_n}\left[ K^2\left(\frac{1}{h_n^*}\left(\eta_0(\bm{Z}_1)-1\right)\right)\zeta_1^2\right]\nonumber\\
 &\le & M_{n6}\cdot\frac{1}{nh_n^*},\label{unif2starp5} 
\end{eqnarray}
where use has been made of the assumptions that $E_{\psi_n}\left[U_1^2\right]<\infty$ for all $\left\{\psi_n\right\}$, and that 
$\int K^2(u) du<\infty$.

It follows from (\ref{unif2starp5}) that
\begin{eqnarray}
 & & nh_n^*\sup_{\psi_n\in\Psi_n} E_{\psi_n}\left[\bm{e}^{\top}_1\left(\sum_{i=1}^n\bm{S}_i K_i\bm{S}_i^{\top}\right)^{-1}\left(\sum_{i=1}^n\bm{S}_i K_i\zeta_i\right)\left(\sum_{i=1}^n\bm{S}_i^{\top} K_i\zeta_i\right)\right.\nonumber\\
 & &\left.\cdot\left(\sum_{i=1}^n\bm{S}_i K_i\bm{S}_i^{\top}\right)^{-1}\bm{e}_1\right]\nonumber\\
 &<& \infty.\label{unifdagp5}
\end{eqnarray}
Combining (\ref{unifstarp5}) and (\ref{unifdagp5}) enables one to deduce that $nh_n^* \sup_{\psi_n\in\Psi_n} Var_{\psi_n}\left[\hat{\theta}_n^*\right]<\infty$.
This completes the proof of (\ref{uniftwop1}).

\bibliography{chuan}
\bibliographystyle{cje}

\end{document}